\documentclass[nonacm, sigconf]{acmart}

\setcopyright{none}
\renewcommand\footnotetextcopyrightpermission[1]{}
\usepackage[utf8]{inputenc}
\usepackage{graphicx}
\graphicspath{{figures/}}
\usepackage{float}
\usepackage[ruled,vlined]{algorithm2e}
\SetKw{KwBy}{by}
\LinesNumbered

\usepackage{amsthm}
\usepackage{amsmath}

\usepackage{relsize}

\newtheorem{lemma}{Lemma}[section]
\usepackage{listings}
\usepackage{xcolor}
\definecolor{commentgreen}{RGB}{2,112,10}
\definecolor{eminence}{RGB}{108,48,130}
\definecolor{weborange}{RGB}{255,165,0}
\definecolor{frenchplum}{RGB}{129,20,83}
\definecolor{todocolor}{rgb}{0.8,0,0}

\ifdefined\notodo
\newcommand{\TODO}[1]{}
\else
\newcommand{\TODO}[1]{{\color{todocolor}#1}}
\fi

\usepackage{longtable}

\usepackage{pgfplots}
\usetikzlibrary{pgfplots.colormaps}
\pgfplotsset{compat=1.14}
\usepackage[caption=false]{subfig}
\usepackage{colortbl}
\usepackage{pgfplotstable}
\usepgfplotslibrary{statistics}

\pgfplotstableset{
    /color cells/min/.initial=0,
    /color cells/max/.initial=1000,
    /color cells/textcolor/.initial=,
    color cells/.code={%
        \pgfqkeys{/color cells}{#1}%
        \pgfkeysalso{%
            postproc cell content/.code={%
                \begingroup
                \pgfkeysgetvalue{/pgfplots/table/@preprocessed cell content}\value
                \ifx\value\empty
                    \endgroup
                \else
                \pgfmathfloatparsenumber{\value}%
                \pgfmathfloattofixed{\pgfmathresult}%
                \let\value=\pgfmathresult
                \pgfplotscolormapaccess
                    [\pgfkeysvalueof{/color cells/min}:\pgfkeysvalueof{/color cells/max}]
                    {\value}
                    {\pgfkeysvalueof{/pgfplots/colormap name}}%
                \pgfkeysgetvalue{/pgfplots/table/@cell content}\typesetvalue
                \pgfkeysgetvalue{/color cells/textcolor}\textcolorvalue
                \toks0=\expandafter{\typesetvalue}%
                \xdef\temp{%
                    \noexpand\pgfkeysalso{%
                        @cell content={%
                            \noexpand\cellcolor[rgb]{\pgfmathresult}%
                            \noexpand\definecolor{mapped color}{rgb}{\pgfmathresult}%
                            \ifx\textcolorvalue\empty
                            \else
                                \noexpand\color{\textcolorvalue}%
                            \fi
                            \the\toks0 %
                        }%
                    }%
                }%
                \endgroup
                \temp
                \fi
            }%
        }%
    }
}

\AtBeginDocument{%
  \providecommand\BibTeX{{%
    \normalfont B\kern-0.5em{\scshape i\kern-0.25em b}\kern-0.8em\TeX}}}

\begin{document}

\fancyhead{}
\fancyfoot{}

\title{Sparse Tensor Transpositions}

\author{Suzanne Mueller}
\affiliation{\institution{MIT CSAIL}}
\email{suzmue@csail.mit.edu}

\author{Peter Ahrens}
\affiliation{\institution{MIT CSAIL}}
\email{pahrens@csail.mit.edu}

\author{Stephen Chou}
\affiliation{\institution{MIT CSAIL}}
\email{s3chou@csail.mit.edu}

\author{Fredrik Kjolstad}
\affiliation{\institution{Stanford University}}
\email{kjolstad@stanford.edu}

\author{Saman Amarasinghe}
\affiliation{\institution{MIT CSAIL}}
\email{saman@csail.mit.edu}


\begin{abstract}
We present a new algorithm for transposing sparse tensors called \text{\sc Quesadilla}. The algorithm converts the sparse tensor data structure to a list of coordinates and sorts it with a fast multi-pass radix algorithm that exploits knowledge of the requested transposition and the tensors input partial coordinate ordering to provably minimize the number of parallel partial sorting passes. We evaluate both a serial and a parallel implementation of \text{\sc Quesadilla} on a set of 19 tensors from the FROSTT collection, a set of tensors taken from scientific and data analytic applications. We compare \text{\sc Quesadilla} and a generalization, \textsc{Top-2-sadilla} to several state of the art approaches, including the tensor transposition routine used in the SPLATT tensor factorization library. In serial tests, \text{\sc Quesadilla} was the best strategy for 60\% of all tensor and transposition combinations and improved over SPLATT by at least 19\% in half of the combinations. In parallel tests, at least one of \text{\sc Quesadilla} or \text{\sc{Top-2-sadilla}} was the best strategy for 52\% of all tensor and transposition combinations.
\end{abstract}



\maketitle

\section{Introduction}\label{sec:introduction}

Tensors generalize vectors and matrices to any number of dimensions. Tensors used in computation are often sparse, which means many of the values are zero. To take advantage of the large number of zeroes in the tensor, we use sparse formats that allow the zeroes to be compressed away. These formats range from a simple list of coordinates to complicated data structures such as Compressed Sparse Row (CSR) \cite{eisenstat_yale_1982}, Doubly Compressed Sparse Row (DCSR) \cite{buluc_representation_2008}, Block Compressed Sparse Row (BCSR) \cite{im_optimizing_2001}, and Compressed Sparse Fiber (CSF) \cite{smith_splatt:_2015}. These formats have a natural ordering of their dimensions that provides a lexicographical ordering of the tensor nonzeros. In a sorted list of coordinates, the order of the sorting keys determines this lexicographic ordering.

Tensor algebra is used to compute with data stored in tensors. These multidimensional computations need to access the nonzero entries in one or more tensors, compute, and store the results. Accessing the nonzero entries requires some traversal of the tensor. However, unlike for dense tensors, traversing the nonzeros of a sparse tensor in different lexicographical orderings may be asymptotically more expensive than the natural lexicographical ordering. Therefore, it is often faster to first transpose input tensors by reordering their dimensions before executing tensor expressions. This way, the tensor can be accessed naturally in the expression itself.

Tensor transposition is ubiquitous in data processing. Anytime multiple tensor expressions are composed and the output of one expression must be used as an input to the next, with a different index ordering and possibly a different sparse format, we need to transpose. For example, element-wise operations between tensors without matching index orderings (thus requiring transposition as a bottleneck) is listed as one of the five benchmark operations in the Parallel Sparse Tensor Algorithm Benchmark Suite (PASTA) \cite{li_pasta_2019}. Sparse tensor transposition may also occur when several different orderings of input data are required for efficient operation, but the space is not available to hold all of them. Such a situation might arise when using an alternating least squares method for canonical polyadic decomposition \cite{smith_splatt:_2015}. 

Prior work has focused extensively on dense tensor transpositions \cite{ruetsch_optimizing_2009, sung_dl_2012, kaushik_efficient_1993, sung_in-place_2014,catanzaro_decomposition_2014, gustavson_parallel_2012, springer_ttc_2017, karlsson_blocked_2009, vladimirov_multithreaded_2013}; we refer readers to \cite{springer_ttc_2017} for a summary. Sparse matrix and tensor transposition have received relatively little attention \cite{wang_parallel_2016}. A fast CSR sparse matrix transposition algorithm is proposed in \cite{gustavson_two_1978}, and improvements are proposed in \cite{wang_parallel_2016,gonzalez-mesa_parallelizing_2013}. Further
variations on sparse matrix transposition are discussed in \cite{weng_parallel_2013, weng_designing_2013, guo_designing_2016, vazquez_new_2011, cameron_two_1993}. None of these techniques, however, readily generalize to sparse tensor transposition with input tensors of arbitrary ranks.

Tensors are often stored as a list of coordinates of nonzeros. If the coordinates are ordered lexicographically, adjacent coordinates may share the same indices in the first several modes. The Compressed Sparse Fiber (CSF) format \cite{smith_splatt:_2015} compresses these duplicate nonzeros using a tree-like storage format. In CSF, nodes represent indices, leaves represent nonzeros and paths from root to leaf represent coordinates. The children of each node are ordered. The matrix case of CSF is called Compressed Sparse Row (CSR).

More complicated sparse tensor formats like CSF often have similar ordering constraints, and require access to the coordinates in some lexicographic order in order to construct the tensor. The current state of the art for transposing sparse tensors involves converting the sparse tensor into a list of coordinates, sorting the list of coordinates, and finally packing the list of coordinates into the desired sparse tensor format \cite{smith_splatt:_2015}.

This approach reduces the problem of transposing a tensor into a problem of sorting a list of coordinates. However, the lists of coordinates have partial orderings we can use to accelerate the sorting algorithms. Consider the example matrix in Figure \ref{fig:matrix-example}. In order to transpose the matrix, the column coordinates must be ordered lexicographically before the row coordinates. This could be accomplished by sorting with the column coordinate as the primary key and the row coordinate as the secondary key. 

We can do better than that. The coordinates are already sorted on the row coordinates. By doing a stable sort on just the column coordinate, we get the same result. In this paper, we will generalize this optimization to arbitrary tensor transpositions.

\begin{figure}[htbp!]
\centering
\includegraphics[width=0.9\linewidth]{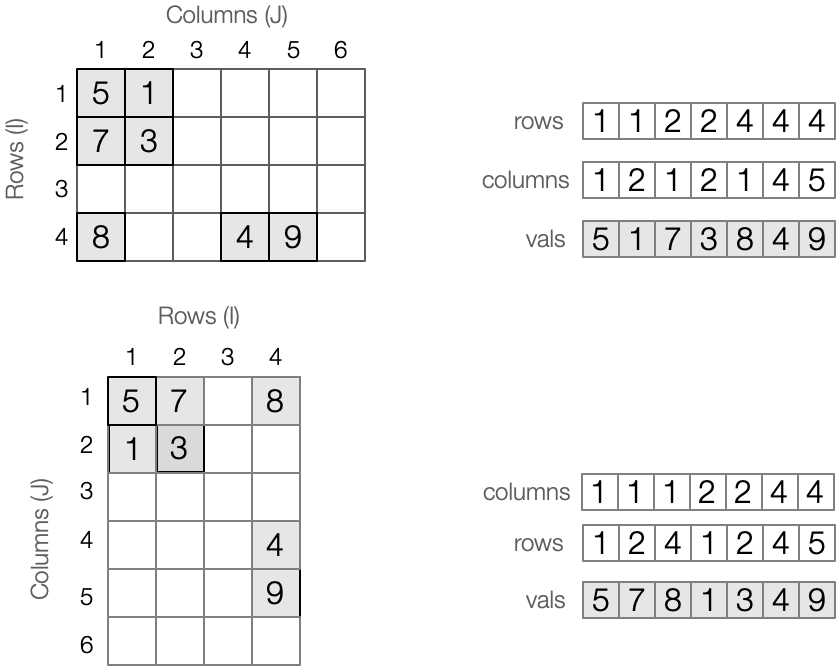}
\caption[Lexicographical ordering of rows and columns in sorted coordinate list]{The matrix $A$ can be represented as a list of coordinates including only the nonzero values. Transposing the tensor in this format switches the lexicographic ordering of the rows and columns, such that the columns appear first. The top list of coordinates represent the matrix in the top left. The bottom list of coordinates represent the transposed version of this matrix.}
\label{fig:matrix-example}
\end{figure}

The main contributions of this work are:
\begin{enumerate}
	\item A decomposition of tensor transposition into parallelizeable near-linear-work partial sorts (one of the two partial sorts is novel) that optionally respect previous partial orderings.
	\item An algorithm that uses partial orderings in the original sparse tensor format to minimize the number of partial sorts required by the transposition algorithm. This relates the parallel span of radix sorting to partial orderings in the input.
	\item A parallel implementation that demonstrates this transposition algorithm is competitive with, and often faster than, state of the art approaches.
\end{enumerate}

\section{Background}\label{sec:background}

A tensor of \textbf{rank} $r$ is a multidimensional array that associates
$r$-tuples (referred to as \textbf{coordinates}) with values, or
\textbf{entries}. We refer to the $k^{th}$ position in a coordinate as mode
$k$. The size of a tensor is specified by an $r$-tuple of \textbf{dimensions}
$n$, where each index $i_k$ is an integer in the range $1 \leq i_k \leq n_k$.

Let $N$ be the number of nonzero entries in our tensor. A tensor is
\textbf{sparse} if most of its entries are zero. This has led to the
development of sparse tensor storage formats that support efficient
computation over only the nonzero entries. These formats range from a simple
sorted list of nonzero coordinates together with their values, the
\textbf{Coordinates} (COO) format \cite{bader_efficient_2007}, to more complicated hierarchical
mode-by-mode compression schemes such as \textbf{Compressed Sparse Row} (CSR)
or \textbf{Compressed Sparse Fiber} (CSF) \cite{eisenstat_yale_1982,
smith_splatt:_2015}. All three of these formats induce a natural lexicographic ordering
of the dimensions; iterating over the tensor in the natural order can be done
very efficiently.

We define lexicographic ordering on $r$-tuples recursively using a tuple
$\sigma$ of modes in order of their priority. We consider the coordinate $i =
(i_1, i_2, ...)$ to be less than the coordinate $i' = (i'_1, i'_2, ...)$
under the ordering $\sigma$ in two cases. The first case is when
$i_{\sigma_1} < i'_{\sigma_1}$. The second case is when both $i_{\sigma_1} =
i'_{\sigma_1}$ and $i < i'$ under the ordering $(\sigma_2, \sigma_3, ...)$.
For completeness, we say that all tuples are considered equal under $\sigma =
()$. We will refer to the $(1, 2, ..., r)$ ordering of $k$-tuples as the
\textbf{simple} ordering. We say an ordering is \textbf{complete} if it
contains $r$ distinct modes.

\subsection{Coordinates (COO)}

COO stores the nonzero coordinates in the tensor as a list of $\sigma$-sorted
coordinates. Transposing a tensor in COO format is equivalent to reordering
the coordinate list to a new complete ordering. This simplicity makes COO a
popular format; it is the only sparse tensor format for the MATLAB Tensor
Toolbox and TensorFlow libraries, and is used as an intermediate format
during transpositions in the SPLATT library \cite{bader_efficient_2007,
abadi_tensorflow_2016, smith_splatt:_2015}. Since most sparse tensor formats
can be converted to and from COO format and the format is readily sorted, we
focus on transposing tensors in COO format.

The COO format can be implemented either with a list of lists (one for each
mode) or as a list of coordinate tuples. We will use the latter for
notational purposes. Thus, we store an array $A$ in COO using two arrays,
$A.crd$ and $A.val$. The $crd$ array is an array of coordinates, and $val$ is
an array of corresponding values. This requires $O(r)$ bits to store each
coordinate, so the total storage requirement for indices is $O(r
* N)$ bits. Figure \ref{fig:coo} shows an example of COO storage.

\begin{figure}[htbp!]
\centering
\includegraphics[width=0.9\linewidth]{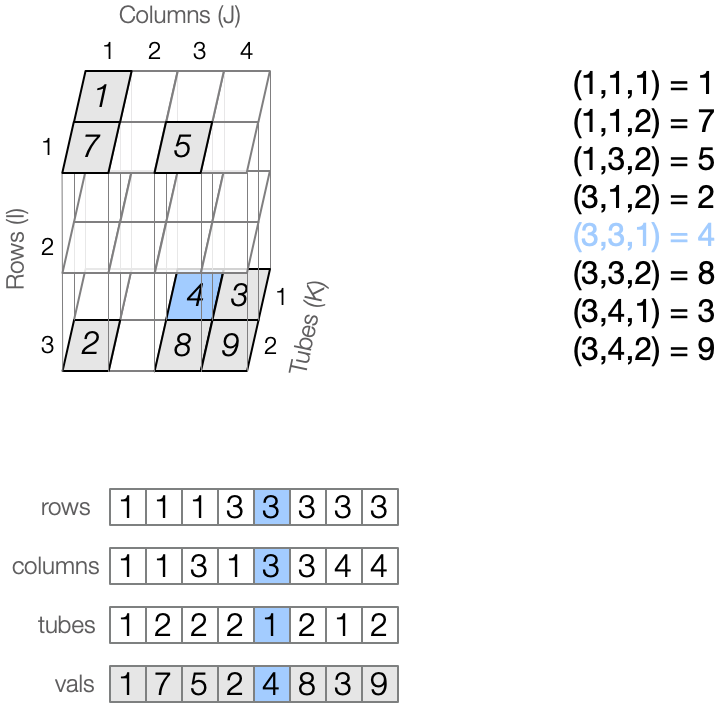}
\caption[COO format example]{COO represents a tensor as a list of coordinates with all of the zero values compressed out. The list of arrays in the bottom left represent the tensor in the top left.}
\label{fig:coo}
\end{figure}


\subsection{Transposition}
\label{transposition}

The COO, CSR, and CSF formats are all sorted by lexicographic orderings on
the coordinates. In COO, we sort the coordinates by some $\sigma$ ordering,
and in CSR and CSF, $\sigma$ gives the order in which modes correspond to
levels in the tree. CSR and CSF formats sort each level by the indices within
that mode. In this work, \textbf{transposition} corresponds to a change in
this storage order.

We will express tensor transposition operations using the final desired
storage order $\sigma$. Without loss of generality, we assume the tensor is
initially ordered by $(1, 2, \ldots)$. For example, transposing a matrix
stored in $(1, 2)$ order is equivalent to changing the storage order to
$\sigma = (2,1)$, then relabeling the modes.

Certain computations will perform better when the dimensions can be iterated
over efficiently in a different order than the initial storage order \cite{smith_splatt:_2015, kjolstad_sparse_2018}. When we
encounter such a computation, it will be beneficial to transpose the tensor.

Tensors can be reordered using any sorting method. Since the coordinates come
from a fixed range of values, we can also use sorts that work on fixed length
keys, like histogram or radix sorts \cite{cormen_introduction_2009}. SPLATT,
a sparse tensor library designed to be highly parallel, uses a specialized
sorting strategy to take advantage of potential parallelism that exists in
the problem \cite{smith_splatt:_2015}. SPLATT chooses to first do a histogram
sort on mode $\sigma_{1}$. It then sorts the coordinates for each index in
mode $i_{\sigma_{1}}$ using $n_{\sigma_1}$ separate calls to quicksort. In
the sequential implementation this strategy benefits from smaller subproblems
for quicksort. In the parallel version, SPLATT is able to sort these buckets
in parallel.

\section{Algorithms}\label{sec:algorithm}
A naive algorithm for sparse tensor transposition is to comparison sort the
coordinates into the desired lexicographic order. However, since coordinates
already have an initial ordering, we can think of sorting coordinates as
simply changing the lexicographic ordering to prioritize different
dimensions. It takes $O(r)$ time to compare two coordinates of an $r$-tensor.
Thus, a comparison based coordinate sort would run in $O(rN\log N)$ time.
However, since the indices are bounded by the dimensions, we can use
parallelizable stable sorts like a histogram sort (called
counting sort in \cite{cormen_introduction_2009, obeya_theoretically-efficient_2019}) to sort the
coordinates on a single mode $k$ in $O(rN + n_k)$ time. If we perform $r$
histogram sorts (a radix sort on $r$-digit numbers), we can sort our
coordinate list in $O(r^2N + n_1 + n_2 + \cdots)$ time, an asymptotic improvement
over comparison sort when the dimensions are small. If we assume that
coordinates are each processed in constant time, our histogram sort takes
$O(N + n_k)$ time and our radix sort takes $O(rN + n_1 + n_2 + \cdots)$ time.
This algorithm can be improved further; for some transpositions, we do not
need to perform all $r$ sorts. For example, \textsc{HALFPERM} uses a
histogram sort to prioritize the second dimension in the new ordering.
Depending on the size of the second dimension, this single histogram sort is
faster than a generic sort of the coordinates, and certainly faster than
redundantly sorting the first dimension before sorting the second.

In this section, we formalize and generalize this idea to produce the Quesadilla tensor transposition algorithm, which provably performs the minimal number of histogram sorts.
We start with a description of our histogram partial sort and a bucketing
modification to produce two sorting primitives. We then use these primitives to build the Quesadilla and Top-$K$-sadilla tensor transposition algorithms.

%
%

%
%

\subsection{Histogram Partial Sorts}

A histogram sort sorts integer keys of bounded size. It first counts the
number of occurrences of each key values. It then performs a prefix sum, also
known as a cumulative sum or prefix scan, over the array of counts to
determine where each group of equivalent keys will lie in the output array.
This reserves enough space for all of the coordinates to appear in the output
order, and the scanned array can be used record how full each output group is
as the algorithm puts each coordinate directly into its output location.
There is extensive research on the topic of parallelizing histogram sort as a
subroutine of radix sort \cite{obeya_theoretically-efficient_2019}. In our
experiments, we use the same implementation as SPLATT
\cite{smith_splatt:_2015}, where each processor uses a private copy of
$count$ which is synchronized before moving coordinates to their output
destinations.

\begin{algorithm}[htbp!]
    \caption{$\text{\sc PartialSort}(A, (), \tau_k)$ (Non-Bucketed)}\label{alg:histogram}
    \SetAlgoLined
    \DontPrintSemicolon
    \KwIn{$A$ is a rank-$r$ tensor of dimension $n$ with $N$ nonzeros stored
    in COO, sorted under the ordering $\tau$. Our goal is to sort $A$ on
    $p = \tau_k$.}
    \KwOut{$A = B$, a tensor in COO format sorted under the ordering
    \begin{equation}\label{eq:histogram_order}
    (\tau_k, \tau_1, \ldots, \tau_{k - 1}, \tau_{k + 1}, \ldots, \tau_r).
    \end{equation}
    }
    $count \gets \text{a length $n_{\tau_k} + 1$ array of integers initialized to 0}$\;
    $count[1] \gets 1$\;

    \tcp{Compute the $count$ array}
    \For{$j \gets 1$ \KwTo $N$}{
      	$i = A.crd[j]$\;
        $count[i_{\tau_k} + 1] \gets count[i_{\tau_k} + 1] + 1$\;
    }
    \tcp{Prefix sum}
    \For{$i_{\tau_k} \gets 2$ \KwTo $n_{\tau_k} + 1$}{
        $count[i_{\tau_k}] \gets count[i_{\tau_k}] + count[i_{\tau_k}-1]$\;
    }
    \tcp{Move coordinates to final output destination}
    \For{$j \gets 1$ \KwTo $N$}{
        $i = A.crd[j]$\;
        $j' = count[i_{\tau_k}]$\;
        $B.crd[j'] = i$\;
        $count[i_{\tau_k}] \gets count[i_{\tau_k}] + 1$\;
    }
\end{algorithm}

The histogram sort iterates over the coordinates twice and the count array once. 
The total runtime is $O(rN + n_{\tau_k})$, where $n_{\tau_k}$ is the dimension of the mode being sorted on.
If we can process coordinates in constant time, the runtime is $O(N + n_{\tau_k})$. 
The histogram sort clearly produces a lexicographic ordering which prioritizes $\tau_k$ first.
Since the sort is stable, the relative ordering of other modes is unaffected.
Thus, it moves the mode $\tau_k$ to be the first mode in the lexicographic order, as described in \eqref{eq:histogram_order}.

%
%
%
%
%

\subsection{Bucketed Histogram Sort}

Although radix sort is most commonly performed from the least significant
digit to the most significant digit, it will be useful for us to be able to
work backwards sometimes, sorting one mode while respecting another partial
ordering. Informally, we wish to sort a mode to a different position than the
first spot in the output ordering. Formally, if our tensor is sorted with
respect to $\tau$, we wish to sort on $\tau_k$ while leaving the ordering
$(\tau_1, \ldots, \tau_l)$ of the first $l < k$ modes unaffected. This means
that we need to sort each group of contiguous coordinates (a \textbf{bucket})
which agree on the values $i_{\tau_1}, \ldots, i_{\tau_l}$. If we use a
comparison-based sort within each bucket, we can perform the sort recursively
but incur a logarithmic overhead. If we use a radix-based algorithm within
each bucket, we need to perform an $O(n_{\tau_k})$ prefix sum in each bucket.
Since the number of buckets is bounded only by $N$, the resulting runtime of
$O(rNn_{\tau_k})$ is unacceptable.

Here, we describe a variation on histogram sort that discovers the buckets
for $(\tau_1, \ldots, \tau_l)$, sorts on $\tau_k$, then reimposes the previous
ordering. Since there are at most $N$ buckets, our algorithm runs in time
$O(N * (l + 1) + n_{\tau_k})$. If we assume operations on coordinates occur in constant time, our algorithm runs in time $O(N + n_{\tau_k})$. Note that the input must be sorted under $(\tau_1, \ldots,
\tau_l)$ to discover the buckets by examining adjacent coordinates.

Algorithms \ref{alg:histogram} and \ref{alg:bucketed_histogram} are both called $\text{\sc PartialSort}$; we use Algorithm \ref{alg:histogram} when $l = 0$ and Algorithm
\ref{alg:bucketed_histogram} otherwise.

\begin{algorithm}[htbp!]
    \caption{$\text{\sc PartialSort}(A, (\tau_1, \ldots, \tau_l), \tau_k)$  (Bucketed)}\label{alg:bucketed_histogram}
    \SetAlgoLined
    \DontPrintSemicolon
    \KwIn{$A$ is a rank-$r$ tensor of dimension $n$ with $N$ nonzeros stored
    in COO, sorted under the ordering $\tau$. Our goal is to sort $A$ on
    $\tau_k$ while maintaining the ordering $(\tau_1, \ldots, \tau_l)$. We require that $l < k$.}
    \KwOut{$A = B$, a tensor in COO format sorted under the ordering
    \begin{equation}\label{eq:bucketed_histogram_order}
    (\tau_1, \ldots, \tau_l, \tau_k, \tau_{l + 1}, \ldots, \tau_{k - 1}, \tau_{k + 1}, \ldots, \tau_r).
    \end{equation}
    }
    $count \gets \text{a length $n_{\tau_k} + 1$ array of integers initialized to 0}$\;
    $count[1] \gets 1$\;
    $bucket \gets \text{an uninitialized integer array of size $N$}$\;
    $bucket[1] \gets 1$\;
    $pos \gets \text{an uninitialized integer array of size $N$}$\;
    $pos[1] \gets 1$\;
    $perm \gets \text{an uninitialized integer array of size $N$}$\;
    $n' \gets 1$\; 
    \tcp{Compute the $count$ array, equivalence classes under $(\tau_1, \ldots, \tau_l)$, and the $pos$ array for each of those classes}
    $i \gets A.crd[1]$\;
    $count[i_{\tau_k} + 1] \gets count[i_{\tau_k} + 1] + 1$\;
    \For{$j \gets 2$ \KwTo $N$}{
        $i \gets A.crd[j]$\;
        $i' \gets A.crd[j - 1]$\;
        \If{$(i_{\tau_1}, \ldots, i_{\tau_l}) \neq (i'_{\tau_1}, \ldots, i'_{\tau_l})$ }{
            $n' \gets n' + 1$\;
            $pos[n'] \gets j$\;
        }
        $bucket[j] \gets n'$\;
        $count[i_{\tau_k} + 1] \gets count[i_{\tau_k} + 1] + 1$\;
    }
    \tcp{Prefix sum}
    \For{$i_{\tau_k} \gets 2$ \KwTo $n_{\tau_k} + 1$}{
        $count[i_{\tau_k}] \gets count[i_{\tau_k}] + count[i_{\tau_k} - 1]$\;
    }
    \tcp{Create permutation of $A$ ordered on $(\tau_k,)$}
    \For{$j \gets 1$ \KwTo $N$}{\label{alg:bucketed_histogram:presort}
        $i \gets A.crd[j]$\;
        $perm[count[i_{\tau_k}]] \gets j$\;
        $count[i_{\tau_k}] \gets count[i_{\tau_k}] + 1$\;
    }
    \tcp{Reintroduce the stored ordering on $(\tau_1, \ldots, \tau_l)$}
    \For{$j\gets 1$ \KwTo $N$}{\label{alg:bucketed_histogram:sort}
        $B.crd[pos[bucket[perm[j]]]] \gets A.crd[perm[j]]$\;
        $pos[bucket[perm[j]]] \gets pos[bucket[perm[j]]] + 1$\;
    }
\end{algorithm}

Although we can save buckets as we fill the $count$ array, Algorithm
\ref{alg:bucketed_histogram} performs an extra bucketing step to create the
perm array, and the perm array introduces more indirection in the final
bucketing step than the similar loop in Algorithm \ref{alg:histogram}. Saving
the buckets takes $O(lN)$ time, the prefix sum takes $O(n_{\tau_k})$ time,
and the last two bucketing steps take $O(rN)$ time. The total runtime of
bucketed histogram sort is $O(rN + n_{\tau_k})$, or $O(N + n_{\tau_k})$ if we
assume constant-time operations on coordinates.

Bucketed histogram sort works by first stably sorting on mode $\tau_k$, then
by sorting on $(\tau_1, \ldots, \tau_l)$ using the bucket array. After the
loop on line \ref{alg:bucketed_histogram:presort}, $perm$ sorts $A$ under
\eqref{eq:bucketed_histogram_order}
\[
    (\tau_k, \tau_1, \ldots, \tau_{k - 1}, \tau_{k + 1}, \ldots, \tau_r).
\]
Since the stored buckets correspond to equivalence classes of $(\tau_1, \ldots,
\tau_l)$ in order, the loop on line \ref{alg:bucketed_histogram:sort} sorts
$A$ stably on the buckets, reprioritizing $(\tau_1, \ldots, \tau_l)$ in the
ordering to produce the final order
\[
    (\tau_1, \ldots, \tau_l, \tau_k, \tau_{l + 1}, \ldots, \tau_{k - 1}, \tau_{k +
    1}, \ldots, \tau_r).
\]

As we describe parallelization strategies, we focus our attention on these
three steps. Algorithm \ref{alg:bucketed_histogram} discovers the buckets,
stably sorts on the desired mode, then stably sorts on the buckets.
Discovering the buckets is a simple linear-time algorithm that we can easily
parallelize, taking care to account for buckets that cross processor
boundaries. Most parallel implementations of histogram sort, including
SPLATT, create private copies of the $count$ array \cite{smith_splatt:_2015}.
On $P$ processors, these implementations run in $O(N/P + n_{\tau_k})$ time.
Since we can usually assume the dimension of the mode to be sorted is small
relative to the number of nonzeros, these parallel implementations of
histogram sort are acceptable for sorting the desired mode. However, we
cannot assume that the number of buckets is small relative to the number of
nonzeros. To effectively parallelize the second sort, we would need to use an
algorithm whose runtime is linear in both the number of nonzeros and the
range of keys to be sorted, such as a sample sort
\cite{blelloch_comparison_1991, zhang_novel_2012}. Notice that the sampling
step can be avoided because the bucket discovery step calculates the exact
distribution of buckets (keys).

We can simplify parallelization of Algorithm \ref{alg:bucketed_histogram} by
decomposing the problem along bucket boundaries. The buckets limit the travel
of coordinates between input and output orderings; coordinates do not escape
their buckets. Therefore, running Algorithm \ref{alg:bucketed_histogram} on a
contiguous region of input buckets will compute the corresponding region of
the output ordering. This gives our chosen parallel algorithm where we assign to
each processor the buckets which begin in their region, and each processor
simply runs Algorithm \ref{alg:bucketed_histogram} locally on their section.
Assuming that the buckets are small enough to permit effective decomposition,
this algorithm also runs in time $O(N/P + n_{\tau_k})$. Notice that because
SPLATT decomposes the local sorts along the index $\sigma_1$, SPLATT operates
under the similar assumption that slices of the tensor are small enough to
permit effective decomposition. 

\subsection{Bucketed Histogram Sort Example}

We give an example of our bucketed histogram sort on a 4-tensor. For
simplicity of presentation, we represent our coordinate list as 4-digit
integers. The integers are initially sorted under the ordering $(1, 2, 3,
4)$.
\[
A.crd = [1218, 1224, 1274, 1421, 1437, 1456, 1472, 3216, 3283, 3286]
\]
Suppose that we would like them to be sorted under the ordering $\sigma = (1, 2, 4,
3)$. We rearrange our digits to show the current ordering.
\[
A.crd_{\sigma} = [1281, 1242, 1247, 1412, 1473, 1465, 1427, 3261, 3238, 3268]
\]
Since our ordering doesn't change the first two digits, we can reorder $A$ to
be sorted under $\sigma$ by bucketing on the first two digits. Our algorithm
starts by discovering the buckets and computing the $bucket$ and $pos$
arrays, which store the numbers and positions of each bucket:
\begin{align*}
A.crd_{\sigma} &= [\underbrace{1281, 1242, 1247}_{\text{``12...''}}, \underbrace{1412, 1473, 1465, 1427}_{\text{``14...''}}, \underbrace{3261, 3238, 3268}_{\text{``32...''}}]\\
bucket &= [1, 1, 1, 2, 2, 2, 2, 3, 3, 3]\\
pos &= [1, 4, 8, undefined, ...]
\end{align*}
Our counting sort sorts $A$ by digit $\sigma_3 = 4$, producing:
\begin{align*}
    perm &= [4, 7, 9, 2, 3, 6, 8, 10, 5, 1]\\
    A.crd[perm] &= \\
    [1421 & , 1472, 3283, 1224, 1274, 1456, 3216, 3286, 1437, 1218]\\
    A.crd_{\sigma}[perm] &= \\
    [1412 & , 1427, 3238, 1242, 1247, 1465, 3261, 3268, 1473, 1281]
\end{align*}
At this point, if we restrict our attention to one bucket at a time, the
coordinates are sorted. We just need to put each element of $A[perm]$ back
into it's corresponding bucket by sorting on $bucket[perm]$. The $pos$ array
functions as the $count$ array does in counting sort.
\begin{align*}
    B.crd &= [1224, 1274, 1218, 1421, 1472, 1456, 1437, 3283, 3216, 3286]\\
    B.crd_{\sigma} &= [\underbrace{1242, 1247, 1281}_{\text{``12...''}}, \underbrace{1412, 1427, 1465, 1473}_{\text{``14...''}}, \underbrace{3238, 3261, 3268}_{\text{``32...''}}]
\end{align*}
Notice that $B.crd_{\sigma}$ is lexicographically ordered, as desired.

\subsection{Minimizing Partial Sorts}

Transposition via a full radix sort would consist of $r$ calls to Algorithm
\ref{alg:histogram}. Not all transpositions, however, are equally difficult.
For example, if we have a simply ordered 4-tensor and are asked to transpose
it to the ordering $(4, 1, 2, 3)$, this can be accomplished with the single
call $\text{\sc PartialSort}(A, (), 4)$, as seen in \eqref{eq:histogram_order}.
On the other hand, if we are asked to transpose to $(4, 3, 2, 1)$, we
show that this requires at least 3 calls $\text{\sc PartialSort}$, since the only
relevant partial ordering we can use is that of the first mode. In this work,
we generalize this insight to produce the $\text{\sc Quesadilla}$ algorithm
which transposes tensors to a given target ordering with the minimal number
of calls to either Algorithm \ref{alg:histogram} or
\ref{alg:bucketed_histogram}. 
 
Although Algorithms \ref{alg:histogram} and \ref{alg:bucketed_histogram}
perform similar tasks, Algorithm \ref{alg:bucketed_histogram} streams through
and randomly accesses more vectors than Algorithm \ref{alg:histogram} does.
If we count the number of unique vectors in each loop body separately
(including initialization), Algorithm \ref{alg:histogram} streams through 4
vectors and randomly accesses 4 vectors, while Algorithm
\ref{alg:bucketed_histogram} streams through 7 vectors, and randomly accesses
7 vectors. While the costs of these algorithms are similar, they are not
identical, and we should prefer to avoid the bucketed histogram variant
whenever possible. For example, we can transpose to $(2, 4, 1, 3)$ by calling
$\text{\sc PartialSort}(A, (), 2)$ and then $\text{\sc PartialSort}(A, (2),
4)$, but we can avoid a bucketed histogram sort by calling $\text{\sc
PartialSort}(A, (), 4)$ and then $\text{\sc PartialSort}(A, (), 2)$. Among
transpositions that use the minimum number of partial sorts, we show that
$\text{\sc Quesadilla}$ uses the minimal number of bucketed partial sorts
(Algorithm \ref{alg:bucketed_histogram}). Thus, our algorithm minimizes a
cost model that weighs each pass equally, but breaks ties towards the
non-bucketed variant.

We start by showing that for a given target order $\sigma$, we must sort on a
certain set of modes and that in order to achieve the minimum number of
sorts, some of these sorts must be bucketed. We then give an algorithm that
only sorts on this necessary set of modes, and only uses bucketed sorts when
required.

\subsubsection{Necessary Sorts}

The number of dimensions $r$ is an upper bound on the number of passes needed
to sort coordinates. This is the number of passes that are needed if we have
a completely unsorted coordinate list and do a standard radix sort. The
histogram sort and bucketed histogram sort can only move dimensions to the
beginning of the lexicographic ordering. We use this fact to show a lower
bound on the number of passes needed to sort the coordinates into the new
lexicographic ordering. In several proofs, we will use a
function $f(\tau, p)$ that we define on complete $r$-orderings $\tau$ as
the set $\{\tau_{k + 1}, \ldots, \tau_r\}$ where $\tau_k = p$. Thus, $f(\tau,
p)$ is the set of modes which follow $p$ in the ordering $\tau$. For example,
$f((1, 3, 2, 4), 3) = \{2, 4\}$.

\begin{lemma}\label{lem:followsort}
    Let $A$ be a list of $r$-coordinates ordered by the complete ordering
    $\tau$. Assume that $A'$ is the $\tau'$ ordered result of calling
    $\text{\sc PartialSort}(A, (\tau_1, \ldots, \tau_l), p)$ where $p =
    \tau_k$ and $k > l$. If $q \neq p$, then $f(\tau', q) \subseteq
    f(\tau, q)$.
\end{lemma}

\begin{proof}
The result follows from a close examination of \eqref{eq:histogram_order} and
\eqref{eq:bucketed_histogram_order} which describe the output ordering of
Algorithms \ref{alg:histogram} and \ref{alg:bucketed_histogram}. Let $h$ be
such that $\tau_h = q$. Note that $k > l$. If $1 \leq h \leq l$ or $k < h
\leq r$, then $f(\tau', q) = f(\tau, q)$. Otherwise, $l < h < k$ and
$f(\tau', q) = f(\tau, q) \setminus p \subset f(\tau, q)$.
\end{proof}

This idea that the set following some mode never expands when we sort on a
different mode allows us to show that certain modes must be direct arguments
to $\text{\sc PartialSort}$ at some point in our sequence of calls that
transposes the tensor.

\begin{lemma}\label{lem:mustsort}
Let $A$ be a list of $r$-coordinates ordered by the complete ordering $\tau$.
Assume we wish to call $\text{\sc PartialSort}$ some number of times to produce a
$\sigma$ ordering of $A$, and that $f(\sigma, \sigma_i) \not\subseteq f(\tau,
\sigma_i)$. Consider any sequence of statements of the form \[
    A \gets \text{\sc PartialSort}(A, (\psi_1, \ldots, \psi_l), \psi_k),
\] where $\psi$ is a complete intermediate ordering of $A$, $\psi_k \neq
\sigma_i$, and $k > l$. No such sequence will result in a $\sigma$ ordering
of $A$.

Therefore, any sequence of calls to $\text{\sc Sort}$ designed to return a
$\sigma$ ordering of $A$ must include a call for each value of $\sigma_i$ for
which $f(\sigma, \sigma_i) \not \subseteq f(\tau, \sigma_i)$.
\end{lemma}

\begin{proof}
At some point in our sequence of calls, assume that $f(\sigma, \sigma_i)
\not\subseteq f(\psi, \sigma_i)$, and let $\psi'$ be the ordering after the
next call to $\text{\sc PartialSort}$. Lemma \ref{lem:followsort} implies
that $f(\psi', \sigma_i) \subseteq f(\psi, \sigma_i)$, so it must still be
the case that $f(\sigma, \sigma_i) \not\subseteq f(\psi', \sigma_i)$. Since
we start with $f(\sigma, \sigma_i) \not\subseteq f(\tau, \sigma_i)$, there is
no ordering in our sequence for which $f(\sigma, \sigma_i) \not\subseteq
f(\psi, \sigma_i)$, and thus $\psi$ can never equal $\sigma$.
\end{proof}

Lemma \ref{lem:mustsort} implies a lower bound on the number of calls to
$\text{\sc PartialSort}$ required to transpose a $\tau$-ordered tensor $A$ to
$\sigma$-order. We refer to this number with the function $b(\tau, \sigma)$.
We define $b$ formally as the number of modes $i$ for which $f(\sigma,
\sigma_i) \not\subseteq f(\tau, \sigma_i)$. For example, $b((1, 3, 2), (3, 1,
2)) = 1$. While $b$ gives us a lower bound on the number of sorts, it does
not show a bound on whether each sort must be bucketed or not. We now show
that no sequence of calls to $\text{\sc PartialSort}$ of length $b(\tau,
\sigma)$ may include a call $\text{\sc PartialSort}(A, (), p)$ if there
exists $i$ such that $f(\sigma, \sigma_i) \subseteq f(\tau, \sigma_i)$ and $p
\in f(\sigma, \sigma_i)$.

\begin{lemma}\label{lem:mustbucket}
Let $A$ be a list of $r$-coordinates ordered by the complete ordering $\tau$.
Consider any length $b(\tau, \sigma)$ sequence of statements of the form \[
    A \gets \text{\sc PartialSort}(A, (\psi_1, \ldots, \psi_l), \psi_k),
\] where $\psi$ is a complete intermediate ordering of $A$ and $k > l$. If
this sequence reaches the ordering $\sigma$, it may not contain any call
where $l = 0$ and there exists $i$ such that $f(\sigma, \sigma_i) \subseteq
f(\tau, \sigma_i)$ and $\psi_k \in f(\sigma, \sigma_i)$.
\end{lemma}
  
\begin{proof}
Lemma \ref{lem:mustsort} implies that we must sort on each mode where 
$f(\sigma, \sigma_k) \not\subseteq f(\tau, \sigma_k)$. Since our sequence
only involves $b(\tau, \sigma)$ calls to $\text{\sc PartialSort}$, this sequence
must only sort on these modes.

Assume for contradiction that our sequence involves a call where $l = 0$ and
there exists $i$ such that $f(\sigma, \sigma_i) \subseteq f(\tau, \sigma_i)$
and $\psi_k \in f(\sigma, \sigma_i)$. Let $\psi'$ be the ordering following
$\psi$ after this call. Using \eqref{eq:histogram_order}, we see that
$f(\sigma, \sigma_i) \not\subseteq f(\psi', \sigma_i)$ since $\psi_k \in f(\sigma,
\sigma_i)$ but $f(\psi', \sigma_i) = f(\psi, \sigma_i) \setminus \psi_k$. Thus,
Lemma \ref{lem:mustsort} implies that we must sort on $\sigma_i$ in order
to achieve $\sigma$ order, a contradiction as we are given that $f(\sigma,
\sigma_i) \subseteq f(\tau, \sigma_i)$.
\end{proof}

Lemma \ref{lem:mustbucket} implies that if $f(\sigma, \sigma_i) \subseteq
f(\tau, \sigma_i)$, a minimal sequence of sorts cannot involve non-bucketed
sorts on modes in $f(\sigma, \sigma_i)$.

\subsection{Quesadilla Sort}

We now present the $\text{\sc Quesadilla}$ algorithm for tensor transposition.

\begin{algorithm}
    \DontPrintSemicolon
    \caption{$\text{\sc QuesadillaSort}(A, \sigma)$}\label{algo:quesadilla}
    \KwIn{$A$ is any simply ordered list of $r$-coordinates, $\sigma$ is an
    $r$-complete ordering.}
    \KwOut{$A$, sorted in $\sigma$ order.}
    $l \gets 0$\;
    \While{$l < r$} {\label{algo:quesadilla:mainloop}
        $k \gets l$\;
        \While{$k + 1 < r$ and $f(\sigma, \sigma_{k + 1}) \not\subseteq f((1, 2, \ldots), \sigma_{k + 1})$} {\label{algo:quesadilla:scanloop}
            $k \gets k + 1$\;
        }
        $l' \gets k + 1$\;
        \While{$k > l$}{\label{algo:quesadilla:sortloop} 
            $A \gets \text{\sc PartialSort}(A, (\sigma_1, \ldots, \sigma_{l}), \sigma_k)$\;
            $k \gets k - 1$\;
        }
        $l \gets l'$\;
    }
    \Return{$A$}\;
\end{algorithm}

\begin{theorem}\label{thm:quesadilla}
Let $A$ be a simply ordered list of $r$-coordinates. The
sequence of sorts described by $\text{\sc Quesadilla}(\sigma)$ will result in
the $\sigma$ ordered list of coordinates in $A$.
\end{theorem}

\begin{proof}
We prove the result by showing that before and after each execution of the
body of the loop on line \ref{algo:quesadilla:sortloop}, $A$ is sorted under
a complete ordering $\tau$, where
\begin{equation} \label{eq:quesadillahead}
(\tau_1, \ldots, \tau_{l + l' - (k + 1)}) = (\sigma_1, \ldots, \sigma_l, \sigma_{k + 1}, \ldots, \sigma_{l' - 1}),
\end{equation}
and the remaining modes of $\tau$ are in ascending order.

Before the first execution of our loop body, $A$ is simply ordered, $l = 0$,
and $l' = k + 1$. Thus, our claim is initially satisfied.

Assume our claim holds before some execution of the loop body. Let $A'$ and
$k'$ be the values of $A$ and $k$ after executing the loop body. Let $t$ be
the mode such that $\sigma_k = \tau_t$. Since $k > l$,
\eqref{eq:quesadillahead} implies that $t > l
+ l' - (k + 1)$. Combining this observation with \eqref{eq:histogram_order}
and \eqref{eq:bucketed_histogram_order} leads to the observation that $A'$ is
sorted under the complete ordering $\tau'$, where
\[
\tau' = (\sigma_1, \ldots, \sigma_l, \sigma_k, \ldots, \sigma_{l' - 1}, \tau_{l + l' - (k + 1)}, \ldots, \tau_{t - 1}, \tau_{t + 1}, \ldots, \tau_r).
\]
Therefore, \eqref{eq:quesadillahead} still holds
for $\tau'$ and $k'$. Because $(\tau_{l + l' - k}, \ldots, \tau_r)$ was
ascending, $(\tau'_{l + l' - k'}, \ldots, \tau'_r)$ is also ascending. Thus, the
claim holds after the execution of the loop body on line \ref{algo:quesadilla:sortloop}.

All that remains to be shown is that our claim holds after we move through the
loop on line \ref{algo:quesadilla:mainloop}. After leaving the line
\ref{algo:quesadilla:sortloop} loop, $k = l$ and $A$ is sorted under the
complete ordering
\[
\tau = (\sigma_1, \ldots, \sigma_{l' - 1}, \tau_{l'}, \ldots, \tau_r).
\]

Notice that now $f(\sigma, \sigma_{k + 1}) \subseteq f((1, 2, 
\ldots), \sigma_{k + 1})$, either because that was the condition that stopped
the loop on line
\ref{algo:quesadilla:scanloop} or because that loop stopped when $l' = r$ and
$f(\sigma, \sigma_r) = \emptyset$. Thus, for all $j > l'$,
$\sigma_j > \sigma_{l'}$, and since $(\tau_{l'}, \ldots, \tau_r)$ is
ascending, $\tau_{l'} = \sigma_{l'}$. Thus, we set $l$ to $l'$ and
we have $(\tau_1, \ldots, \tau_l) = (\sigma_1, \ldots, \sigma_l)$ when we reach line
\ref{algo:quesadilla:sortloop}. The other claims will hold because $k + 1$
will be equal to $l'$.
\end{proof}

\begin{theorem}\label{thm:quesadilla_minimal}
    Given a target ordering $\sigma$, $\text{\sc QuesadillaSort}(\sigma)$ uses
    the minimum-length sequence of calls to \textsc{PartialSort} required to sort
    any simply ordered list of $r$-coordinates to $\sigma$ order.
\end{theorem}

\begin{proof}
\textsc{QuesadillaSort} only calls \textsc{PartialSort} on
modes $\sigma_k$ where $f(\sigma, \sigma_k) \not\subseteq f((1, 2,
\ldots), \sigma_k)$. Thus, Lemma \ref{lem:mustsort} implies that
\textsc{QuesadillaSort} makes the minimum number of required calls to
\textsc{PartialSort}.
\end{proof}

\begin{theorem}\label{thm:quesadilla_minimal_buckets}
    Among minimum-length sequences of \textsc{PartialSort} calls that
    sort simply ordered lists of $r$-coordinates to target ordering $\sigma$,
    the sequence used by $\text{\sc QuesadillaSort}(\sigma)$ minimizes the
    number of bucketed partial sorts.
\end{theorem}

\begin{proof}
    The first execution of the loop on line
    \ref{algo:quesadilla:scanloop} stops once $f(\sigma, \sigma_{k + 1})
    \subseteq f((1, 2, \ldots), \sigma_{k + 1})$. For all $j > k + 1$, we have
    $\sigma_j \in f(\sigma, \sigma_{k + 1})$. Since \textsc{Quesadilla} uses
    non-bucketed sorts for all sorts on modes $\sigma_1$ through $\sigma_{k + 1}$, Lemma
    \ref{lem:mustbucket} implies that \textsc{Quesadilla} uses the minimum number
    of bucketed partial sorts possible when the total number of sorts is minimized.
\end{proof}


%

\subsection{Top-$K$-sadilla Sort}

Although the two sorting primitives presented are both histogram sort
variants, they could be replaced with any stable sort such as quicksort or
merge sort. However, if a comparison sort is used at some level $k$ where the
current ordering is $\tau$ and $(\tau_1, \ldots, \tau_l) = (\sigma_1, \ldots,
\sigma_l)$, it makes more sense to completely sort each
bucket (equivalence class under $(\tau_1, \ldots, \tau_l)$) to $\sigma$ order.

Thus, we propose the \text{\sc Top-$K$-sadilla} algorithm, which uses
\textsc{Quesadilla} to sort the tensor to $(\sigma_1, \ldots, \sigma_K)$
order, then sorts each bucket using quicksort. The best choice of
the value $K$ will be investigated in our experiments, since it depends both
on the permutation and on the dimension of the tensor.

\section{Evaluation}

We evaluate \text{\sc Quesadilla} and \text{\sc Top-$K$-sadilla} sort against various state of the art approaches for sparse tensor transposition.
As we will show, on the whole, our technique outperforms these existing approaches.

\subsection{Experimental Setup}
We created both parallel and serial implementations of our technique. We implemented the serial version in a code generator that emits C++ code to transpose sparse tensors stored in the COO format using either \text{\sc Quesadilla} or \text{\sc Top-$K$-sadilla} sort. 
We implemented the parallel version by implementing parallel counting sort and bucketed counting sort primitives and calling the necessary sorts for \text{\sc Quesadilla}. To implement the quicksort portion of \text{\sc Top-$K$-sadilla}, we identified the buckets in parallel and then sorted each bucket using the OpenMP for-loop parallelization construct. The buckets were scheduled using dynamic scheduling for \text{\sc Top-$1$-sadilla} and guided scheduling for \text{\sc Top-$K$-sadilla} when $K$ > 1. We made these scheduling choices because we expected more smaller buckets when buckets correspond to more coordinates. The overhead for dynamically scheduling many small buckets caused significant slow down.

Our serial implementation is available at \url{https://github.com/suzmue/taco/tree/transpose} and our parallel implementation is available at \url{https://github.com/suzmue/splatt}.

To evaluate our technique, we compare it against SPLATT~\cite{smith_splatt:_2015}, a high-performance C++ toolkit for sparse tensor factorization that uses a combination of histogram sort, quicksort, and insertion sort to sort tensors in COO.
We also evaluate against sparse tensor transposition routines that sort nonzeros with (least significant digit) radix sort (using Algorithm~\ref{alg:histogram} for each pass) or glibc's implementation of \lstinline{qsort}.

We ran all experiments on a 2.5 GHz Intel Xeon E5-2680 v3 machine with 24 cores, 30 MB of L3 cache and 128 GB of main memory.
The machine runs Ubuntu 18.04.3 LTS with glibc 2.27.
We compiled the benchmarks using GCC 7.4.0.
We ran each experiment 100 times and report minimum execution times.

We ran our experiments on real-world tensors obtained from the FROSTT Tensor Collection~\cite{smith_frostt_2017}.
Table~\ref{tab:input-tensors} reports statistics about these tensors.
We stored tensors in the COO format and stored coordinates of nonzeros using 32-bit integers.

\begin{table}
\caption{Statistics about tensors used in our experiments.}
\centering
\small
\begin{tabular}{lrl} 
\toprule
\multicolumn{1}{c}{Tensor} & \multicolumn{1}{c}{Nonzeros} & \multicolumn{1}{c}{Dimensions} \\ \midrule
flickr-3d & 112890310 & 319686 $\times$ 28153045 $\times$ 1607191
\\
nell-1 & 143599552 & 2902330 $\times$ 2143368 $\times$ 25495389
\\
nell-2 & 76879419 & 12092 $\times$ 9184 $\times$ 28818
\\
vast-2015-mc1-3d & 26021854 & 165427 $\times$ 11374 $\times$ 2
\\ \midrule
chicago-crime-comm & 5330673 & 6186 $\times$ 24 $\times$ 77 $\times$ 32
\\
delicious-4d & 140126220 & 532924 $\times$ 17262471 $\times$ 2480308 $\times$ 1443
\\
enron & 54202099 & 6066 $\times$ 5699 $\times$ 44268 $\times$ 1176
\\
flickr-4d & 112890310 & 319686 $\times$ 28153045 $\times$ 1607191 $\times$ 731
\\
nips & 3101609 & 2482 $\times$ 2862 $\times$ 14036 $\times$ 17
\\
uber & 3309490 & 183 $\times$ 24 $\times$ 1140 $\times$ 1717
\\ \midrule
lbnl-network & 1698825 & 1605 $\times$ 4198 $\times$ 1631 $\times$ 4209 $\times$ 868131
\\
vast-2015-mc1-5d & 26021945 & 165427 $\times$ 11374 $\times$ 2 $\times$ 100 $\times$ 89
\\ \bottomrule
\end{tabular}
\label{tab:input-tensors}
\end{table}

\subsection{Performance Evaluation}

\begin{figure}
\centering
\subfloat[serial]{
    \begin{tikzpicture}
        \begin{axis}[
         boxplot/draw direction=y,
        ymode=log,
        ymin=0.09,
        y tick label style={log ticks with fixed point},
        ylabel={Normalized execution time},
        xmax=6.4,
        xtick={1,2,3,4,5},
        xticklabel style={align=center,rotate=45},
        xticklabels={qsort\\(0\%),1-sadilla\\(2.2\%),2-sadilla\\(15\%),quesadilla\\(58\%),radix\\(0.25\%)}
        ]
        \addplot+ [boxplot prepared={
        lower whisker=0.61, lower quartile=1.91 ,
        median=2.32, upper quartile=3.18,
        upper whisker=6.36}
        ] coordinates {};
        \addplot[color=black,samples=2,smooth,ultra thick] coordinates {(0.5,1) (5.5,1)} node[right,pos=1,align=center] {splatt\\(25\%)};
        \addplot+ [boxplot prepared={
        lower whisker=0.31, lower quartile=1.34 ,
        median=1.54, upper quartile=1.83,
        upper whisker=3.91}
        ] coordinates {};
        \addplot+ [boxplot prepared={
        lower whisker=0.22, lower quartile=0.91 ,
        median=1.19, upper quartile=1.43,
        upper whisker=3.91}
        ] coordinates {};
        \addplot+ [boxplot prepared={
        lower whisker=0.0000001, lower quartile=0.54 ,
        median=0.84, upper quartile=1.27,
        upper whisker=5.78}
        ] coordinates {};
        \addplot+ [boxplot prepared={
        lower whisker=0.51, lower quartile=1.0 ,
        median=1.41, upper quartile=2.26,
        upper whisker=7.84}
        ] coordinates {};
        \end{axis}
        \end{tikzpicture}
        }

\subfloat[parallel]{
    \begin{tikzpicture}
        \begin{axis}[
         boxplot/draw direction=y,
        ymode=log,
        ymin=0.05,
        y tick label style={log ticks with fixed point},
        ylabel={Normalized execution time},
        xmax=6.4,
        xtick={1,2,3,4,5},
        xticklabel style={align=center,rotate=45},
        xticklabels={qsort\\(0\%),1-sadilla\\(12\%),2-sadilla\\(28\%),quesadilla\\(24\%),radix\\(0\%)}
        ]
        \addplot+ [boxplot prepared={
            lower whisker=1.2776128821716999, lower quartile=4.4 ,
            median=20.38, upper quartile=28.29,
            upper whisker=86.7}
            ] coordinates {};
            \addplot[color=black,samples=2,smooth,ultra thick] coordinates {(0.5,1) (5.5,1)} node[right,pos=1,align=center] {splatt\\(36\%)};
            \addplot+ [boxplot prepared={
            lower whisker=0.26983874001633557, lower quartile=0.99 ,
            median=1.06, upper quartile=1.24,
            upper whisker=2.02}
            ] coordinates {};
            \addplot+ [boxplot prepared={
            lower whisker=0.06302045278677847, lower quartile=0.62 ,
            median=1.08, upper quartile=1.49,
            upper whisker=5.69}
            ] coordinates {};
            \addplot+ [boxplot prepared={
            lower whisker=2.973154144423424e-07, lower quartile=0.64 ,
            median=1.25, upper quartile=2.09,
            upper whisker=6.82}
            ] coordinates {};
            \addplot+ [boxplot prepared={
            lower whisker=0.13142526670319157, lower quartile=1.25 ,
            median=2.22, upper quartile=3.84,
            upper whisker=9.83}
            ] coordinates {};
            \end{axis}
            \end{tikzpicture}
}

\caption{Normalized execution times of sparse tensor transposition with various algorithms, aggregated over all 408 possible combinations of test tensors and output orderings.  Percentages in parentheses indicate the proportion of combinations for which each algorithm is the fastest.  Results are normalized to SPLATT (horizontal line) for each tensor and output ordering.  \text{\sc Top-1-sadilla} denotes \text{\sc Top-$K$-sadilla} with $K=1$.}
\label{fig:aggregate-results}
\end{figure}
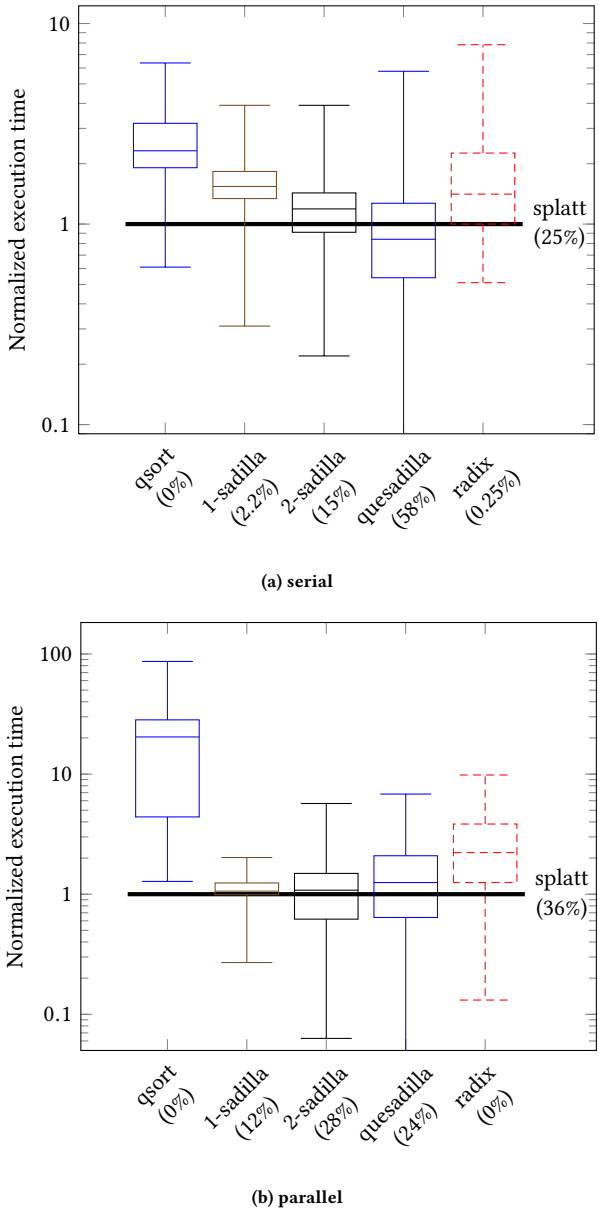

For each tensor in Table~\ref{tab:input-tensors}, we measured the normalized running times of SPLATT, \lstinline{qsort}, \text{\sc Top-$K$-sadilla}, \text{\sc Quesadilla}, and radix sort for transposing the tensor from its initial ordering $\sigma = (1, ..., r)$ to every $r!$ possible ordering.
Figure~\ref{fig:aggregate-results} shows the results of these experiments aggregated over all 408 possible combinations of input tensors and output orderings.
The appendix includes more detailed results that show the performance of each algorithm for every combination of input tensor and output ordering.

In serial tests, these results demonstrate that \text{\sc Quesadilla}
outperforms SPLATT, radix sort, and \lstinline{qsort} on 60\% of the sparse
tensor transpositions. For half of all combinations, \text{\sc Quesadilla} is
at least 1.19$\times$ faster than SPLATT, 1.68$\times$ faster than radix
sort, and 2.76$\times$ faster than \lstinline{qsort}. In parallel tests, at
least one of \text{\sc Quesadilla} or \text{\sc{Top-2-sadilla}} was the best
strategy for 52\% of all tensor and transposition combinations.

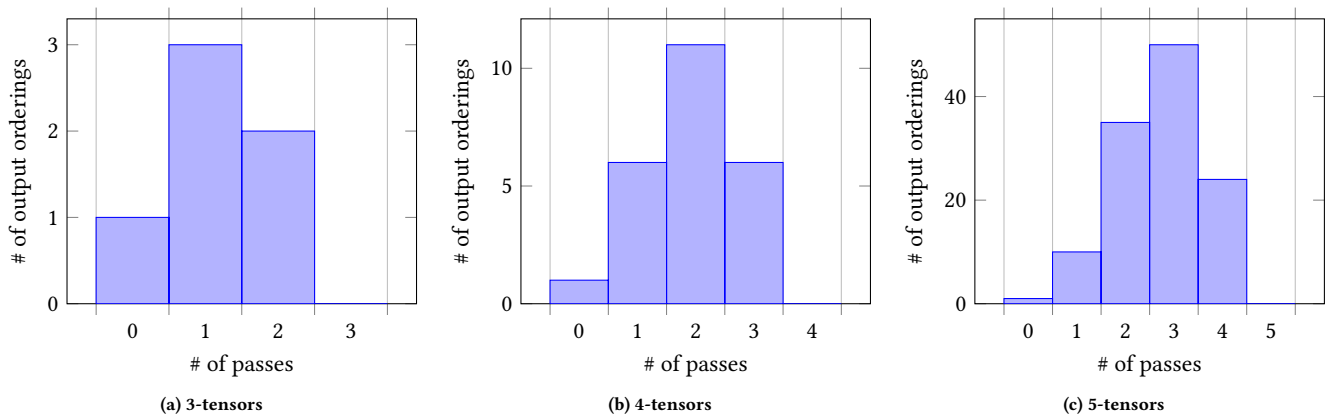
\begin{figure*}
\centering
\subfloat[3-tensors]{%
\begin{tikzpicture}
\begin{axis}[
width=0.35\textwidth,
ybar interval,
enlarge y limits={upper=0},
ylabel={\# of output orderings},
xlabel={\# of passes}
]
\addplot coordinates {
 (0,1)
 (1,3)
 (2,2)
 (3,0)
  (4,0)
 };
\end{axis}
\end{tikzpicture}
}
\hspace{2mm}
\subfloat[4-tensors]{%
\begin{tikzpicture}
\begin{axis}[
width=0.35\textwidth,
ybar interval,
enlarge y limits={upper=0},
ylabel={\# of output orderings},
xlabel={\# of passes}
]
\addplot coordinates {
 (0,1)
 (1,6)
 (2,11)
 (3,6)
 (4,0)
  (5,0)
 };
\end{axis}
\end{tikzpicture}
}
\hspace{2mm}
\subfloat[5-tensors]{%
\begin{tikzpicture}
\begin{axis}[
ybar interval,
width=0.35\textwidth,
enlarge y limits={upper=0},
ylabel={\# of output orderings},
xlabel={\# of passes}
]
\addplot coordinates {
 (0,1)
 (1,10)
 (2,35)
 (3,50)
 (4,24)
 (5,0)
 (6,0)
 };
\end{axis}
\end{tikzpicture}
}
\caption{Distributions of the number of sorting passes needed by \text{\sc Quesadilla} to transpose tensors of varying rank.}
\label{fig:eliminated-passes}
\end{figure*}

\begin{figure*}
\centering
\subfloat[Sort required on last mode (serial)]{    
\begin{tikzpicture}
\begin{axis}[
boxplot/draw direction=y,
xmin = 0.45,
xmax = 7.55,
ymode = log,
ymin=0.25,
ymax=6.6,
y tick label style={log ticks with fixed point},
ytick={0.1, 0.2, 0.3, 0.4, 0.5, 1, 1.5, 2, 3, 4, 5, 6},
ylabel=Normalized execution time, 
xtick={1,2,3,4,5,6, 7},
xticklabel style={align=center,font=\footnotesize},
xticklabels={qsort\\(0\%),1-sadilla\\(0\%),2-sadilla\\(18.8\%),3-sadilla\\(6.25\%),4-sadilla\\(0\%),quesadilla\\(33.3\%),radix\\(0\%)}
]
\addplot+ [boxplot prepared={
lower whisker=1.412925302693152, lower quartile=2.1196618412065584 ,
median=2.5592829618071224, upper quartile=3.733662488935095,
upper whisker=6.361807466533736}
] coordinates {};
\addplot[color=black,samples=2,smooth,ultra thick] coordinates {(0.5,1) (7.5,1)} node[below,pos=0.05,align=center] {\small splatt\\(36\%)};
\addplot+ [boxplot prepared={
lower whisker=1.0172613595557694, lower quartile=1.5170405518628376 ,
median=1.7590638363534845, upper quartile=1.8891245360002524,
upper whisker=2.2084654711464755}
] coordinates {};
\addplot+ [boxplot prepared={
lower whisker=0.89367832008742, lower quartile=1.090118783454837 ,
median=1.443702530812307, upper quartile=1.7276085075559162,
upper whisker=2.390701662603499}
] coordinates {};
\addplot+ [boxplot prepared={
lower whisker=0.8750474026137663, lower quartile=1.1067723751943184 ,
median=1.267125965341029, upper quartile=1.6532099179772897,
upper whisker=3.2282255353744986}
] coordinates {};
\addplot+ [boxplot prepared={
lower whisker=0.8412729456157428, lower quartile=1.030108508855823 ,
median=1.2873064143244533, upper quartile=1.7322538256846793,
upper whisker=4.625910896358661}
] coordinates {};
\addplot+ [boxplot prepared={
lower whisker=0.699845956744554, lower quartile=0.8995458291005227 ,
median=1.1387139929259646, upper quartile=1.5711000181918948,
upper whisker=4.203885185224443}
] coordinates {};
\addplot+ [boxplot prepared={
lower whisker=0.8811747214702095, lower quartile=1.2352013386554037 ,
median=1.4460492608927047, upper quartile=1.9839666044565856,
upper whisker=3.4367855183694114}
] coordinates {};
\end{axis}
\end{tikzpicture}
}
\hspace{4 mm}
\subfloat[No sort required on last mode (serial)]{
\begin{tikzpicture}
\begin{axis}[
 boxplot/draw direction=y,
xmin = 0.45,
xmax = 7.55,
ymin=0.25,
ymax=6.6,
ymode=log,
y tick label style={log ticks with fixed point},
ytick={0.1, 0.2, 0.3, 0.4, 0.5, 1, 1.5, 2, 3, 4, 5, 6},
ylabel=Normalized execution time,
xtick={1,2,3,4,5,6,7},
xticklabel style={align=center,font=\footnotesize},
xticklabels={qsort\\(0\%),1-sadilla\\(0\%),2-sadilla\\(0\%),3-sadilla\\(0\%),4-sadilla\\(0\%),quesadilla\\(100\%),radix\\(0\%)}
]
\addplot+ [boxplot prepared={
lower whisker=1.4444242979702309, lower quartile=1.6768999451724331 ,
median=1.78992074596488, upper quartile=2.049994421956943,
upper whisker=2.824507968293154}
] coordinates {};
\addplot[color=black,samples=2,smooth,ultra thick] coordinates {(0.5,1) (7.5,1)} node[below,pos=0.05,align=center] {\small splatt\\(0\%)};
\addplot+ [boxplot prepared={
lower whisker=0.986069614693971, lower quartile=1.2463801797279115 ,
median=1.3586986553449707, upper quartile=1.4352981059827814,
upper whisker=1.879344074130988}
] coordinates {};
\addplot+ [boxplot prepared={
lower whisker=1.0846173164178006, lower quartile=1.169995924810673 ,
median=1.2212069212635304, upper quartile=1.3284197910533362,
upper whisker=1.6397531108439642}
] coordinates {};
\addplot+ [boxplot prepared={
lower whisker=1.0681223892039622, lower quartile=1.1769922651722617 ,
median=1.2280867012952594, upper quartile=1.3919326199586197,
upper whisker=1.5815545626232266}
] coordinates {};
\addplot+ [boxplot prepared={
lower whisker=1.0360769573126118, lower quartile=1.1279167396132208 ,
median=1.224517838348123, upper quartile=1.3427897394642125,
upper whisker=1.5351350261413381}
] coordinates {};
\addplot+ [boxplot prepared={
lower whisker=0.00000001, lower quartile=0.3629366083240041 ,
median=0.41560209999103914, upper quartile=0.5146290248058191,
upper whisker=0.6213630904360449}
] coordinates {};
\addplot+ [boxplot prepared={
lower whisker=0.8670487216980513, lower quartile=0.911484455612947 ,
median=0.9580684067683403, upper quartile=1.1035930435829264,
upper whisker=1.461404301280929}
] coordinates {};
\addplot+ [boxplot prepared={
lower whisker=0.986069614693971, lower quartile=1.2463801797279115 ,
median=1.3586986553449707, upper quartile=1.4352981059827814,
upper whisker=1.879344074130988}
] coordinates {};
\end{axis}
\end{tikzpicture}
}

\subfloat[Sort required on last mode (parallel)]{
    \begin{tikzpicture}
        \begin{axis}[
         boxplot/draw direction=y,
         ymode=log,
           ymin=0.04,
        ymax=10,
        y tick label style={log ticks with fixed point},
        ytick={0.05, 0.1, 0.2, 0.3, 0.4, 0.5, 1, 1.5, 2, 3, 4, 5, 7,9},
        ylabel=Normalized execution time, 
        xmin = 0.45,
        xmax = 7.55,
        xtick={1,2,3,4,5,6, 7},
        xticklabel style={align=center,font=\footnotesize},
        xticklabels={qsort\\(0\%),1-sadilla\\(14.6\%),2-sadilla\\(43.8\%),3-sadilla\\(25\%),4-sadilla\\(1.04\%),quesadilla\\(2.08\%),radix\\(0\%)}
        ]
        \addplot+ [boxplot prepared={
        lower whisker=1.7931824067358233, lower quartile=2.6430758536675865 ,
        median=3.8842287511077913, upper quartile=5.624083113320541,
        upper whisker=9.340446321854028}
        ] coordinates {};
        \addplot[color=black,samples=2,smooth,ultra thick] coordinates {(0.5,1) (7.5,1)} node[below,pos=0.05,align=center] {\small splatt\\(14\%)};
        \addplot+ [boxplot prepared={
        lower whisker=0.7209110982768987, lower quartile=0.9035587749629512 ,
        median=1.0088162096653372, upper quartile=1.0535247796044351,
        upper whisker=1.2918243352673278}
        ] coordinates {};
        \addplot+ [boxplot prepared={
        lower whisker=0.1943247072575996, lower quartile=0.5365106297780777 ,
        median=0.6580714869990805, upper quartile=1.0041275601433277,
        upper whisker=3.0240469912714705}
        ] coordinates {};
        \addplot+ [boxplot prepared={
        lower whisker=0.17122275423345013, lower quartile=0.5515666332123627 ,
        median=0.9868814912728949, upper quartile=1.3712158771251133,
        upper whisker=3.004430725808945}
        ] coordinates {};
        \addplot+ [boxplot prepared={
        lower whisker=0.495581595515349, lower quartile=0.688533210982723 ,
        median=1.2060660093381403, upper quartile=1.641391412303317,
        upper whisker=3.4300837567036746}
        ] coordinates {};
        \addplot+ [boxplot prepared={
        lower whisker=0.5284658639046319, lower quartile=0.6443517828916756 ,
        median=1.1930410870441412, upper quartile=1.6182491349345858,
        upper whisker=2.971874414178237}
        ] coordinates {};
        \addplot+ [boxplot prepared={
        lower whisker=0.5591210803053872, lower quartile=0.8471531296488181 ,
        median=1.3105347988768226, upper quartile=1.975346911323074,
        upper whisker=3.024655473879586}
        ] coordinates {};
        \end{axis}
        \end{tikzpicture}
}
\hspace{4mm}
\subfloat[No sort required on last mode (parallel)]{
    \begin{tikzpicture}
        \begin{axis}[
         boxplot/draw direction=y,
         ymode = log,
        ymin=0.04,
        ymax=10,
        y tick label style={log ticks with fixed point},
        ytick={0.05, 0.1, 0.2, 0.3, 0.4, 0.5, 1, 1.5, 2, 3, 4, 5, 7,9},
        ylabel=Normalized execution time, 
        xmin = 0.45,
        xmax = 7.55,
        xtick={1,2,3,4,5,6, 7},
        xticklabel style={align=center,font=\footnotesize},
        xticklabels={qsort\\(0\%),1-sadilla\\(0\%),2-sadilla\\(0\%),3-sadilla\\(0\%),4-sadilla\\(0\%),quesadilla\\(100\%),radix\\(0\%)}
        ]
        \addplot+ [boxplot prepared={
            lower whisker=2.979762247809981, lower quartile=3.2877723083813732 ,
            median=4.159196776454441, upper quartile=5.640530394499387,
            upper whisker=9.724420852056829}
            ] coordinates {};
            \addplot[color=black,samples=2,smooth,ultra thick] coordinates {(0.5,1) (7.5,1)} node[below,pos=0.05,align=center] {\small splatt\\(0\%)};
           \addplot+ [boxplot prepared={
lower whisker=0.7743447278003213, lower quartile=0.9150982537301806 ,
median=1.002271961441089, upper quartile=1.0777197625655672,
upper whisker=1.2444756130184254}
] coordinates {};
\addplot+ [boxplot prepared={
lower whisker=0.20975399801423816, lower quartile=0.53504458118924 ,
median=0.710821663163784, upper quartile=0.8237667155713103,
upper whisker=1.2063443306769408}
] coordinates {};
\addplot+ [boxplot prepared={
lower whisker=0.19263063986986498, lower quartile=0.3089054265630806 ,
median=0.4893668307546333, upper quartile=0.6707210477376914,
upper whisker=1.2434038141745467}
] coordinates {};
\addplot+ [boxplot prepared={
lower whisker=0.17063931539881846, lower quartile=0.25161801845864257 ,
median=0.40310690900202517, upper quartile=0.6200444203471318,
upper whisker=1.092514599036201}
] coordinates {};
\addplot+ [boxplot prepared={
lower whisker=0.00001, lower quartile=0.07911931533299622 ,
median=0.13884755264631563, upper quartile=0.31018213509700954,
upper whisker=0.8033147577229772}
] coordinates {};
\addplot+ [boxplot prepared={
lower whisker=0.9418560790627816, lower quartile=1.0893215581126023 ,
median=1.5619000174364865, upper quartile=1.891018080007967,
upper whisker=3.209592256566592}
] coordinates {};
        \end{axis}
        \end{tikzpicture}
}
\caption{Normalized execution times of sparse tensor transposition with various algorithms for the lbnl-network tensor, aggregated over (a) all output orderings where \text{\sc Quesadilla} needs to sort on the last mode and (b) all output orderings where the last mode does not need to be sorted.  Percentages in parentheses indicate the proportion of combinations for which each algorithm is the fastest.  Results are normalized to SPLATT (horizontal lines) for each tensor and output ordering. Again, e.g. \text{\sc Top-1-sadilla} denotes \text{\sc Top-$K$-sadilla} with $K=1$.}
\label{fig:results-lbnl}
\end{figure*}
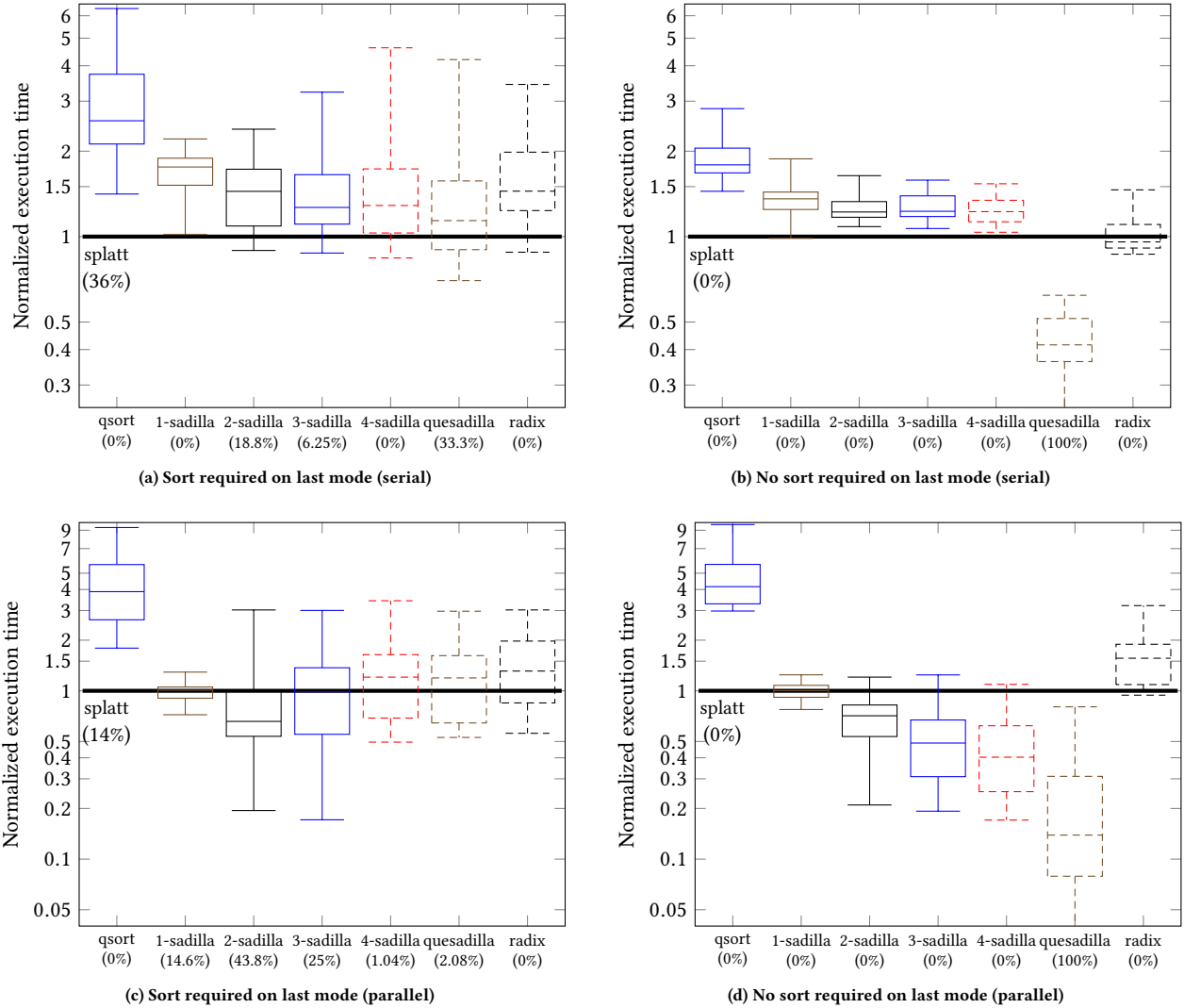

\begin{figure*}
\centering
\subfloat[lbnl-network (serial)]{%
\begin{tikzpicture}
\begin{axis}[
ybar,
ymin = 0,
scale = .8, 
enlargelimits=0.075,
enlarge y limits={upper=0},
ylabel={Execution time (ms)},
symbolic x coords={(2 []),(3 []),(4 []),(5 []),(3 [1]),(4 [1]),(5 [1]),(4 [1 2]),(5 [1 2]),(5 [1 2 3])},
xtick=data,
x tick label style={rotate=90,anchor=east}
]

\addplot coordinates {
 ((2 []),44.9451)
 ((3 []),45.9637)
 ((4 []),46.4687)
 ((5 []),71.9164)
 ((3 [1]),61.9874)
 ((4 [1]),66.1256)
 ((5 [1]),96.1472)
 ((4 [1 2]),67.2675)
 ((5 [1 2]),89.5054)
 ((5 [1 2 3]),101.775)
 };
\end{axis}
\end{tikzpicture}
}
\hspace{4mm}
\subfloat[vast-2015-mc1-5d (serial)]{%
\begin{tikzpicture}
\begin{axis}[
ybar,
ymin = 0,
scale = .8, 
enlargelimits=0.075,
enlarge y limits={upper=0},
ylabel={Execution time (ms)},
symbolic x coords={(2 []),(3 []),(4 []),(5 []),(3 [1]),(4 [1]),(5 [1]),(4 [1 2]),(5 [1 2]),(5 [1 2 3])},
xtick=data,
x tick label style={rotate=90,anchor=east}
]

\addplot coordinates {
 ((2 []),1558.43)
 ((3 []),664.807)
 ((4 []),843.065)
 ((5 []),843.578)
 ((3 [1]),971.527)
 ((4 [1]),2553.26)
 ((5 [1]),2538.61)
 ((4 [1 2]),3565.54)
 ((5 [1 2]),3472.33)
 ((5 [1 2 3]),3347.07)
 };
\end{axis}
\end{tikzpicture}
}

\subfloat[lbnl-network (parallel)]{%
\begin{tikzpicture}
\begin{axis}[
ybar,
scale = .8, 
enlargelimits=0.05,
ylabel={time (ms)},
xlabel={sort},
symbolic x coords={(2 []),(3 []),(4 []),(5 []),(3 [1]),(4 [1]),(5 [1]),(4 [1 2]),(5 [1 2]),(5 [1 2 3])},
xtick=data,
x tick label style={rotate=90,anchor=east}
]

\addplot coordinates {
 ((2 []),3.451671451330185)
 ((3 []),3.5523921251296997)
 ((4 []),4.220690578222275)
 ((5 []),106.08839616179466)
 ((3 [1]),26.423150673508644)
 ((4 [1]),24.866636842489243)
 ((5 [1]),68.97846609354019)
 ((4 [1 2]),24.102650582790375)
 ((5 [1 2]),72.22556695342064)
 ((5 [1 2 3]),65.76675176620483)
 };
\end{axis}
\end{tikzpicture}
}
\hspace{4mm}
\subfloat[vast-2015-mc1-5d (parallel)]{%
\begin{tikzpicture}
\begin{axis}[
ybar,
scale = .8, 
enlargelimits=0.05,
ylabel={time (ms)},
xlabel={sort},
symbolic x coords={(2 []),(3 []),(4 []),(5 []),(3 [1]),(4 [1]),(5 [1]),(4 [1 2]),(5 [1 2]),(5 [1 2 3])},
xtick=data,
x tick label style={rotate=90,anchor=east}
]

\addplot coordinates {
 ((2 []),304.16058003902435)
 ((3 []),179.5085184276104)
 ((4 []),179.37827482819557)
 ((5 []),175.57837441563606)
 ((3 [1]),282.12855756282806)
 ((4 [1]),330.7906538248062)
 ((5 [1]),330.170214176178)
 ((4 [1 2]),420.0591817498207)
 ((5 [1 2]),408.06737169623375)
 ((5 [1 2 3]),413.84611651301384)
 };
\end{axis}
\end{tikzpicture}
}
\centering
\caption{
Execution times of \text{\sc PartialSort} for sorting different modes of two test tensors.  Labels along the x-axes indicate the modes being sorted; for instance, (4 [1 2]) denotes sort on mode 4 assuming modes 1 and 2 are bucketed.
}
\label{fig:single-pass-results}
\end{figure*}
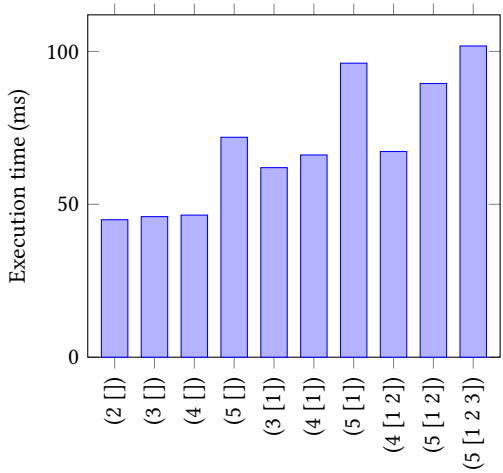
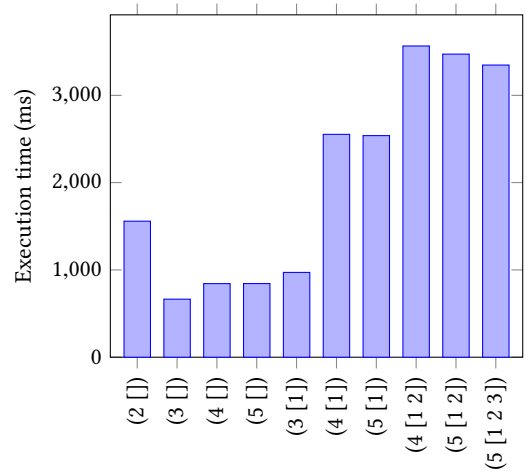
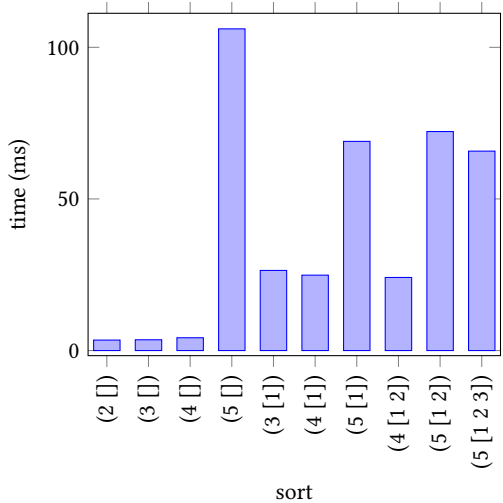
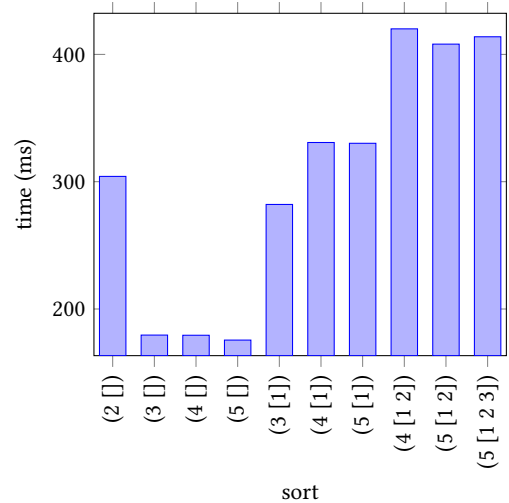

\text{\sc Quesadilla} is able to significantly outperform radix sort by minimizing the number of passes over the input tensor.
As Figure~\ref{fig:eliminated-passes} shows, \text{\sc Quesadilla} exploits the partial ordering of the input tensor to eliminate at least one sorting pass for all possible output orderings and eliminate two or more passes for the majority of output orderings.
By contrast, radix sort always makes as many sorting passes as there are modes in the input tensor, thereby incurring overhead from unnecessary memory traffic.

\text{\sc Quesadilla}'s performance, however, depends to a large degree on the dimensions of the input tensor as well as the ordering of modes in the output.
In particular, \text{\sc Quesadilla} is more efficient when it does not have to sort large modes.
Figure~\ref{fig:results-lbnl}, for instance, shows \text{\sc Quesadilla}'s performance for the lbnl-network tensor, whose last mode is significantly larger than the other modes.
For output orderings where \text{\sc Quesadilla} does not have to sort the last mode, \text{\sc Quesadilla} significantly outperforms all other algorithms we evaluate.
On the other hand, SPLATT and \text{\sc Top-$K$-sadilla} (where $K < r$) are more efficient for the other output orderings, with both being faster than \text{\sc Quesadilla} in approximately two-thirds of cases where \text{\sc Quesadilla} must sort the last mode in the serial implementation, and nearly all cases in the parallel implementation.
This can be attributed to the fact that each invocation of \text{\sc PartialSort} in \text{\sc Quesadilla} requires a histogram containing $n_k$ bins, where $n_k$ is the size of the mode being sorted.
When $n_k$ is large, accesses into the histogram are less likely to hit the cache, thereby limiting performance.
Furthermore, constructing the histogram incurs $O(n_k)$ overhead, which becomes more significant when $n_k$ is large.
Thus, as Figure~\ref{fig:single-pass-results} shows, \text{\sc PartialSort} is significantly slower for large modes than for small modes, assuming the bucketed dimensions are the same.
This, in turn, limits \text{\sc Quesadilla}'s performance for input tensors and output orderings that require sorting large modes.
By contrast, SPLATT and \text{\sc Top-$K$-sadilla} use comparison-based sorting algorithms to sort all but the first mode or the first several modes respectively, thus making their performance less dependent on the dimensions of the input tensor. 

When $K=1$, \text{\sc Top-$K$-sadilla} reduces to the \text{\sc Top-1-sadilla} algorithm that is similar to what SPLATT implements for sorting COO tensors, which we summarize in Section~\ref{transposition}.
Unlike SPLATT, which uses a custom hand-optimized implementation of quicksort, serial \text{\sc Top-1-sadilla} uses \lstinline{qsort} from C \lstinline{stdlib} to sort nonzeros within each bucket created by the initial histogram sort.
As Figure~\ref{fig:aggregate-results} shows, serial SPLATT outperforms serial \text{\sc Top-1-sadilla} for most tensor transpositions in our experiments, thereby demonstrating that SPLATT's custom implementation of quicksort is more efficient than \lstinline{qsort}.
This performance difference suggests we can improve \text{\sc Top-$K$-sadilla}'s performance by using more optimized implementations of comparison sort to sort each bucket.

\section{Conclusion}

We have described an algorithm to transpose sparse tensors faster than simply sorting a list of coordinates.
By taking advantage of the lexicographic ordering of the input and knowledge of the requested transposition, our algorithm applies only a subset of the passes of a radix sort and thereby reduces the amount of work required to sort the coordinates.
We provide two non-comparison based partial sorting algorithms for radix sort passes that are optimized for different situations.
We prove two things: (1) We prove that our algorithm minimizes the total calls to either sorting algorithm. (2) We prove that among sorts with the minimum total calls, we minimize the number of calls to the more expensive of the two.
The amount of work required by our algorithm is proportional to the number of modes that need reordering in the transposition.
We evaluated our algorithm empirically with a C++ implementation, and showed that it produced significant improvements over existing approaches.

As sparse tensor representations receive increasing study, diversity in tensor formats will increase and applications will more frequently convert between formats.
Sparse tensor transposition is the most basic instance of sparse format conversion, and an important subroutine in several format conversions.
We have provided evidence that naive algorithms for sparse tensor transpositions can be improved substantially, but there are further improvements that need investigation.

Focusing on the multi-pass coordinate-sorting-based transposition technique we describe, improvements can be made in scheduling passes, the implementation of passes themselves, and handling the buckets.
Although we minimize the number of passes over the data, we don't necessarily pick a schedule of passes that minimizes the true runtime.
Since the bucketed histogram sort costs more than the histogram sort, we can improve our scheduling by minimizing a cost model which reflects the true costs of the passes.

We can improve the implementation of a sorting pass by reducing the size of coordinates using bit-packing techniques.
If the mode to be sorted has a large dimension, it can make sense to perform the histogram sort itself as a radix sort, with multiple passes and a radix smaller than the dimension.
In some cases, we can also fuse the first loop of the next histogram sort into the last loop of the current one, reducing the number of reads.

Discovering the buckets is expensive.
If we need to perform several bucketed histogram sorts with the same buckets, we only need to discover the buckets once at the beginning, perform the histogram sorts, and then sort on the buckets at the end, skipping the bucketing step between the two sorts.
Since $l$ is constant, we can use the same buckets for all iterations of the loop on line \ref{algo:quesadilla:sortloop} of Algorithm \ref{algo:quesadilla}. 
Additionally, instead of evaluating the all $l$ entries of each coordinate to discover the buckets, we can use buckets from the previous pass, which differ precisely when the previous entries differed.
While such an optimization would involve permuting a bucket array, we can avoid examining entire coordinates during bucket discovery, saving a factor of $r$ in the asymptotic analysis.

If our goal is to transpose tensors stored in formats other than COO, including formats like HiCOO \cite{li_hicoo_2018}, BICRS \cite{yzelman_cache-oblivious_2012}, and JAD \cite{saad_krylov_1989}, additional optimizations may present themselves.
For example, instead of converting to coordinates, then sorting, our first histogram sort can iterate over the input format in order, fusing the conversion to coordinates into the first histogram sort. 
Additionally, the sorting techniques we describe in this work may apply directly to the format we want to transpose.
If the tensor is in CSF, for example, it may be possible to sort the nodes in the CSF tree directly, moving the nodes instead of moving entire subtrees.

\begin{acks}
This work was supported by a grant from the Toyota Research Institute, DARPA
PAPPA Grant HR00112090017, and a DOE CSGF Fellowship DE-FG02-97ER25308.
\end{acks}

\bibliographystyle{ACM-Reference-Format}
\bibliography{main}

\pagebreak[4]

\appendix
\clearpage
\newpage
\begin{onecolumn}
\section{Aggregate Results}

These tables contain statistics about the performance of the algorithms across all permutations and tensors. In addition, we counted the number of times that each strategy was the best of all of the strategies. We exclude Top-2-sadilla and Top-3-sadilla from these results, as the strategy is not comparable across tensors of different orders.
\end{onecolumn}

\twocolumn

\begin{table*}[h]
\caption{Aggregate timing results (serial)}
\begin{tabular}{ |c|c|c|c|c|c|c| } 
\hline
stat & qsort & splatt & 1-sadilla & 2-sadilla & quesadilla & radix\\
\hline
min & 0.61 & 1.00 & 0.31 & 0.22 & 0.00 & 0.51\\
Q1 &  1.91 &  1.00 &  1.34 &  0.91 &  0.54 &  1.00\\
median &  2.32 &  1.00 &  1.54 &  1.19 &  0.84 &  1.41\\
Q3 &  3.18 &  1.00 &  1.83 &  1.43 &  1.27 &  2.26\\
max &  6.36 &  1.00 &  3.91 &  3.91 &  5.78 &  7.84\\
wins &  0 & 25 & 2.2 & 15 & 58 & 0.25\\
\hline
\end{tabular}
\end{table*}

\begin{table*}[h]
\caption{Aggregate timing results (parallel)}
\begin{tabular}{ |c|c|c|c|c|c|c| } 
\hline
stat & qsort & splatt & 1-sadilla & 2-sadilla & quesadilla & radix\\
\hline
min & 1.27 & 1 & 0.27 & 0.063 & 0.00 & 0.13\\
Q1 &  4.40 &  1.00 &  0.99 &  0.62 &  0.64 &  1.25\\
median & 20.38 &  1.00 &  1.06 &  1.08 &  1.25 &  2.22\\
Q3 & 28.29 &  1.00 &  1.24 &  1.49 &  2.09 &  3.84\\
max & 86.70 &  1.00 &  2.02 &  5.69 &  6.82 &  9.83\\
wins &  0 & 36 & 12 & 28 & 24 &  0\\
\hline
\end{tabular}
\end{table*}

\begin{table*}[h]
\caption{Median results by tensor (serial)}
\centering
\begin{tabular}{ |c|c|c|c|c|c|c| } 
    \hline
filename & qsort & splatt & 1-sadilla & 2-sadilla & quesadilla & radix\\
\hline
flickr-3d &  4.31 &  1.00 &  1.71 &  1.45 &  0.99 &  2.40\\
nell-1 &  3.04 &  1.00 &  1.94 &  1.15 &  0.81 &  3.08\\
nell-2 &  2.80 &  1.00 &  1.90 &  1.35 &  0.60 &  1.34\\
vast-2015-mc1-3d &  3.49 &  1.00 &  1.49 &  1.18 &  0.55 &  2.07\\
chicago-crime-comm &  2.18 &  1.00 &  1.71 &  1.08 &  0.52 &  0.93\\
delicious-4d &  2.66 &  1.00 &  1.51 &  1.11 &  1.17 &  1.89\\
enron &  2.43 &  1.00 &  1.53 &  1.06 &  0.64 &  1.18\\
flickr-4d &  3.48 &  1.00 &  1.54 &  1.21 &  1.11 &  2.02\\
nips &  2.84 &  1.00 &  1.88 &  1.35 &  0.81 &  1.79\\
uber &  2.24 &  1.00 &  1.68 &  1.30 &  0.63 &  1.10\\
lbnl-network &  2.29 &  1.00 &  1.67 &  1.34 &  1.04 &  1.37\\
vast-2015-mc1-5d &  1.97 &  1.00 &  1.35 &  0.92 &  0.75 &  1.25\\
\hline
\end{tabular}
\end{table*}

\begin{table*}
\caption{Median by tensor (parallel)}
\centering
\begin{tabular}{ |c|c|c|c|c|c|c| } 
    \hline
    filename & qsort & splatt & 1-sadilla & 2-sadilla & quesadilla & radix\\
    \hline
    flickr-3d & 32.39 &  1.00 &  1.06 &  1.03 &  0.90 &  5.52\\
    nell-1 & 22.94 &  1.00 &  1.08 &  1.05 &  0.88 &  6.33\\
    nell-2 & 32.69 &  1.00 &  1.29 &  1.52 &  1.35 &  2.69\\
    vast-2015-mc1-3d & 25.64 &  1.00 &  1.30 &  0.93 &  0.60 &  2.45\\
    chicago-crime-comm & 19.41 &  1.00 &  1.25 &  0.93 &  0.84 &  1.82\\
    delicious-4d & 27.22 &  1.00 &  1.15 &  1.52 &  1.88 &  4.85\\
    enron & 28.81 &  1.00 &  1.28 &  1.68 &  1.61 &  3.20\\
    flickr-4d & 27.69 &  1.00 &  1.11 &  1.49 &  1.74 &  5.73\\
    nips & 40.06 &  1.00 &  1.43 &  1.37 &  1.24 &  2.02\\
    uber & 41.10 &  1.00 &  1.60 &  1.53 &  0.95 &  1.70\\
    lbnl-network &  3.97 &  1.00 &  1.01 &  0.68 &  1.01 &  1.33\\
    vast-2015-mc1-5d & 23.18 &  1.00 &  1.05 &  1.28 &  1.84 &  3.06\\
    \hline
\end{tabular}
\end{table*}

\begin{table*}[h]
\caption{Wins by tensor (serial)}
\begin{tabular}{ |c|c|c|c|c|c|c| } 
    \hline
filename & qsort & splatt & 1-sadilla & 2-sadilla & quesadilla & radix\\
\hline
flickr-3d &     0\% &    50\% &     0\% &     0\% &    50\% &     0\%\\
nell-1 &     0\% &  33.3\% &     0\% &     0\% &  66.7\% &     0\%\\
nell-2 &     0\% &  33.3\% &     0\% &     0\% &  66.7\% &     0\%\\
vast-2015-mc1-3d &     0\% &     0\% &     0\% &     0\% &   100\% &     0\%\\
chicago-crime-comm &     0\% &  12.5\% &     0\% &     0\% &  87.5\% &     0\%\\
delicious-4d &     0\% &  45.8\% &     0\% &  16.7\% &  37.5\% &     0\%\\
enron &     0\% &  12.5\% &     0\% &  4.17\% &  79.2\% &  4.17\%\\
flickr-4d &     0\% &  45.8\% &     0\% &  12.5\% &  41.7\% &     0\%\\
nips &     0\% &    25\% &  12.5\% &  8.33\% &  54.2\% &     0\%\\
uber &     0\% &  16.7\% &     0\% &     0\% &  83.3\% &     0\%\\
lbnl-network &     0\% &  36.7\% &     0\% &  16.7\% &  46.7\% &     0\%\\
vast-2015-mc1-5d &     0\% &  10.8\% &     5\% &  25.8\% &  58.3\% &     0\%\\
\hline
\end{tabular}
\end{table*}

\begin{table*}
\caption{Wins by tensor (parallel)}
\centering
\begin{tabular}{ |c|c|c|c|c|c|c| } 
    \hline
    filename & qsort & splatt & 1-sadilla & 2-sadilla & quesadilla & radix\\
    \hline
    flickr-3d &     0\% &  16.7\% &  16.7\% &     0\% &  66.7\% &     0\%\\
    nell-1 &     0\% &  16.7\% &  16.7\% &  16.7\% &    50\% &     0\%\\
    nell-2 &     0\% &    50\% &  16.7\% &     0\% &  33.3\% &     0\%\\
    vast-2015-mc1-3d &     0\% &  33.3\% &  16.7\% &     0\% &    50\% &     0\%\\
    chicago-crime-comm &     0\% &  8.33\% &  16.7\% &  20.8\% &  54.2\% &     0\%\\
    delicious-4d &     0\% &    50\% &  16.7\% &  12.5\% &  20.8\% &     0\%\\
    enron &     0\% &  79.2\% &     0\% &     0\% &  20.8\% &     0\%\\
    flickr-4d &     0\% &  41.7\% &  16.7\% &  20.8\% &  20.8\% &     0\%\\
    nips &     0\% &  62.5\% &  4.17\% &  8.33\% &    25\% &     0\%\\
    uber &     0\% &  37.5\% &     0\% &  8.33\% &  54.2\% &     0\%\\
    lbnl-network &     0\% &  11.7\% &  11.7\% &  53.3\% &  23.3\% &     0\%\\
    vast-2015-mc1-5d &     0\% &  49.2\% &    15\% &  26.7\% &  9.17\% &     0\%\\
    \hline
\end{tabular}
\end{table*}
\clearpage
\newpage

\newpage
\begin{onecolumn}
\section{Detailed Results}

These tables contain the results of running all of the experiments. They are organized by file and the permutations are ordered lexicographically. A cell that contains a value of 1 is colored white. This value means that it performed as well as SPLATT. A cell that contains a value $> 1$ is colored red and performed worse than SPLATT. A cell that contains a value $< 1$ is colored blue and performed better than SPLATT.
\end{onecolumn}

\twocolumn
\begin{figure*}[h]
    \centering
    \caption{flickr-3d results normalized by splatt (serial)}

\pgfplotstabletypeset[
    /pgfplots/colormap={whiteblue}{rgb255(0cm)=(49,130,189); rgb255(.14cm)=(255,255,255); rgb255(1cm)=(222,45,38)},
    colorCell/.style={
        color cells={min=0,max=7}
    },
    itemCell/.style={
        string type,
        column name={},
    },
    col sep=comma,
    columns={permutation,qsort,splatt,1-sadilla,2-sadilla,quesadilla,radix},
    columns/permutation/.style={itemCell},
   columns/qsort/.style={colorCell},
   columns/splatt/.style={colorCell},
   columns/1-sadilla/.style={colorCell},
   columns/2-sadilla/.style={colorCell},
   columns/quesadilla/.style={colorCell},
   columns/radix/.style={colorCell},
]{
permutation,qsort,splatt,1-sadilla,2-sadilla,quesadilla,radix
123,4.11808550728536,1.0,1.8382139831549893,0.5617846928868515,0,5.594703938872487
132,2.75092756597953,1.0,1.564882647905631,1.4746990389022663,1.014534999634603,1.853590455844966
213,4.402011806002546,1.0,1.3999019057063113,1.43488823738455,0.9700926900059105,2.444900787951217
231,4.8658617423959605,1.0,1.575739891057223,3.3689069193708012,3.092146152405659,3.40710019419011
312,5.477160549141566,1.0,2.0726373394836024,1.3769028185884427,0.7455381851206154,2.3534761835734748
321,4.216974657128984,1.0,2.000087525946792,1.5636410623934236,1.3951540104594877,1.6595628634540418
}
\end{figure*}
\begin{figure*}[h]
    \centering
    \caption{nell-1 results normalized by splatt (serial)}

\pgfplotstabletypeset[
    /pgfplots/colormap={whiteblue}{rgb255(0cm)=(49,130,189); rgb255(.14cm)=(255,255,255); rgb255(1cm)=(222,45,38)},
    colorCell/.style={
        color cells={min=0,max=7}
    },
    itemCell/.style={
        string type,
        column name={},
    },
    col sep=comma,
    columns={permutation,qsort,splatt,1-sadilla,2-sadilla,quesadilla,radix},
    columns/permutation/.style={itemCell},
   columns/qsort/.style={colorCell},
   columns/splatt/.style={colorCell},
   columns/1-sadilla/.style={colorCell},
   columns/2-sadilla/.style={colorCell},
   columns/quesadilla/.style={colorCell},
   columns/radix/.style={colorCell},
]{
permutation,qsort,splatt,1-sadilla,2-sadilla,quesadilla,radix
123,4.215129852234928,1.0,1.9939447981970528,0.9602057421355381,0,7.840832045247769
132,3.120579599494259,1.0,1.8903795075100787,3.9136140940377278,3.687825580120249,3.7965011668849944
213,4.225125160696888,1.0,2.1373570608391352,1.3443147076620716,0.7039480802744225,4.831746640218517
231,2.9585025325731102,1.0,2.048482654227607,2.5123780461207676,2.355601760298219,2.370951542766977
312,1.7717035814567648,1.0,1.0055601315442972,0.8757679970792224,0.7735909247360611,1.161527505373592
321,1.6996877479725148,1.0,1.0832274200730236,0.9270569858355585,0.8408911683297499,0.9050544966421933
}
\end{figure*}

\begin{figure*}[h]
    \centering
    \caption{nell-2 results normalized by splatt (serial)}

\pgfplotstabletypeset[
    /pgfplots/colormap={whiteblue}{rgb255(0cm)=(49,130,189); rgb255(.14cm)=(255,255,255); rgb255(1cm)=(222,45,38)},
    colorCell/.style={
        color cells={min=0,max=7}
    },
    itemCell/.style={
        string type,
        column name={},
    },
    col sep=comma,
    columns={permutation,qsort,splatt,1-sadilla,2-sadilla,quesadilla,radix},
    columns/permutation/.style={itemCell},
   columns/qsort/.style={colorCell},
   columns/splatt/.style={colorCell},
   columns/1-sadilla/.style={colorCell},
   columns/2-sadilla/.style={colorCell},
   columns/quesadilla/.style={colorCell},
   columns/radix/.style={colorCell},
]{
permutation,qsort,splatt,1-sadilla,2-sadilla,quesadilla,radix
123,4.150286115586783,1.0,2.4082875694396875,1.341319478731659,0,4.1520563529595
132,2.4977817282668955,1.0,1.8041684388640968,1.6163954528192577,1.3405781126658853,1.4700464960674875
213,5.500879197792959,1.0,2.866275209898877,1.763687269013009,0.5156003455285613,3.6015345772309773
231,2.5482539161719644,1.0,1.985944981296006,1.3674375030034953,1.0946371179117576,1.219578268292642
312,3.0418456819507966,1.0,1.3643853255021892,0.9197745907287163,0.6535575281658257,0.9521947090882684
321,2.256800311836851,1.0,1.3985105288520194,0.7731532953827084,0.545360257543548,0.6081827708422005
}
\end{figure*}

\begin{figure*}[h]
    \centering
    \caption{vast-2015-mc1-3d results normalized by splatt (serial)}

\pgfplotstabletypeset[
    /pgfplots/colormap={whiteblue}{rgb255(0cm)=(49,130,189); rgb255(.14cm)=(255,255,255); rgb255(1cm)=(222,45,38)},
    colorCell/.style={
        color cells={min=0,max=7}
    },
    itemCell/.style={
        string type,
        column name={},
    },
    col sep=comma,
    columns={permutation,qsort,splatt,1-sadilla,2-sadilla,quesadilla,radix},
    columns/permutation/.style={itemCell},
   columns/qsort/.style={colorCell},
   columns/splatt/.style={colorCell},
   columns/1-sadilla/.style={colorCell},
   columns/2-sadilla/.style={colorCell},
   columns/quesadilla/.style={colorCell},
   columns/radix/.style={colorCell},
]{
permutation,qsort,splatt,1-sadilla,2-sadilla,quesadilla,radix
123,4.717960362395435,1.0,1.5419366446638256,0.3807820439847572,0,4.682876274709794
132,4.0501300725226015,1.0,1.4109370557942995,2.130196321905472,0.8423424197267729,3.2221098071632515
213,3.3474671532685667,1.0,1.4430978242294037,0.7210858785795746,0.6167806918612416,1.1354883780985652
231,2.715373578526802,1.0,1.202123224395051,1.466287997177493,0.719786490350531,0.8239197097234264
312,3.626023189150076,1.0,3.9089234928017627,1.5479861186171788,0.4792704450614844,3.012922179322869
321,1.9536516296868098,1.0,2.078826068722575,0.884160284928987,0.40747550470671307,0.6107053707593626
}
\end{figure*}
\newpage

\begin{figure*}[h]
    \centering
    \caption{chicago-crime-comm results normalized by splatt (serial)}

\pgfplotstabletypeset[
    /pgfplots/colormap={whiteblue}{rgb255(0cm)=(49,130,189); rgb255(.14cm)=(255,255,255); rgb255(1cm)=(222,45,38)},
    colorCell/.style={
        color cells={min=0,max=7}
    },
    itemCell/.style={
        string type,
        column name={},
    },
    col sep=comma,
    columns={permutation,qsort,splatt,1-sadilla,2-sadilla,3-sadilla,quesadilla,radix},
    columns/permutation/.style={itemCell},
   columns/qsort/.style={colorCell},
   columns/splatt/.style={colorCell},
   columns/1-sadilla/.style={colorCell},
   columns/2-sadilla/.style={colorCell},
   columns/3-sadilla/.style={colorCell},
   columns/quesadilla/.style={colorCell},
   columns/radix/.style={colorCell},
]{
permutation,qsort,splatt,1-sadilla,2-sadilla,3-sadilla,quesadilla,radix
1234,3.1976466847735745,1.0,1.4258156902856063,0.95135644912045,0.3827921190754883,0,3.068091069292038
1243,2.4564342127112555,1.0,1.3714967046252762,1.023854435060034,1.2197001718940208,0.8948057617339407,1.3200244606351845
1324,2.1683147601170685,1.0,1.2276086516135656,1.0100053034412007,0.8165761448430394,0.6112269633540226,1.622833929373905
1342,2.1874574739808432,1.0,1.3301672016457409,1.1822513622568374,1.4677765219125654,1.2425027501278287,1.5454971223298792
1423,2.0760260705807037,1.0,1.3089450718637794,0.8608093327610542,0.6222182124127069,0.3775072255464369,1.1005290771413236
1432,2.3291671569442083,1.0,1.5129629106746088,1.2591277397984149,0.9467686111269339,0.8196666656587215,1.4485966131014354
2134,3.7920988443003356,1.0,2.1463726719935274,1.0884991811379374,0.6737261808748363,0.38176148354610684,2.2436064757619567
2143,2.90533269702557,1.0,1.8832795238137243,1.073967173904364,1.0184969834866713,0.8475217507953102,1.1357981669588462
2314,1.7716845961869447,1.0,1.674471349490963,0.8242824924704937,0.398644413721194,0.2754275176956146,0.7339664505180196
2341,1.7629232099032357,1.0,1.6483681040717102,1.2134408494905797,0.886611680486777,0.5328372857112155,0.6032568938573112
2413,1.9388977023176532,1.0,1.745043524266026,0.8479631578377923,0.40237096539010575,0.2364853962769206,0.5704605118118364
2431,1.905556841421606,1.0,1.7574697761279816,1.33016139039026,0.6419175462570924,0.3075825180749486,0.5109857491270338
3124,4.38139704051233,1.0,2.0659201431679275,0.9891559060752684,0.7084923272641118,0.38677786111858525,2.165019180224749
3142,3.4223404179242873,1.0,1.9009253125287582,1.0192468923091262,1.3946270032824988,1.1368271029836101,1.6248687492048273
3214,2.1502873161957186,1.0,1.6349538989062316,1.0332624201421583,0.5986380218962094,0.5377367216088144,0.7532483528680484
3241,2.073792087081117,1.0,1.6175708542554765,1.4657359260706933,0.903872441693425,0.5473247950029515,0.6188523851421626
3412,2.189533209069813,1.0,1.7635406850325723,1.1285515133063668,0.6758325671451797,0.525765226129921,0.7674701880830583
3421,2.1104155697989246,1.0,1.7161917369285384,1.5427642786953741,0.8270676386586583,0.5394537514433208,0.6097250620604543
4123,4.69711631110951,1.0,2.2236639159923532,1.1330188095860807,0.7405716699757852,0.3490247543987543,1.7346874905695386
4132,3.334955133126036,1.0,1.9627719853335959,1.1434957656592406,1.5929965726155029,1.3670459476680474,1.5181666526366255
4213,2.2772011396322434,1.0,1.6954985759682513,0.9425226246021431,0.44459790005960265,0.25993480996918866,0.6202864395154787
4231,2.1021236630049387,1.0,1.620670383029775,1.4145338581260813,0.6585135472121106,0.31383101151862647,0.5241825406597985
4312,2.146310166991488,1.0,1.7713745049406153,0.7896548632540964,0.36733598598202616,0.25022956816946434,0.6861352695129934
4321,2.1098907742538695,1.0,1.7621511142313588,1.2305343025055886,0.8267300109684402,0.504465896438609,0.5300997631953648
}
\end{figure*}
\newpage

\begin{figure*}[h]
    \centering
    \caption{delicious-4d results normalized by splatt (serial)}

\pgfplotstabletypeset[
    /pgfplots/colormap={whiteblue}{rgb255(0cm)=(49,130,189); rgb255(.14cm)=(255,255,255); rgb255(1cm)=(222,45,38)},
    colorCell/.style={
        color cells={min=0,max=7}
    },
    itemCell/.style={
        string type,
        column name={},
    },
    col sep=comma,
    columns={permutation,qsort,splatt,1-sadilla,2-sadilla,3-sadilla,quesadilla,radix},
    columns/permutation/.style={itemCell},
   columns/qsort/.style={colorCell},
   columns/splatt/.style={colorCell},
   columns/1-sadilla/.style={colorCell},
   columns/2-sadilla/.style={colorCell},
   columns/3-sadilla/.style={colorCell},
   columns/quesadilla/.style={colorCell},
   columns/radix/.style={colorCell},
]{
permutation,qsort,splatt,1-sadilla,2-sadilla,3-sadilla,quesadilla,radix
1234,3.5108895051952156,1.0,1.4345966244646944,0.46970788202243813,0.3133899354689956,0,5.435636937944005
1243,3.6417362827278286,1.0,1.4926419144932337,0.4873025763895327,2.275019768009702,1.8351497381430724,4.898154442177464
1324,2.6052157776708813,1.0,1.4887503816426804,1.5405619700692197,1.3106972384171505,1.1520346835255675,2.3602297262110357
1342,2.58476165228415,1.0,1.5027529090010854,1.5946837912609997,2.0480577685392323,1.8137953338576902,2.5999649436993355
1423,2.584081551189462,1.0,1.304241193026742,1.2644568439739654,1.0736105715778754,0.7677350311248077,2.931736762717067
1432,2.5248319130293004,1.0,1.3841067223127936,1.434345621607874,2.349383585883486,2.1239629859586313,2.851505952318251
2134,2.6830750184608103,1.0,1.3522952277238793,1.0956790221293204,0.9885202473942645,0.8655991759179348,2.0743927804684583
2143,2.5717175146158096,1.0,1.2436542290246504,1.051873906552075,1.8066933441433093,1.6009759155399284,2.0295684551262623
2314,2.2156878493493872,1.0,1.285711503107075,1.3813773217062586,1.2926282408269043,1.1811744443869288,1.70243447082115
2341,2.235549229028286,1.0,1.2910239896683007,1.4827827607894717,1.4709075569253305,1.350093815507335,1.590435285720308
2413,2.2330863097412985,1.0,1.2397333574706704,0.8790283073529616,0.8509709230581759,0.653023224739785,1.4492654555748428
2431,2.193750711665083,1.0,1.2987114498500545,0.911904272366357,1.375583784195196,1.2565093094475757,1.4740589663809054
3124,3.9490263818896074,1.0,1.5796669321221153,0.8967386938136088,0.7105930958734156,0.5530533934498854,2.2734564895213496
3142,3.4619315760635505,1.0,1.511383069393987,0.8928711638803569,2.081898105968778,1.8705212587490732,2.208795176478188
3214,2.8498802719911187,1.0,1.6716731250335684,1.3002344874200333,1.2192745380242145,1.0977906223962621,1.4395427376205867
3241,2.7959040642650903,1.0,1.7007272738631307,1.3760322424055758,1.1537420115029537,1.0448816450333274,1.2496145767142777
3412,2.4987813612320697,1.0,1.5322598374276912,0.844858490811282,0.66660789045976,0.5429555063081688,1.6174324863135379
3421,2.651239871256537,1.0,1.6826710747900036,1.1330066862875703,1.355149325717893,1.2398982768851678,1.2934016905060925
4123,4.79764546334516,1.0,2.051744120210087,0.9929964590091077,0.731710534707777,0.38455560212033896,4.058591509093747
4132,3.7477820791796836,1.0,1.83618156318595,0.9955874365713853,2.4372346253728416,2.2452762126244403,3.005911525389104
4213,2.6727460828849923,1.0,1.813344140541863,1.2571627047298297,1.222647368761608,0.9531710986622647,1.7264515808352694
4231,2.6896535005762447,1.0,1.8472422030623434,1.3421998317500026,1.69176434248051,1.5429286684654573,1.759332604615614
4312,2.536008725941168,1.0,1.8094332859871611,0.9451492425005042,0.7646629236616126,0.6128035359994584,1.7409769280971343
4321,2.6883848677915436,1.0,1.9638679277549234,1.2341315700110207,1.4972337597929666,1.3772833396607116,1.407870848668861
}
\end{figure*}
\newpage

\begin{figure*}[h]
    \centering
    \caption{enron results normalized by splatt (serial)}

\pgfplotstabletypeset[
    /pgfplots/colormap={whiteblue}{rgb255(0cm)=(49,130,189); rgb255(.14cm)=(255,255,255); rgb255(1cm)=(222,45,38)},
    colorCell/.style={
        color cells={min=0,max=7}
    },
    itemCell/.style={
        string type,
        column name={},
    },
    col sep=comma,
    columns={permutation,qsort,splatt,1-sadilla,2-sadilla,3-sadilla,quesadilla,radix},
    columns/permutation/.style={itemCell},
   columns/qsort/.style={colorCell},
   columns/splatt/.style={colorCell},
   columns/1-sadilla/.style={colorCell},
   columns/2-sadilla/.style={colorCell},
   columns/3-sadilla/.style={colorCell},
   columns/quesadilla/.style={colorCell},
   columns/radix/.style={colorCell},
]{
permutation,qsort,splatt,1-sadilla,2-sadilla,3-sadilla,quesadilla,radix
1234,3.0562530255716127,1.0,1.9588887294070685,1.3019659469007152,0.3685114662059793,0,2.825755288336641
1243,2.170827613101148,1.0,1.6089027907715325,1.1891129267705554,0.8070602978029452,0.3575218878212739,1.2681690949875475
1324,1.8844301334099944,1.0,1.3806028015248706,1.0395109795432484,0.814380562029772,0.6475747115661483,1.179617637421644
1342,1.8535981414844727,1.0,1.3865959148920157,1.2172328761445377,1.2512589995411358,0.9672455172681724,1.0319318619851463
1423,1.9085098277287325,1.0,1.5175311320157643,0.8175215421661141,0.6294702415078142,0.24715777575750442,1.0128011652464421
1432,1.8058962348881313,1.0,1.461332534876219,1.0147026687032246,1.0197403611343185,0.7749221677983509,0.9873502217076032
2134,4.077448445837368,1.0,2.0229565208189713,1.6393623804053068,0.7265348233972142,0.3919190221597186,2.6159205824352454
2143,2.7615450884710033,1.0,1.6577443516939665,1.4295567540313763,1.091142575357119,0.6258350854522994,1.2603306366777967
2314,2.3019592828836197,1.0,1.5397451712504713,1.2452815922718532,1.1505360533053188,0.9172293692292763,1.2535781941128332
2341,2.443648782123402,1.0,1.6393907960393692,1.4239652873637865,1.4376511524651772,1.2315606704885311,1.4556032781881423
2413,2.388668547440412,1.0,1.6476018450189902,0.7726202154190028,0.7381549846495172,0.30104881210345136,1.011933544261549
2431,2.5033952740074845,1.0,1.726447243528132,0.8576111147608396,1.1790707941486012,1.0959672715541044,1.294242444554984
3124,3.1139336113387466,1.0,1.4331050848550828,1.0356211845276329,0.7218638412339082,0.47893403009613467,1.2200818369117825
3142,2.483381443013241,1.0,1.349378843545252,1.0611231337440814,1.3820222794434591,1.1027829855557407,0.9174923527912571
3214,2.150989215033235,1.0,1.3130524624409177,0.85277111043045,0.7756957679468972,0.5910643690580645,0.777281656216916
3241,2.2315560817230593,1.0,1.3799846473479553,0.9553956120126715,0.8612140441746633,0.754038856676912,0.9353346539264044
3412,2.0081456016819135,1.0,1.3443106317549802,0.816502370296037,0.803722154698627,0.4972186737890743,0.7935009928721942
3421,2.1366300875511803,1.0,1.4440857245838121,0.9880237323560846,0.8849058312306628,0.7118640622530099,0.7962173530958342
4123,5.529888216977918,1.0,2.2287430845021308,1.7745120674832595,1.3030747584091793,0.38806088668693733,2.4638820214639456
4132,2.8503618323436095,1.0,1.4740594733418204,1.1961953532196494,1.3454025924456297,1.0699342792731412,1.1756381692290345
4213,3.555866938771978,1.0,1.774428475684921,1.1541860279205267,1.099973543380238,0.47159181957483903,1.4335269540294109
4231,3.423092714837422,1.0,1.7738741393465,1.1710656770508718,1.6054470675208343,1.470586289599146,1.6779085070539685
4312,2.4153989157278817,1.0,1.4604017856512184,0.9048053998344221,0.8597865767678622,0.5362081553546858,0.8903096352320077
4321,2.6264721637191846,1.0,1.5809949318188359,1.0627274518899419,1.0099396488744596,0.9592276753200136,1.0699473256585126
}
\end{figure*}
\newpage

\begin{figure*}[h]
    \centering
    \caption{flickr-4d results normalized by splatt (serial)}

\pgfplotstabletypeset[
    /pgfplots/colormap={whiteblue}{rgb255(0cm)=(49,130,189); rgb255(.14cm)=(255,255,255); rgb255(1cm)=(222,45,38)},
    colorCell/.style={
        color cells={min=0,max=7}
    },
    itemCell/.style={
        string type,
        column name={},
    },
    col sep=comma,
    columns={permutation,qsort,splatt,1-sadilla,2-sadilla,3-sadilla,quesadilla,radix},
    columns/permutation/.style={itemCell},
   columns/qsort/.style={colorCell},
   columns/splatt/.style={colorCell},
   columns/1-sadilla/.style={colorCell},
   columns/2-sadilla/.style={colorCell},
   columns/3-sadilla/.style={colorCell},
   columns/quesadilla/.style={colorCell},
   columns/radix/.style={colorCell},
]{
permutation,qsort,splatt,1-sadilla,2-sadilla,3-sadilla,quesadilla,radix
1234,3.446014878191377,1.0,1.46949124365551,0.47121592388816635,0.28431116445689225,0,4.925387379515189
1243,3.5103608723674893,1.0,1.5410183701698128,0.48345647746048953,1.4424393710901222,0.904951837882455,3.491611803250785
1324,2.4434639703258516,1.0,1.397405644812246,1.266514570428233,0.99370189576,0.8948957671675333,1.6124912577981894
1342,2.289410451932883,1.0,1.3058881031313612,1.2244325389087103,1.5112418676716233,1.2250666776214085,1.6445806964303904
1423,3.3292813503083467,1.0,1.469475351320005,1.764453481111478,1.181681254094343,0.7821116937461507,4.447644040053039
1432,2.337924491704075,1.0,1.2752715318834884,1.4412292301266747,1.8344269549450691,1.534456713374551,2.084321856957324
2134,4.205907013103996,1.0,1.2756354594681065,1.3197023641964751,1.11935599354324,0.9188292369206528,2.534117600835197
2143,3.963172334328921,1.0,1.2639477541229884,1.2570617872453553,1.5481607154886659,1.1664211509833329,2.5109685401693826
2314,4.013845179510334,1.0,1.3015031183262926,2.738018413308847,2.741809109327273,2.5577556284326333,3.133875383233012
2341,4.2156729751233435,1.0,1.3883902581773078,3.074650900604542,3.4377253954629103,3.058709585382596,3.3140982877810097
2413,3.899994359223211,1.0,1.2222104803443279,1.152753501348814,1.2091936603002897,0.8531238476636367,2.0238727241756185
2431,4.0083641744659735,1.0,1.3380895841543452,1.2609926270157583,2.0825258223741665,1.855381843010508,2.109643243930558
3124,4.633016770871379,1.0,1.721376648506508,1.1904449417861784,0.791809744158967,0.6126957876326563,2.0186773437667207
3142,4.492032792046977,1.0,1.6638647633469348,1.185380533310996,2.162872955391569,1.7141295349090888,2.1440236660557614
3214,3.6606041138449323,1.0,1.7518565363955045,1.3113168532981765,1.3602861094901333,1.1922230810357954,1.5854461116137193
3241,3.610356795860356,1.0,1.733689205102516,1.3340839412172503,1.282456355312306,1.1573103121806219,1.3732453539402687
3412,3.0613887822622488,1.0,1.5404999362237697,1.030258903460607,0.9361404342526579,0.6218511836646082,1.50784759409594
3421,3.1897094440887415,1.0,1.6247190623320038,1.166799177921628,1.3604354317146226,1.1483568175383483,1.2835353952413253
4123,4.679546327885619,1.0,2.296294684169742,1.2711304323367991,0.7813600103197591,0.37608923664987093,3.9262698741256647
4132,3.172415950106986,1.0,1.7921268906060643,1.117563004779363,1.6846594455681667,1.3715887833647105,1.7920061823992652
4213,2.723657761855313,1.0,1.9634627965714815,0.9668223280494701,0.9146758699627273,0.7282278014328546,1.4886133205137881
4231,2.7716637305415497,1.0,1.952055058385084,0.955014315869169,1.9646443158701261,1.8450463697622281,2.0090890179659766
4312,2.361445778983605,1.0,1.7650030654843134,0.8453356752692535,0.7650275410313219,0.5346988316875239,1.2426823680082972
4321,2.4910073159458834,1.0,1.8507096581658797,0.989116119006847,1.1365569845850716,1.0725989964133071,1.116183187122781
}
\end{figure*}
\newpage

\begin{figure*}[h]
    \centering
    \caption{nips results normalized by splatt (serial)}

\pgfplotstabletypeset[
    /pgfplots/colormap={whiteblue}{rgb255(0cm)=(49,130,189); rgb255(.14cm)=(255,255,255); rgb255(1cm)=(222,45,38)},
    colorCell/.style={
        color cells={min=0,max=7}
    },
    itemCell/.style={
        string type,
        column name={},
    },
    col sep=comma,
    columns={permutation,qsort,splatt,1-sadilla,2-sadilla,3-sadilla,quesadilla,radix},
    columns/permutation/.style={itemCell},
   columns/qsort/.style={colorCell},
   columns/splatt/.style={colorCell},
   columns/1-sadilla/.style={colorCell},
   columns/2-sadilla/.style={colorCell},
   columns/3-sadilla/.style={colorCell},
   columns/quesadilla/.style={colorCell},
   columns/radix/.style={colorCell},
]{
permutation,qsort,splatt,1-sadilla,2-sadilla,3-sadilla,quesadilla,radix
1234,3.8804191972811677,1.0,1.8886990772327386,1.648376767160887,0.3439526009571065,0,3.775384028446028
1243,3.8545296782265766,1.0,1.8958339059747933,1.6611509185232782,2.405436597227166,0.7808454388198861,3.4769743356594645
1324,0.8522501731188465,1.0,0.4364291096441487,0.7648112124068469,0.7296406122453565,0.6597551332342211,0.8828307153826487
1342,0.8587431921864221,1.0,0.4428683648270861,0.7615685084831805,0.9237492096201558,0.8231018770618305,0.8768366567763175
1423,3.7893365258202585,1.0,1.877215008117433,2.5960114842822617,2.3715944667336837,0.7707533984550996,3.5203447885786536
1432,0.6092841167330937,1.0,0.3149449981768034,0.4120032921371232,0.5998015345794044,0.5265799115043501,0.6015695717157804
2134,5.66529578656033,1.0,2.489078224813099,2.209827798597575,0.9570137803141545,0.6202602109586818,4.006862326224126
2143,5.608832591024253,1.0,2.4816738577107182,2.2049470995689235,2.698211011471152,1.1470869457524162,3.3005603187002532
2314,3.6785166067107884,1.0,1.8875951853002362,2.0457968082381384,2.035645282143099,1.596901129771632,2.534709922707694
2341,3.658883486025812,1.0,1.883316983082713,2.0566786363741683,2.5421977676298075,2.287705227309942,2.2022432573392323
2413,5.486437662730772,1.0,2.4246027065447424,2.5074622845945784,2.4403327333937757,0.907214528539722,3.3294208061114543
2431,4.3621703663242135,1.0,2.0285859487677924,1.9933135186677025,2.5895970628602516,2.315812680189273,2.9835421248524714
3124,2.892345199368458,1.0,1.348082129184643,0.8158935591284515,0.7274319477265717,0.5669286398283989,1.254522517182456
3142,2.94810902419376,1.0,1.403411075748156,0.8452759611070192,1.2320261875885588,0.9744787760132442,1.2878491868924797
3214,2.415552063996835,1.0,1.5793973615587498,0.8569406813754897,0.8425315747250947,0.7116177868759666,0.9197281175372503
3241,2.3169485652092403,1.0,1.572901500159453,0.7765297487066176,0.9378157469329013,0.8061614172432867,0.9110380448778019
3412,2.7862163452404345,1.0,1.2806152506015114,1.3383793177946761,1.0913110916372883,0.8704096931246152,1.1657038558528205
3421,2.3710694516648525,1.0,1.4100484528536,1.468074835538494,0.9843799029875618,0.8707538291411544,0.9863941528686287
4123,3.438101326186365,1.0,3.13817518681661,2.153348736472028,1.8886908212516402,0.5070805054154828,3.180987370131586
4132,1.8194278672867348,1.0,1.7313165591537902,1.1450597313349196,1.8614875687105124,1.641230577190973,1.8290544022073647
4213,2.551078836488903,1.0,2.4100173631300885,1.3570891198628057,1.279407018278913,0.4848048960395144,1.750736900185622
4231,2.463604113732609,1.0,2.331809543985169,1.320789510486256,1.7258571760638612,1.554736467236403,1.9698956867308968
4312,1.878123953138074,1.0,1.7318548183384082,0.8460678807467656,0.6662664489856613,0.49614952998216205,0.876029433329986
4321,2.0573054089411404,1.0,1.99792157040889,1.2720572085128024,0.9567469637910053,0.8542461156281047,0.9626916416387754
}
\end{figure*}
\newpage

\begin{figure*}[h]
    \centering
    \caption{uber results normalized by splatt (serial)}

\pgfplotstabletypeset[
    /pgfplots/colormap={whiteblue}{rgb255(0cm)=(49,130,189); rgb255(.14cm)=(255,255,255); rgb255(1cm)=(222,45,38)},
    colorCell/.style={
        color cells={min=0,max=7}
    },
    itemCell/.style={
        string type,
        column name={},
    },
    col sep=comma,
    columns={permutation,qsort,splatt,1-sadilla,2-sadilla,3-sadilla,quesadilla,radix},
    columns/permutation/.style={itemCell},
   columns/qsort/.style={colorCell},
   columns/splatt/.style={colorCell},
   columns/1-sadilla/.style={colorCell},
   columns/2-sadilla/.style={colorCell},
   columns/3-sadilla/.style={colorCell},
   columns/quesadilla/.style={colorCell},
   columns/radix/.style={colorCell},
]{
permutation,qsort,splatt,1-sadilla,2-sadilla,3-sadilla,quesadilla,radix
1234,3.408381056023164,1.0,2.1900461849783532,1.5016466624013365,0.6242859177091502,0,2.830397498559087
1243,2.274323792829738,1.0,1.6841750530386626,1.3431920495337193,1.083376981942823,0.7941587222561429,1.4563011659703091
1324,1.9502406566204686,1.0,1.437103880197562,1.0583219496864014,0.7879373526261811,0.46023651312674885,1.5025307456843757
1342,2.0627300568361684,1.0,1.4686916110506887,1.3558852548701237,1.3525024508441263,1.136128759964867,1.6945881048744382
1423,1.9503655294125166,1.0,1.5171176536135735,0.7738768324776416,0.6044744824404811,0.3447662033903455,1.1228795974090988
1432,2.117163617082595,1.0,1.6626760999865569,1.1276749781200972,1.106809219514517,0.9094145362766501,1.3956070983291433
2134,4.364908075626544,1.0,2.7904481901817464,1.8540984433461187,1.0329880161029361,0.4575394737732074,2.5430154686692585
2143,2.8215257430715277,1.0,2.055490888992739,1.5406889091387692,1.4983369836925158,1.2046175090061115,1.439622043883531
2314,1.9259653512358097,1.0,1.7277700832180856,1.006842093865004,0.6281182723636568,0.37211617988770135,1.001717232978123
2341,1.8774797041950053,1.0,1.7435125017672986,1.3943915440239822,1.0127462369401112,0.7091345041808749,1.0760703198299033
2413,1.9376339765999697,1.0,1.7802060637038521,0.7666466613267279,0.4847471097963886,0.26198316322060017,0.746579777174146
2431,2.0722539602101784,1.0,1.8805385355330906,1.2055020515493435,0.8712455214727812,0.613011969114627,0.9205260780754253
3124,3.9573568905557783,1.0,2.0222155073150048,1.3255847033769586,0.8826410109244112,0.4564196555257724,2.2505474117034536
3142,2.70404498389076,1.0,1.7560741545391203,1.265195921182968,1.4821764552835857,1.3764449630922677,1.7412495665792154
3214,2.2236173237145076,1.0,1.5574960485651752,1.241864567766444,0.8865316245928181,0.6414671171978582,0.8504669373973035
3241,2.0575526873054923,1.0,1.547808950679026,1.5728718015926535,0.9686828339573645,0.678474785001636,1.024823418837233
3412,1.90019098025828,1.0,1.5924478649135982,0.9982261424435387,0.7236869951039662,0.5506568517140668,0.9651330745589696
3421,1.982138235025894,1.0,1.673304395077928,1.369912310257012,1.0990819205140283,0.7380992560040355,0.9157729025779264
4123,4.819944175725723,1.0,2.043232989785601,1.3360924034414494,0.9289151101788624,0.4608747783929369,2.297980693018896
4132,3.1335826159507367,1.0,1.7443816118728113,1.2615585995047094,1.6059325289746975,1.3528424609903857,1.7542796606754925
4213,2.5997615962112754,1.0,1.606792713160671,1.2569302587688949,0.9005202878319772,0.6454886724018143,0.8451119290687203
4231,2.438053357778273,1.0,1.5424754211457425,1.5486132644217987,0.9561962572971848,0.5992827753394195,1.008376590916741
4312,2.2619138021545053,1.0,1.6077665005796142,0.9913078709473546,0.7208349882403992,0.55049527578152,0.9567291569932938
4321,2.337353597130801,1.0,1.6548480425565406,1.336861029454954,1.0750336038460269,0.8229993535521809,0.9098393148389542
}
\end{figure*}
\newpage

\begin{figure*}[h]
    \centering
    \caption{lbnl-network results normalized by splatt (serial) (1)}

\pgfplotstabletypeset[
    /pgfplots/colormap={whiteblue}{rgb255(0cm)=(49,130,189); rgb255(.14cm)=(255,255,255); rgb255(1cm)=(222,45,38)},
    colorCell/.style={
        color cells={min=0,max=7}
    },
    itemCell/.style={
        string type,
        column name={},
    },
    col sep=comma,
    columns={permutation,qsort,splatt,1-sadilla,2-sadilla,3-sadilla,4-sadilla,quesadilla,radix},
    columns/permutation/.style={itemCell},
   columns/qsort/.style={colorCell},
   columns/splatt/.style={colorCell},
   columns/1-sadilla/.style={colorCell},
   columns/2-sadilla/.style={colorCell},
   columns/3-sadilla/.style={colorCell},
   columns/4-sadilla/.style={colorCell},
   columns/quesadilla/.style={colorCell},
   columns/radix/.style={colorCell},
]{
permutation,qsort,splatt,1-sadilla,2-sadilla,3-sadilla,4-sadilla,quesadilla,radix
12345,2.01391860236142,1.0,1.4165512698176752,1.327886636724859,1.2110159498157325,1.1370346823933577,0,1.461404301280929
12354,1.7599123059633253,1.0,1.2457636232240443,1.1519075931222864,1.067424100864231,0.8412729456157428,0.699845956744554,1.3385125523772798
12435,1.6533339918764052,1.0,1.1409386568442546,1.0846173164178006,1.263163295383066,1.240228072791269,0.3976275994055949,0.9978240500994098
12453,1.6929105596816403,1.0,1.2119433069122285,1.1253959253515786,1.3502396852551728,1.072985192965527,0.9484707311392984,1.1875990721047702
12534,2.1172808631768842,1.0,1.4952710603899564,1.4013498869500949,0.9342853886099958,0.8678031792530976,0.699896111664656,1.4795771302347818
12543,2.080081267808821,1.0,1.5248530556587234,1.3988247273661172,0.9367916859192195,1.334929136170434,1.1803409975010408,1.519847809199744
13245,1.4560698633965423,1.0,0.9998749101540091,1.211242546697939,1.1610370882251444,1.087519296230802,0.32932336894566305,0.9746074857628332
13254,1.4577325130579248,1.0,1.040876764160145,1.2463404900162693,1.192183815361609,0.9031278478382051,0.7964130795798272,0.9464141713464529
13425,1.4444242979702309,1.0,0.986069614693971,1.1994929869276245,1.3700135969778005,1.3213873360572106,0.579471638320313,0.961219091006764
13452,1.412925302693152,1.0,1.0172613595557694,1.230008201803851,1.4225254301337187,1.1310380392280799,1.0236867270188121,1.047746640090796
13524,1.8444177008011626,1.0,1.3060659312437546,1.6070135220756554,1.0378803970337314,1.0506814655151275,0.9882508854965093,1.348346551635287
13542,1.876957816544909,1.0,1.3046581693433985,1.6056393537009463,1.1224741644060396,1.4616100518471724,1.330127147702644,1.2454581206298136
14235,1.5098633254442104,1.0,1.0561325258570564,1.1924842043777695,1.1161598058164095,1.1062600326364547,0.34932147352637455,0.8915078558426837
14253,1.5750668779748087,1.0,1.1484037031775087,1.264495266543412,1.2073349534307294,0.997591388512258,0.7993214195955148,1.0653196392192925
14325,1.491526557341368,1.0,1.0389243465163354,1.1719347205381059,1.3101096110637702,1.2809166620329289,0.5590703150421997,0.9305049774767943
14352,1.471136619185277,1.0,1.0609847199027649,1.1829056467217547,1.365658820325478,1.091850909814702,0.993846308512422,1.0142981684947636
14523,1.9334043692155314,1.0,1.3853149423908424,1.5109565208641698,1.1404570318606648,1.1359683505643317,0.9972203760100312,1.3316513396508587
14532,1.9462036290580145,1.0,1.3829122284052953,1.5347199834131033,1.1323837990487224,1.4461424123890083,1.30793485113691,1.2757769959473801
15234,2.278494833510805,1.0,1.6598715669024857,0.9537416209286296,0.9526017625806014,0.9456970091342717,0.788447399708945,1.594473219236835
15243,2.274081389041734,1.0,1.6314612649263176,0.9241355590152535,0.9402772367867722,1.4360946831532195,1.2760339370726843,1.5845801471476
15324,2.2507370405060216,1.0,1.6274655604903951,0.9526657542541306,1.3056237663270125,1.2909760374029398,1.1431433841703504,1.5257724368843035
15342,2.246997647947215,1.0,1.632144905957692,0.9468822004035101,1.3069385296742313,1.6958052416252871,1.5252608346880991,1.5346629765713131
15423,2.2630358240965105,1.0,1.6383658744688079,0.9529805093981962,1.3513952130015014,1.3558192211433406,1.1501601908015058,1.5819597706603574
15432,2.2584578057798708,1.0,1.6236963254014543,0.9494136687706323,1.3445193889229043,1.6932030215514535,1.540909861462205,1.5732935703831645
}
\end{figure*}
\newpage

\begin{figure*}[h]
    \centering
    \caption{lbnl-network results normalized by splatt (serial) (2)}

\pgfplotstabletypeset[
    /pgfplots/colormap={whiteblue}{rgb255(0cm)=(49,130,189); rgb255(.14cm)=(255,255,255); rgb255(1cm)=(222,45,38)},
    colorCell/.style={
        color cells={min=0,max=7}
    },
    itemCell/.style={
        string type,
        column name={},
    },
    col sep=comma,
    columns={permutation,qsort,splatt,1-sadilla,2-sadilla,3-sadilla,4-sadilla,quesadilla,radix},
    columns/permutation/.style={itemCell},
   columns/qsort/.style={colorCell},
   columns/splatt/.style={colorCell},
   columns/1-sadilla/.style={colorCell},
   columns/2-sadilla/.style={colorCell},
   columns/3-sadilla/.style={colorCell},
   columns/4-sadilla/.style={colorCell},
   columns/quesadilla/.style={colorCell},
   columns/radix/.style={colorCell},
]{
permutation,qsort,splatt,1-sadilla,2-sadilla,3-sadilla,4-sadilla,quesadilla,radix
21345,2.583849198232619,1.0,1.8451205984288368,1.6397531108439642,1.5271404049872428,1.429924117490066,0.345133946539443,1.4114253784236601
21354,2.1726733986722793,1.0,1.5857217476918875,1.4069518187179606,1.3158522135444248,1.0351567816783394,0.9029446954679744,1.2497753420199211
21435,2.214327079337619,1.0,1.5555656066957686,1.4070984043294694,1.466687785452204,1.4743371061330026,0.5785584379104652,1.143409836454095
21453,2.2763498290147117,1.0,1.6702525061122335,1.4929821095000113,1.5734953838977703,1.3596546845542974,1.2259276569562658,1.2386150411673447
21534,2.6680852507835486,1.0,1.973509234278741,1.7378671465824689,1.2115121339861505,1.2075548095029545,1.0546925444756416,1.4417543637131494
21543,2.6688508539651683,1.0,1.9539422295160096,1.7370637218539944,1.1270569109079163,1.4418468505346553,1.4196082191571346,1.4835049154604905
23145,1.8079590473644518,1.0,1.4271949463335953,1.2116378943796233,1.1790413607250927,1.113500084271444,0.372269357048342,0.9469989444930322
23154,1.7461918636112188,1.0,1.4156486352763564,1.2346623978475642,1.2077721216739334,0.9679822741147682,0.8645281047475768,1.0026283597151722
23415,1.7187045674517736,1.0,1.342119521405405,1.1766852346394188,1.2374605517137633,1.2278506311722173,0.5060960628292147,0.9156025170813126
23451,1.785997396379555,1.0,1.4543357508748673,1.2491102597239605,1.3255602776004158,0.9470154542786954,0.8434777366242641,0.9476930501039932
23514,2.120455500549783,1.0,1.715681128672113,1.462629542206692,0.8750474026137663,0.8736116215339051,0.8102762920155391,1.1950398674976535
23541,2.1314609420528443,1.0,1.7228727168986755,1.5016505762079133,0.9482340244145713,1.2027768712342173,1.0764987480070982,1.2374778776169888
24135,1.738093850077166,1.0,1.4043334378167953,1.1526491468372264,1.1233254311416503,1.0785525082611933,0.36913460478483784,0.9163635919311865
24153,1.9056014496973512,1.0,1.5010923854953393,1.231086003912168,1.1837155036283353,1.047894542652866,0.9422408229089598,1.0081438006710401
24315,1.7062135853145914,1.0,1.3807336111047464,1.1641795376283748,1.2065981031165438,1.1973426091460244,0.4909623298986254,0.8914705703107612
24351,1.845782374215359,1.0,1.4520165296090601,1.2417192442623743,1.3008932304568002,0.925535300622693,0.8306020683969114,0.8811747214702095
24513,2.2489627439979296,1.0,1.809711999030777,1.5420264358517477,0.9874713669864444,0.9910499615615408,0.8611302560809869,1.2174244965241303
24531,2.2649259361469833,1.0,1.8234296440180378,1.528566703784454,0.9882907589781043,1.1808476329701416,1.0473314709232193,1.2112495923337347
25134,2.5593303001410423,1.0,2.0566602124463422,0.8973445002991883,0.903530881687185,0.9005288907188083,0.7547301693607414,1.412166951432963
25143,2.6072750869148282,1.0,2.094367388883201,0.9091883505663296,0.9045202226597273,1.3759173738416166,1.1572235773949509,1.3809814930225146
25314,2.5634790045440234,1.0,2.13249990228484,0.9122350022434123,1.1218317807790084,1.1339707386649451,0.9842082664683844,1.402879387377397
25341,2.5711139845866047,1.0,2.117181919668476,0.9056632378662994,1.119888466637681,1.3616134217188194,1.2099118021978124,1.4052352528742573
25413,2.559235623473203,1.0,2.0971879841827605,0.903786247062024,1.1984923511842764,1.1966072212387076,1.0367756774326213,1.4454367419487557
25431,2.6476115434658105,1.0,2.06715490364637,0.89367832008742,1.190710175716135,1.4102899102706066,1.2399958112885605,1.446661779836654
}
\end{figure*}
\newpage

\begin{figure*}[h]
    \centering
    \caption{lbnl-network results normalized by splatt (serial) (3)}

\pgfplotstabletypeset[
    /pgfplots/colormap={whiteblue}{rgb255(0cm)=(49,130,189); rgb255(.14cm)=(255,255,255); rgb255(1cm)=(222,45,38)},
    colorCell/.style={
        color cells={min=0,max=7}
    },
    itemCell/.style={
        string type,
        column name={},
    },
    col sep=comma,
    columns={permutation,qsort,splatt,1-sadilla,2-sadilla,3-sadilla,4-sadilla,quesadilla,radix},
    columns/permutation/.style={itemCell},
   columns/qsort/.style={colorCell},
   columns/splatt/.style={colorCell},
   columns/1-sadilla/.style={colorCell},
   columns/2-sadilla/.style={colorCell},
   columns/3-sadilla/.style={colorCell},
   columns/4-sadilla/.style={colorCell},
   columns/quesadilla/.style={colorCell},
   columns/radix/.style={colorCell},
]{
permutation,qsort,splatt,1-sadilla,2-sadilla,3-sadilla,4-sadilla,quesadilla,radix
31245,2.6825907227976624,1.0,1.7404489606939204,1.615396810196022,1.545325763135086,1.4433183687561588,0.3588151276509298,1.3514516038482827
31254,2.2781994725380175,1.0,1.5180474015952408,1.4247755194179224,1.3437880955471007,0.9912127241194243,0.863391222735645,1.2936898043999403
31425,2.0471711104018246,1.0,1.353902725168795,1.2582534678139174,1.457689688901078,1.4069969496852184,0.5572587197252622,1.081042874298063
31452,2.1983735047123343,1.0,1.435917958848451,1.3299511483508364,1.5815541655635879,1.2604398971520863,1.132669031391962,1.225280136371323
31524,2.673229889598049,1.0,1.8075770405990668,1.660918415838093,1.1520183658222187,1.1529674078288688,1.002299080589799,1.5276729181044386
31542,2.7036197049673714,1.0,1.8254508568139491,1.655293859581501,1.1574414655561127,1.4325950178053712,1.2633331251158801,1.5165526554231132
32145,2.0584643566222995,1.0,1.4596075849303398,1.4202484460254605,1.3544477931849412,1.2856758464969436,0.43357660057648345,1.090320779292537
32154,1.8986473296879312,1.0,1.3199702525609174,1.286315312755917,1.2524461152538233,0.9745483883364008,0.8619806905557195,1.0630532985973467
32415,1.7944225176693758,1.0,1.2784231631943581,1.2307759481474374,1.3140481339148702,1.2959569176056642,0.5402279107356319,0.9589817184741589
32451,1.975459853639945,1.0,1.3688258218571523,1.3262547112528458,1.4069251049722908,0.9288060130030367,0.8931058281183223,1.0611228650480828
32514,2.3088090571517377,1.0,1.6583671491474024,1.5949593763881353,0.9283314574406045,0.9323608284967428,0.8730646588559845,1.2359410447518269
32541,2.339644129103429,1.0,1.6430647842843822,1.6008886259130928,1.0023839514473352,1.2701792174340496,1.1395525706815648,1.2040593024873927
34125,1.837326286339287,1.0,1.2643521806597033,1.2382175963789692,1.1708449785137693,1.132722291393813,0.3886506534089802,0.9571550950625217
34152,1.9145380442100979,1.0,1.3349049878898591,1.3522313278045313,1.2787829576919123,1.0053836713645223,0.8997312756989995,1.0877372623169759
34215,1.661999907479663,1.0,1.192464176932536,1.2343584913796726,1.2111110051627958,1.198544664398925,0.5010028636763996,0.8947107641808312
34251,1.7948700362411174,1.0,1.2918691907824582,1.29090438929569,1.266301951383104,0.886795696230759,0.8347494270159843,0.9825973032025872
34512,2.339518269536442,1.0,1.654891571604527,1.7074807718926508,1.025624108022937,1.028394142085138,0.8945690301474231,1.2829955022610715
34521,2.324582229754408,1.0,1.6593394878616763,1.7321569146658,1.0197395898114088,1.2790897999875381,1.1378754151703645,1.3015290070389596
35124,2.7908278379964284,1.0,1.9607602692216353,0.9617206717217134,0.9666732160230708,0.9594843011422827,0.8072769195062313,1.5383871232053798
35142,2.7163795150983114,1.0,1.9759513331637357,0.9729145540478117,0.9661376628401489,1.4332881131141668,1.2650928240165515,1.5121080708052081
35214,2.694319578051065,1.0,1.9734419514601016,0.9001209837375205,1.20673880078856,1.2743139612538499,1.1138839582426103,1.5712688882431705
35241,2.7652501717489497,1.0,1.9955648883675876,0.9684515475916741,1.264237579402557,1.5173612272183121,1.3633445795724695,1.4162806958641592
35412,2.8059453344706915,1.0,1.9668687480409799,0.9725949169922081,1.2753016170231926,1.2760671384061064,1.1251942752160602,1.4344160508463124
35421,2.739193364876181,1.0,1.9574694384542646,0.9014886611828812,1.2742346426825772,1.4986532282451233,1.3549941032405277,1.5571855381782118
}
\end{figure*}
\newpage

\begin{figure*}[h]
    \centering
    \caption{lbnl-network results normalized by splatt (serial) (4)}

\pgfplotstabletypeset[
    /pgfplots/colormap={whiteblue}{rgb255(0cm)=(49,130,189); rgb255(.14cm)=(255,255,255); rgb255(1cm)=(222,45,38)},
    colorCell/.style={
        color cells={min=0,max=7}
    },
    itemCell/.style={
        string type,
        column name={},
    },
    col sep=comma,
    columns={permutation,qsort,splatt,1-sadilla,2-sadilla,3-sadilla,4-sadilla,quesadilla,radix},
    columns/permutation/.style={itemCell},
   columns/qsort/.style={colorCell},
   columns/splatt/.style={colorCell},
   columns/1-sadilla/.style={colorCell},
   columns/2-sadilla/.style={colorCell},
   columns/3-sadilla/.style={colorCell},
   columns/4-sadilla/.style={colorCell},
   columns/quesadilla/.style={colorCell},
   columns/radix/.style={colorCell},
]{
permutation,qsort,splatt,1-sadilla,2-sadilla,3-sadilla,4-sadilla,quesadilla,radix
41235,2.824507968293154,1.0,1.879344074130988,1.6201331231141851,1.5055998749695985,1.483052293227944,0.3643104352150288,1.4446680928777558
41253,2.857669444170361,1.0,1.8950421358000624,1.6555646446807615,1.540919143447017,1.3396430501590075,1.105320764975417,1.5261309725737449
41325,2.2929932357263456,1.0,1.558623742703124,1.3300192540387679,1.5815545626232266,1.5351350261413381,0.6213630904360449,1.1495460301887381
41352,2.311833474310466,1.0,1.6027456211186581,1.3982027169486555,1.6937989553074522,1.3962134285624193,1.2791191078142055,1.2757546737650007
41523,2.988132552931998,1.0,2.0319909179513393,1.7637991639294126,1.3823664787111565,1.3617144272151502,1.206651364979249,1.6535109926041676
41532,3.0033175426665553,1.0,2.0507645402401073,1.7894645859954188,1.3820373501828909,1.8415995778628562,1.6616704883809643,1.5879731722047168
42135,1.811340628974031,1.0,1.3737628155970667,1.153564535592839,1.1020396527719967,1.0982529024040677,0.3699108413056014,0.9107109557052523
42153,1.958623965294565,1.0,1.5140200026656283,1.2267973916983972,1.19476319471892,1.0439750233765148,0.9380073247284065,1.0089703328533461
42315,1.785418974260384,1.0,1.3577391890607862,1.1636323541355187,1.2187128508767557,1.2211850455240285,0.49582335845702386,0.8935945328139157
42351,1.8582534784512437,1.0,1.4839642240849213,1.2596433571089938,1.3217573925270423,0.9444857968704602,0.8409334848339282,0.9986668044564175
42513,2.412384902254625,1.0,1.887152002733649,1.602954194216837,1.0343902467931738,1.0306799644460516,0.8989894893050927,1.277333082217787
42531,2.4376537632804687,1.0,1.93277497608889,1.6041721056560327,1.019651221291078,1.2273403472001467,1.1094615932936638,1.2761818667398939
43125,1.6818666244033567,1.0,1.3349472316163988,1.1276081493538208,1.0681223892039622,1.0360769573126118,0.34624765659588613,0.8670487216980513
43152,1.776540460201094,1.0,1.395285234001877,1.1929046783264428,1.1366090710304708,0.9457813970682271,0.8469766550862367,0.9747623655745294
43215,1.744879940129672,1.0,1.359658121629155,1.2434470071249513,1.1849441808718948,1.216715769223228,0.5053807937017057,0.911742288915512
43251,1.9104972759325338,1.0,1.5036270345451788,1.3490228623037588,1.3450894392197097,0.8792993451438003,0.803749370825032,0.9987580561549521
43512,2.245670692490824,1.0,1.7391645184606173,1.5960994081257451,0.9588534152456951,0.9593960381412625,0.8319700948322054,1.2042804589515153
43521,2.2483963181618876,1.0,1.7329586013075575,1.5658338085965404,0.9499775686977446,1.188675969921392,1.0654277679677817,1.2329822203661338
45123,2.823841342761025,1.0,2.1669419763699955,0.9684748077059225,0.923913232164855,0.9688788843793091,0.8198715703012888,1.4175979486415937
45132,2.751833037584299,1.0,2.1989524899519317,0.9729369396773782,0.9736415556247013,1.461316781511105,1.3057748358979853,1.4910483348745915
45213,2.7558630346933692,1.0,2.2084654711464755,0.9444768044503995,1.1624611527175723,1.2836367912459667,1.1043303660057837,1.5506722959981463
45231,2.8338472804661774,1.0,2.1790505930887782,0.9842873577646116,1.267949979298954,1.5004704484432323,1.3403729224595269,1.5378885627991659
45312,2.860911538009007,1.0,2.1635489003116795,0.9826273275177101,1.2213881210996778,1.2094194867520016,1.0608925690291338,1.5146155833286032
45321,2.7578216793125145,1.0,2.1738276245086414,0.9740575806075419,1.2088147370490423,1.4111421745491723,1.261301488154651,1.4900423446776898
}
\end{figure*}
\newpage

\begin{figure*}[h]
    \centering
    \caption{lbnl-network results normalized by splatt (serial) (5)}

\pgfplotstabletypeset[
    /pgfplots/colormap={whiteblue}{rgb255(0cm)=(49,130,189); rgb255(.14cm)=(255,255,255); rgb255(1cm)=(222,45,38)},
    colorCell/.style={
        color cells={min=0,max=7}
    },
    itemCell/.style={
        string type,
        column name={},
    },
    col sep=comma,
    columns={permutation,qsort,splatt,1-sadilla,2-sadilla,3-sadilla,4-sadilla,quesadilla,radix},
    columns/permutation/.style={itemCell},
   columns/qsort/.style={colorCell},
   columns/splatt/.style={colorCell},
   columns/1-sadilla/.style={colorCell},
   columns/2-sadilla/.style={colorCell},
   columns/3-sadilla/.style={colorCell},
   columns/4-sadilla/.style={colorCell},
   columns/quesadilla/.style={colorCell},
   columns/radix/.style={colorCell},
]{
permutation,qsort,splatt,1-sadilla,2-sadilla,3-sadilla,4-sadilla,quesadilla,radix
51234,6.0857148701988875,1.0,1.7746924672614357,1.6620425391235114,1.651136315162603,1.6551447496454341,1.255540843384458,3.286916275921468
51243,6.040207238702099,1.0,1.6740701703607619,1.5755885496424442,1.6594307264213493,3.1565556407992488,2.6704195823612142,3.274866034435981
51324,6.053509403662715,1.0,1.7515515102586159,1.6405816987721589,3.046861181600453,3.021842470406447,2.572817899453157,3.199770571755123
51342,6.293648703284307,1.0,1.8642259425237389,1.7260923718526218,3.2282255353744986,4.625910896358661,4.185968188139168,3.064107888459915
51423,6.0637698011397045,1.0,1.7622114683335333,1.6386355921435485,3.111920057463501,3.113531329157647,2.631483884574472,3.189503655811754
51432,6.361807466533736,1.0,1.8689547072862964,1.7368693643236446,2.999658394923651,4.394876864660088,4.203885185224443,3.3978032921698
52134,6.063887067849634,1.0,1.7565381403333642,2.256968983923899,2.243737961928218,2.264887030765874,1.8661739534245199,3.294164565704507
52143,6.042541324304617,1.0,1.7788718689259093,2.275578626861067,2.2738127738301728,3.6703947738521,3.3002061918951284,3.342863661190456
52314,5.924697327740713,1.0,1.7910770682118875,2.281209448866068,2.7500475587456745,2.7532555718509646,2.3361791869404978,3.3099932408155186
52341,5.967420439508786,1.0,1.7760326094065415,2.271050084503967,2.7538028339288996,3.282362344398631,2.855345767211744,3.2846457630005617
52413,5.989668103134909,1.0,1.7615895323736046,2.2484799410385743,2.76733885306129,2.7733661231105273,2.3820045201407174,3.270016599752956
52431,6.022178767750835,1.0,1.7688586161324533,2.2727575040967203,2.8255950908508365,3.257817660948443,2.8801445958713687,3.3040548604420366
53124,6.203798442142541,1.0,1.8328455133327866,2.216394597808586,2.1980079480361727,2.2128422256680906,1.7962507249507986,3.3298401249003446
53142,6.2677504143173035,1.0,1.862981501942173,2.209601966112389,2.2270967781994933,3.6643891848000014,3.2712418798425795,3.313089038384264
53214,6.285864456987681,1.0,1.8358233548171727,2.007376088755608,2.4973174001618252,2.7030770236468973,2.2936309214977606,3.0514368413287376
53241,6.266942411166072,1.0,1.849723938489768,2.2090225744378276,2.7193873002062934,3.2704209007384497,2.847288485223167,2.9753334400138396
53412,6.215533411324844,1.0,1.8334284684570954,2.190283925673067,2.750946994760907,2.6113576368095526,2.3300341420394743,3.3133973795715552
53421,6.114435463286668,1.0,1.8448382568969186,2.2041043591236282,2.749556187224474,3.2427302375078004,2.893535506625743,3.330589936903039
54123,6.206078506363153,1.0,1.8170629683732482,2.309639746629327,2.3170641947494612,2.1601078260383284,1.792014294596508,3.4068847924980887
54132,6.094337277260324,1.0,1.8355709652653638,2.3534738769443893,2.3525827538643127,3.695416123232019,3.28703223809443,3.381437817197765
54213,6.150903848468336,1.0,1.8193038761493299,2.344902991026973,2.836878218343938,2.8193399419029097,2.460744192691502,3.3516354896034914
54231,6.204852810170484,1.0,1.8631985907548259,2.390701662603499,2.896207305060243,3.4008243456701566,3.006684089772224,3.4367855183694114
54312,6.283552190218266,1.0,1.7011861816253977,2.3590703640140247,2.8336423524433068,2.8076508707629904,2.4281609669981474,3.4227766575170593
54321,6.022347336632521,1.0,1.7417529320747431,2.328635356247266,2.7778540019101845,3.144127019333851,2.9060092079778097,3.366631123174466
}
\end{figure*}
\newpage

\begin{figure*}[h]
    \centering
    \caption{vast-2015-mc1-5d results normalized by splatt (serial) (1)}

\pgfplotstabletypeset[
    /pgfplots/colormap={whiteblue}{rgb255(0cm)=(49,130,189); rgb255(.14cm)=(255,255,255); rgb255(1cm)=(222,45,38)},
    colorCell/.style={
        color cells={min=0,max=7}
    },
    itemCell/.style={
        string type,
        column name={},
    },
    col sep=comma,
    columns={permutation,qsort,splatt,1-sadilla,2-sadilla,3-sadilla,4-sadilla,quesadilla,radix},
    columns/permutation/.style={itemCell},
   columns/qsort/.style={colorCell},
   columns/splatt/.style={colorCell},
   columns/1-sadilla/.style={colorCell},
   columns/2-sadilla/.style={colorCell},
   columns/3-sadilla/.style={colorCell},
   columns/4-sadilla/.style={colorCell},
   columns/quesadilla/.style={colorCell},
   columns/radix/.style={colorCell},
]{
permutation,qsort,splatt,1-sadilla,2-sadilla,3-sadilla,4-sadilla,quesadilla,radix
12345,3.3843442377780595,1.0,1.069164122928115,0.21943084482989106,0.21913042138494498,0.21893065930207173,0,4.2390060210146725
12354,3.4056023837681026,1.0,1.0957580407004195,0.22550639043076978,0.225682366600026,2.986653403674526,2.629484628716261,4.125973402344846
12435,3.3915977237877395,1.0,1.0752412660506652,0.23118086425894,2.9891747899545438,2.7636078551687744,2.7909495390084067,4.435121697276655
12453,3.3461441284724436,1.0,1.0676739654827028,0.2183117551792653,2.862841230208687,6.3628543945445095,5.571634570365645,4.975482455522854
12534,3.3594328607210984,1.0,1.0773753473997463,0.228579975913978,2.817441363307049,3.014073677063586,2.7184754022385516,4.391909087265862
12543,3.4075251468961025,1.0,1.0809321571536847,0.2256193834613065,3.0859790505009985,5.956836464656662,5.77825916185535,4.7529888925485775
13245,2.819021525020843,1.0,0.9194826536543713,1.4357607914256811,0.8098318421288008,0.7554613265907952,0.6180952942791008,3.5622230828962467
13254,2.813754620280016,1.0,0.9250800774008493,1.4868163888445927,0.8087254786485659,3.2428309715956134,3.162556457361157,3.5629323058161066
13425,2.1053637634757325,1.0,0.9669967418376041,1.2814540425143532,1.354692797202541,1.3045920249854277,1.122700022607987,2.2293372556894298
13452,2.0133381400301005,1.0,0.9788523626462985,1.2513747507918702,1.4143757489824065,2.286008725870764,2.2128385838807176,1.9027431401061767
13524,2.1893349109724514,1.0,0.9773016132115391,1.3028530097997815,1.3137601490703918,1.3306002343535914,1.0832101831774772,2.282974275946344
13542,2.1023396714404696,1.0,0.961981243925935,1.2242747138453514,1.4015818662277273,2.3030947522366283,2.1813762179360925,1.7823705256685598
14235,2.260875047043239,1.0,1.050742261782374,1.285532472967409,1.078562860403809,1.1694733848228134,0.978823503323214,2.3198246605201316
14253,2.3671145443419097,1.0,1.0554360989083507,1.2554511660769996,1.1447604494149288,2.8290717096417746,2.5890929111112753,2.8400937968690676
14325,2.3152761997311515,1.0,1.0397522924082856,1.2782841677534342,1.777119945237171,1.825152257907728,1.5093482152213606,2.2757212896432435
14352,2.1978715530138917,1.0,1.018187795004704,1.2074951131937242,1.9123615888593242,2.559926111570659,2.2918880152364527,2.0930931524623277
14523,2.1914189149309418,1.0,1.0230281093491644,1.297486314678102,2.299943217754324,2.2736249642980795,2.0041117260471113,2.357632113187461
14532,2.1870353310456916,1.0,1.0131120331313699,1.2933558103559395,2.235780324902834,2.936386217968347,2.6275188237115046,1.9012926101546215
15234,2.3846912444455572,1.0,1.0632728739795245,1.2947988012893954,1.1552607991823116,1.1801124068118396,0.9778136956663916,2.462469997214165
15243,2.4045463837039067,1.0,1.0774081390393049,1.3083724885585526,1.1121259553441714,2.5909936274377756,2.734710053659321,2.8827384785879393
15324,2.2610258948872493,1.0,1.038043932174891,1.2882844236866111,1.855186490508194,1.8598626382992187,1.6319185064745656,2.463827647051129
15342,2.1912780535945386,1.0,1.0185253099190974,1.2076657905803154,1.9228197944442107,2.4835508429572433,2.3120712946881863,2.10560404581453
15423,2.217316491638593,1.0,1.022836890768547,1.311958458184235,2.3922777116338265,2.275455342146441,2.1624420448849078,2.1505964073946626
15432,2.225363142004888,1.0,1.0239577640937931,1.3059219362680319,2.31753061618023,2.840917768202389,2.6505393453888897,2.152281268999997
}
\end{figure*}
\newpage

\begin{figure*}[h]
    \centering
    \caption{vast-2015-mc1-5d results normalized by splatt (serial) (2)}

\pgfplotstabletypeset[
    /pgfplots/colormap={whiteblue}{rgb255(0cm)=(49,130,189); rgb255(.14cm)=(255,255,255); rgb255(1cm)=(222,45,38)},
    colorCell/.style={
        color cells={min=0,max=7}
    },
    itemCell/.style={
        string type,
        column name={},
    },
    col sep=comma,
    columns={permutation,qsort,splatt,1-sadilla,2-sadilla,3-sadilla,4-sadilla,quesadilla,radix},
    columns/permutation/.style={itemCell},
   columns/qsort/.style={colorCell},
   columns/splatt/.style={colorCell},
   columns/1-sadilla/.style={colorCell},
   columns/2-sadilla/.style={colorCell},
   columns/3-sadilla/.style={colorCell},
   columns/4-sadilla/.style={colorCell},
   columns/quesadilla/.style={colorCell},
   columns/radix/.style={colorCell},
]{
permutation,qsort,splatt,1-sadilla,2-sadilla,3-sadilla,4-sadilla,quesadilla,radix
21345,2.3232226952437083,1.0,0.9892384736867091,0.4881513931495676,0.5258140569122445,0.5163001776760491,0.4436479598463229,1.5432183956985701
21354,2.3385716334791824,1.0,0.9599321170450984,0.4985049944173832,0.5157352827903374,1.776012686037648,1.6504651154021202,1.4566115552894132
21435,2.434457274727189,1.0,0.9741486772083485,0.5365516255804154,1.7405008772790949,1.8587656432619666,1.669861338363356,1.496673416455204
21453,2.2980022642682396,1.0,0.9933491051147822,0.5057215361578447,1.7071968839184213,2.929990950992137,3.0292140929069755,1.8325345303389493
21534,2.2332770313409256,1.0,0.97961744002596,0.5051771217285225,1.6723324979009198,1.822201029343647,1.6561081669127742,1.4489363772030799
21543,2.4053457044975888,1.0,0.9922381137061713,0.5359929446955538,1.7910815245868037,3.0244567889937106,3.1515384835701306,1.729355713211947
23145,1.9521572048143516,1.0,0.8141886516269562,1.034039865422907,0.6063452930495967,0.634000738590002,0.5879667333210641,1.1656753942170632
23154,1.9317036122327165,1.0,0.7976247858345632,1.0267066146690904,0.5891349196791209,1.647574282762305,1.514506611486336,1.2440929298481942
23415,1.5222817443650476,1.0,0.9122815341703304,1.0456415359675018,0.6507381571700227,0.5640384998628939,0.512411960754183,1.0859417844340435
23451,1.5159361721442681,1.0,0.9123000855921075,1.008222720231181,0.7894374499892278,0.6737964338492276,0.5690132081659274,0.7142143606894308
23514,1.4822531772048413,1.0,0.8472221566568292,0.9837840929090594,0.6486750906011773,0.5632224297190235,0.47090816381815465,0.9961383496016516
23541,1.4810809700359708,1.0,0.8840515556475183,0.9779543067974756,0.7895166572159116,0.679750859939899,0.5499087709499866,0.7140697187076062
24135,1.5781495890832278,1.0,0.9693040203018704,0.6921701253945597,0.5336259447117042,0.5380634785871923,0.47030824191260306,1.16275606029061
24153,1.6658520393145286,1.0,0.9598073135813847,0.6865113673646991,0.5318227678133274,1.375518898897962,1.2354872161447605,1.281253291689248
24315,1.5107555409016538,1.0,0.9326067333151875,0.6453303019702835,0.8961903507468396,0.7554828126581296,0.6862687515630074,1.14498281874299
24351,1.5188215973359211,1.0,0.9523875094314724,0.7711178619348311,0.9482597710499349,0.6663925601398774,0.5638645460548083,0.8098141739050866
24513,1.5239412834304942,1.0,0.9644138896430281,0.7678236933373913,0.6535584808376097,0.6129770943673681,0.5490469258953036,0.7738421628776894
24531,1.5106506723713482,1.0,0.9656287712109949,0.7833855261363968,0.6135240780217471,0.8078769552653071,0.7826591863504513,0.8197433401523248
25134,1.5626926755500745,1.0,0.9269970379575625,0.647825315914631,0.5478398323929855,0.5463676580445045,0.48696338354704494,1.0796433474811697
25143,1.6664757579773817,1.0,0.9485103612627199,0.6914748524211856,0.5198101703772084,1.3297471638654697,1.1756442146277821,1.3815279432800096
25314,1.516809031348117,1.0,0.9122145938753521,0.6775779077466851,0.9273720165798324,0.7239811013315204,0.7041726735481653,1.0415317331420015
25341,1.500625516188713,1.0,0.9211678474642847,0.7624407946762665,0.949892407602736,0.683378720441187,0.5575482590519446,0.8158960902177362
25413,1.5249786965751493,1.0,0.9240082953908317,0.8155600700870711,0.6549248111631968,0.6120655008575492,0.503795921054233,0.8404715780762264
25431,1.5554496965924698,1.0,0.9112303031040012,0.8004757108724713,0.6197276368157278,0.8046167221724736,0.70246609870948,0.8254819000473103
}
\end{figure*}
\newpage

\begin{figure*}[h]
    \centering
    \caption{vast-2015-mc1-5d results normalized by splatt (serial) (3)}

\pgfplotstabletypeset[
    /pgfplots/colormap={whiteblue}{rgb255(0cm)=(49,130,189); rgb255(.14cm)=(255,255,255); rgb255(1cm)=(222,45,38)},
    colorCell/.style={
        color cells={min=0,max=7}
    },
    itemCell/.style={
        string type,
        column name={},
    },
    col sep=comma,
    columns={permutation,qsort,splatt,1-sadilla,2-sadilla,3-sadilla,4-sadilla,quesadilla,radix},
    columns/permutation/.style={itemCell},
   columns/qsort/.style={colorCell},
   columns/splatt/.style={colorCell},
   columns/1-sadilla/.style={colorCell},
   columns/2-sadilla/.style={colorCell},
   columns/3-sadilla/.style={colorCell},
   columns/4-sadilla/.style={colorCell},
   columns/quesadilla/.style={colorCell},
   columns/radix/.style={colorCell},
]{
permutation,qsort,splatt,1-sadilla,2-sadilla,3-sadilla,4-sadilla,quesadilla,radix
31245,2.5889586659264157,1.0,2.882403277912659,1.1495861714695177,0.5407381527261427,0.5439734166658498,0.37647927778134765,2.9398372825976042
31254,2.6399519300834524,1.0,2.8638991203051667,1.137955746951899,0.5450794147225154,2.911262689927003,2.75579913617184,3.2155880132512973
31425,1.7895054756608988,1.0,1.9461671341988491,0.9603672010128961,1.2999601332502582,1.2284506095596563,1.1535651850005728,1.7661323267067646
31452,1.735670656511812,1.0,1.8714627209223673,0.9557360555006605,1.3191348736394668,2.0996772327364788,1.906896367901554,1.498077762146112
31524,1.8080170124263055,1.0,1.8880605387152265,0.9479697111662397,1.36197182448048,1.2294081105467654,1.1333758399791845,1.8645693178542315
31542,1.717485883378404,1.0,1.8563633578317107,0.9330347136333301,1.2932467294827255,2.0349879556591643,1.9144648554890171,1.606525838553673
32145,1.418641280875389,1.0,1.4923004935018787,0.6651070672136107,0.37987130535127245,0.38662627731317967,0.3139977488236982,0.9362949860233929
32154,1.4787940028186095,1.0,1.4743349366070426,0.6486558681640348,0.37018026997072767,1.1410612169780985,0.98889817498873,0.8404095418233121
32415,1.3419664163138534,1.0,1.41593844179733,0.8683057183909889,0.6129128245907719,0.513381378148717,0.42542021071081904,1.0017335181329097
32451,1.3909263011919202,1.0,1.393074591726185,0.8685594013371536,0.7417461317398677,0.6063163807052856,0.5234658218784,0.6683515491587837
32514,1.348424712050376,1.0,1.4109855461079852,0.8425598563975284,0.6326018354057016,0.4727850062154513,0.4660978635929804,0.9408277421537454
32541,1.3550338821373615,1.0,1.3791940647099192,0.8115778003613979,0.7289850378332023,0.6028813073875139,0.46517349117475587,0.6375295749824089
34125,1.4758093564596713,1.0,1.5013589764917572,0.6778090502186901,0.2805984357019954,0.24121641268648722,0.18744656016512012,0.9480567702888661
34152,1.423893038377847,1.0,1.5056281933748628,0.7050615413058777,0.3121312831925181,0.6177682778251181,0.5300474675379169,1.0233450756131601
34215,1.470444132968977,1.0,1.5777187550638985,1.0974736408092465,0.6271499694343944,0.5231549630926204,0.49971489566022426,0.8939907856848441
34251,1.4731912459295946,1.0,1.4659003237803532,1.058019046800634,0.7219779216721379,0.6333593526199105,0.5744595915559751,0.7002990486003315
34512,1.2538489535373238,1.0,1.360010095278669,0.9538071870723835,0.5615341465518361,0.28222951149734704,0.21225972610787516,0.7397269136892943
34521,1.4074888926774123,1.0,1.4418516355666784,1.0597195709467409,0.8637613758085217,0.5779267474384056,0.49000994646937396,0.6231504817203812
35124,1.466817853235681,1.0,1.509739732475799,0.6681779562654263,0.28046896109867014,0.22803395005023483,0.1841857209533472,0.9811804446390147
35142,1.3806420092302443,1.0,1.492613925858096,0.6963204371696683,0.3117868367377927,0.557300058463412,0.5019117112640061,1.0105793432982555
35214,1.4640513527031722,1.0,1.567053076415696,1.117761600028405,0.6054010703968136,0.5269399115930488,0.4977910762306682,0.9062478015069214
35241,1.4946788264130524,1.0,1.5350127096795354,1.0746150739425542,0.734083395523843,0.6498795261293604,0.600860619415093,0.7221297281538964
35412,1.275826324999039,1.0,1.3374740869070267,0.9420750405680337,0.555546485498857,0.27766815414317747,0.20399850890562804,0.7430587445009338
35421,1.3762586938651353,1.0,1.416907315977643,1.0576184535044904,0.8732610299529339,0.5748152248666778,0.4496986411238198,0.6183198653746171
}
\end{figure*}
\newpage

\begin{figure*}[h]
    \centering
    \caption{vast-2015-mc1-5d results normalized by splatt (serial) (4)}

\pgfplotstabletypeset[
    /pgfplots/colormap={whiteblue}{rgb255(0cm)=(49,130,189); rgb255(.14cm)=(255,255,255); rgb255(1cm)=(222,45,38)},
    colorCell/.style={
        color cells={min=0,max=7}
    },
    itemCell/.style={
        string type,
        column name={},
    },
    col sep=comma,
    columns={permutation,qsort,splatt,1-sadilla,2-sadilla,3-sadilla,4-sadilla,quesadilla,radix},
    columns/permutation/.style={itemCell},
   columns/qsort/.style={colorCell},
   columns/splatt/.style={colorCell},
   columns/1-sadilla/.style={colorCell},
   columns/2-sadilla/.style={colorCell},
   columns/3-sadilla/.style={colorCell},
   columns/4-sadilla/.style={colorCell},
   columns/quesadilla/.style={colorCell},
   columns/radix/.style={colorCell},
]{
permutation,qsort,splatt,1-sadilla,2-sadilla,3-sadilla,4-sadilla,quesadilla,radix
41235,5.2962301472861135,1.0,2.280360640632782,0.8582884403889479,0.6647030192298069,0.6803440156168123,0.473648585452974,3.3797136905668563
41253,5.022075836805406,1.0,2.2166170083775123,0.8307589475708498,0.6456587439039015,1.9802383464526574,1.7383157730466483,4.152231768703865
41325,4.8982390064550465,1.0,2.1530242166578923,0.799000361677591,1.1630498353823575,0.9809883566066786,0.8351124234309402,3.264519884810396
41352,3.823535743127946,1.0,1.8692815450786044,0.7214900911622698,1.0198933913477881,1.8023193562887792,1.474815722119161,2.6273726377053723
41523,4.0176965570308125,1.0,1.8920403340799368,0.7507531364726968,1.4875462578893717,1.5271580646050669,1.2813139628484442,3.2480963442555537
41532,3.8946136036692858,1.0,1.8543499898194271,0.7127517885617347,1.4150097025509873,1.7925568053224918,1.4612373807861758,2.7205675594736034
42135,2.1899188317737717,1.0,1.5662506660582671,0.8201244452302465,0.7202439321885955,0.6880624946329561,0.5886810486492474,1.3367343393354574
42153,2.1107809481601976,1.0,1.5584608718779693,0.869816267393878,0.7194498977677009,1.1713978154179088,1.2134628594589627,1.5272092761935272
42315,2.0546241655861692,1.0,1.4451191572350384,0.7895916521458322,0.9796527609890776,0.8597956972696345,0.7979432480059863,1.2522529170059642
42351,1.9214412840720665,1.0,1.4771809041068484,0.8977519562607861,1.1360594973491882,0.9280594098597649,0.7899778784357553,1.025606533408568
42513,2.0445608791875998,1.0,1.4710954546838022,0.9705378521190582,0.8016760901972658,0.7599789996694752,0.693999836893334,1.0666304037883017
42531,1.9742782976943878,1.0,1.4946025586428773,0.9300071770057163,0.8841468938344614,1.0856587242720819,0.9297372079038816,1.0563508019899177
43125,3.4402224439386098,1.0,1.531237176644155,1.986016621140785,0.9671633643613089,0.8559604802575222,0.8092000128424517,2.4079150559375546
43152,2.958533990080588,1.0,1.4590361763235824,1.8031818021100567,0.9636274595445615,1.476190802057342,1.4097862958737974,1.9834762719246972
43215,1.856857539192424,1.0,1.332710120670281,1.5741880934455752,0.8036306657387899,0.6224151803259678,0.5687684061881342,1.095384506960545
43251,1.8488154599370583,1.0,1.311697492535137,1.4933829789955761,0.9059029118679424,0.8161482952671404,0.6978924261861551,0.8595217225827796
43512,1.667087952187137,1.0,1.2197788906198441,1.3957916352907422,0.708147284866045,0.3504054068323382,0.2733527021141928,0.9020632584067106
43521,1.7010009816092084,1.0,1.2714460628174307,1.4452141204694549,1.084276548519756,0.6691238543909271,0.600568960347492,0.7339148248868057
45123,1.9115948238359237,1.0,1.4333303726497424,0.7765032772613798,0.36006901922085777,0.32975115844413133,0.2507324203699633,1.2770283353932235
45132,1.9576536254020223,1.0,1.4195092170972061,0.787757456694948,0.35997077567113467,0.4855307224963009,0.4031204503219152,1.1380675051169382
45213,1.9643387767368343,1.0,1.504750149639427,1.1964559547393865,0.7621539606239401,0.6650674084941302,0.6466753198950101,0.8549669450751236
45231,1.899067624806426,1.0,1.435338596066535,1.198843413058977,0.6892401688368396,0.9288846730840844,0.7785104159654307,0.8636910316950215
45312,1.68801146714841,1.0,1.2522463719379897,0.6734533065965953,0.9005663496908299,0.49119024734164723,0.4487478520716968,0.9503901044498102
45321,1.652983018126132,1.0,1.3313780114145333,1.0500427697054682,1.2399974605595228,0.7057663496010093,0.6146226247260493,0.6812299316152512
}
\end{figure*}
\newpage

\begin{figure*}[h]
    \caption{vast-2015-mc1-5d results normalized by splatt (serial) (5)}

\pgfplotstabletypeset[
    /pgfplots/colormap={whiteblue}{rgb255(0cm)=(49,130,189); rgb255(.14cm)=(255,255,255); rgb255(1cm)=(222,45,38)},
    colorCell/.style={
        color cells={min=0,max=7}
    },
    itemCell/.style={
        string type,
        column name={},
    },
    col sep=comma,
    columns={permutation,qsort,splatt,1-sadilla,2-sadilla,3-sadilla,4-sadilla,quesadilla,radix},
    columns/permutation/.style={itemCell},
   columns/qsort/.style={colorCell},
   columns/splatt/.style={colorCell},
   columns/1-sadilla/.style={colorCell},
   columns/2-sadilla/.style={colorCell},
   columns/3-sadilla/.style={colorCell},
   columns/4-sadilla/.style={colorCell},
   columns/quesadilla/.style={colorCell},
   columns/radix/.style={colorCell},
]{
permutation,qsort,splatt,1-sadilla,2-sadilla,3-sadilla,4-sadilla,quesadilla,radix
51234,5.169260302530314,1.0,2.2568273975678284,0.8191135602712816,0.6584734946907789,0.6677930876949392,0.4608325690072408,3.6534552682850734
51243,5.0775725550247035,1.0,2.2422415254338435,0.8262930343087223,0.6479030241288942,1.989877323496277,1.7804824591424115,3.7595490953896458
51324,5.033441502196777,1.0,2.242038976497081,0.8083495260212649,1.193311017270865,1.0505292034538127,0.7981159083742105,3.4775703083180463
51342,3.8759673792264833,1.0,1.8924216210429827,0.7074597426361794,1.0167831417573774,1.782577227416217,1.4395330808816285,2.63904211643688
51423,3.9444363345882465,1.0,1.8370106259737211,0.7277404341126555,1.4185793461786143,1.367622271336178,1.3212965280662845,2.8450833232143604
51432,3.913229800091442,1.0,1.8338463433739145,0.7197626034923091,1.5019040857612695,1.7778616175338666,1.5606162726037813,2.7290113487179077
52134,2.2207443482748754,1.0,1.5310763528439715,0.8393136137903483,0.7168554073792942,0.6590076081402187,0.6585903453260353,1.3228474568507544
52143,2.100445558936297,1.0,1.5529773910348492,0.7964647704993397,0.6811493300026449,1.1900147596148152,1.1093827338226039,1.5063426384176881
52314,2.075902223747845,1.0,1.4714775571354943,0.8414876118926796,1.0100693483765437,0.8337126971969949,0.811257384444805,1.2513855117731756
52341,1.9369966790150248,1.0,1.467114280140631,0.9308526808590551,1.1152649457312116,0.8822809935617516,0.7490339776250197,1.0629013821420707
52413,2.0527906337506883,1.0,1.4874037737237538,0.9735009596648734,0.8653239388756436,0.8563740345606895,0.7596236090829734,1.088965869476564
52431,1.9897993229437465,1.0,1.4456038746405022,0.9166794788563487,0.8234220716141583,1.0386501908723906,0.9030319143058007,1.0012358270446517
53124,3.4499390054591035,1.0,1.5085006610239913,1.988865313760583,1.0464534800992684,0.943483793261218,0.7132696047264867,2.35918751919084
53142,2.975763307992788,1.0,1.4670739890482072,1.8151370581666508,0.9618204492061947,1.5370424848408708,1.4070992273669507,2.1666230198339167
53214,1.8926854622837128,1.0,1.3693082356917643,1.544802612874964,0.805717241969367,0.620305104837778,0.6065560123032214,1.11167933034649
53241,1.8583704287431524,1.0,1.2970603937752239,1.5193063945705876,0.8794168800219363,0.7672892651617459,0.7032450931768224,0.9315038113264388
53412,1.6712066689700518,1.0,1.2141586637265485,1.388501895722888,0.7068756094390442,0.34076866110321713,0.26415411562432983,0.9192575744596831
53421,1.706995449078498,1.0,1.2683367148091982,1.4232411133416962,1.0791144209592045,0.6954732407504813,0.5468911602701801,0.7121248512893043
54123,1.997089366841611,1.0,1.3932922521179727,0.7972783129429287,0.3461256836833597,0.32812733981598424,0.24464289834284955,1.2265649710807285
54132,1.9519703880559793,1.0,1.4030916650918315,0.7610336168474527,0.35273760603985155,0.4833660217147252,0.39820148514930165,1.1724363951879253
54213,1.984906945743631,1.0,1.4753614434385012,1.208552521474772,0.7569478865438124,0.6593431313727972,0.6543311725592783,0.8745385334575765
54231,1.8881808409106235,1.0,1.4491223549181247,1.1905388116767588,0.751659786672428,0.8979456165676781,0.7769440452954152,0.8299622050298826
54312,1.7054846276006557,1.0,1.2455623847344814,0.6617635047427491,0.9150300787733143,0.5278620708420346,0.44211316770382636,0.9372305022866976
54321,1.7088663588233737,1.0,1.3319657568329928,1.0645261664965857,1.2272407163591095,0.6746457265239997,0.5729058303227732,0.7675737514718584
}
\end{figure*}

\clearpage
\newpage

\begin{figure*}[h]
    \centering
    \caption{flickr-3d results normalized by splatt (parallel)}

\pgfplotstabletypeset[
    /pgfplots/colormap={whiteblue}{rgb255(0cm)=(49,130,189); rgb255(.14cm)=(255,255,255); rgb255(1cm)=(222,45,38)},
    colorCell/.style={
        color cells={min=0,max=7}
    },
    itemCell/.style={
        string type,
        column name={},
    },
    col sep=comma,
    columns={permutation,qsort,splatt,1-sadilla,2-sadilla,k-sadilla,radix},
    columns/permutation/.style={itemCell},
   columns/qsort/.style={colorCell},
   columns/splatt/.style={colorCell},
   columns/1-sadilla/.style={colorCell},
   columns/2-sadilla/.style={colorCell},
   columns/k-sadilla/.style={colorCell},
   columns/radix/.style={colorCell},
]{
permutation,qsort,splatt,1-sadilla,2-sadilla,k-sadilla,radix
123,30.163,1.000,0.293,0.274,0.000,9.829
132,34.607,1.000,0.557,2.060,1.819,7.267
213,5.225,1.000,1.076,0.629,0.599,0.964
231,5.266,1.000,1.051,0.955,0.909,1.019
312,40.667,1.000,1.268,1.098,0.897,5.858
321,36.651,1.000,1.295,4.521,4.371,5.181
}
\end{figure*}

\begin{figure*}[h]
    \centering
    \caption{nell-1 results normalized by splatt (parallel)}

\pgfplotstabletypeset[
    /pgfplots/colormap={whiteblue}{rgb255(0cm)=(49,130,189); rgb255(.14cm)=(255,255,255); rgb255(1cm)=(222,45,38)},
    colorCell/.style={
        color cells={min=0,max=7}
    },
    itemCell/.style={
        string type,
        column name={},
    },
    col sep=comma,
    columns={permutation,qsort,splatt,1-sadilla,2-sadilla,k-sadilla,radix},
    columns/permutation/.style={itemCell},
   columns/qsort/.style={colorCell},
   columns/splatt/.style={colorCell},
   columns/1-sadilla/.style={colorCell},
   columns/2-sadilla/.style={colorCell},
   columns/k-sadilla/.style={colorCell},
   columns/radix/.style={colorCell},
]{
permutation,qsort,splatt,1-sadilla,2-sadilla,k-sadilla,radix
123,19.938,1.000,0.345,0.161,0.000,7.076
132,25.938,1.000,0.455,3.694,3.504,6.403
213,33.904,1.000,1.249,1.161,0.907,7.542
231,34.151,1.000,1.344,5.688,5.410,6.260
312,8.074,1.000,1.073,0.746,0.761,1.155
321,8.180,1.000,1.087,0.934,0.858,1.076
}
\end{figure*}

\begin{figure*}[h]
    \centering
    \caption{nell-2 results normalized by splatt (parallel)}

\pgfplotstabletypeset[
    /pgfplots/colormap={whiteblue}{rgb255(0cm)=(49,130,189); rgb255(.14cm)=(255,255,255); rgb255(1cm)=(222,45,38)},
    colorCell/.style={
        color cells={min=0,max=7}
    },
    itemCell/.style={
        string type,
        column name={},
    },
    col sep=comma,
    columns={permutation,qsort,splatt,1-sadilla,2-sadilla,k-sadilla,radix},
    columns/permutation/.style={itemCell},
   columns/qsort/.style={colorCell},
   columns/splatt/.style={colorCell},
   columns/1-sadilla/.style={colorCell},
   columns/2-sadilla/.style={colorCell},
   columns/k-sadilla/.style={colorCell},
   columns/radix/.style={colorCell},
]{
permutation,qsort,splatt,1-sadilla,2-sadilla,k-sadilla,radix
123,28.951,1.000,0.391,0.235,0.000,4.484
132,30.239,1.000,0.740,2.132,1.864,2.926
213,38.843,1.000,1.257,1.134,0.800,3.877
231,28.383,1.000,1.371,1.980,1.735,2.311
312,36.857,1.000,1.333,1.346,1.157,2.458
321,35.146,1.000,1.406,1.687,1.538,2.087
}
\end{figure*}

\begin{figure*}[h]
    \centering
    \caption{vast-2015-mc1-3d results normalized by splatt (parallel)}

\pgfplotstabletypeset[
    /pgfplots/colormap={whiteblue}{rgb255(0cm)=(49,130,189); rgb255(.14cm)=(255,255,255); rgb255(1cm)=(222,45,38)},
    colorCell/.style={
        color cells={min=0,max=7}
    },
    itemCell/.style={
        string type,
        column name={},
    },
    col sep=comma,
    columns={permutation,qsort,splatt,1-sadilla,2-sadilla,k-sadilla,radix},
    columns/permutation/.style={itemCell},
   columns/qsort/.style={colorCell},
   columns/splatt/.style={colorCell},
   columns/1-sadilla/.style={colorCell},
   columns/2-sadilla/.style={colorCell},
   columns/k-sadilla/.style={colorCell},
   columns/radix/.style={colorCell},
]{
permutation,qsort,splatt,1-sadilla,2-sadilla,k-sadilla,radix
123,25.764,1.000,0.314,0.405,0.000,4.630
132,25.514,1.000,0.384,1.678,1.490,4.218
213,35.880,1.000,1.309,1.464,1.095,2.591
231,33.908,1.000,1.387,1.915,1.664,2.304
312,3.620,1.000,1.285,0.142,0.112,0.590
321,2.040,1.000,1.470,0.111,0.098,0.131
}
\end{figure*}
\newpage

\begin{figure*}[h]
    \centering
    \caption{chicago-crime-comm results normalized by splatt (parallel)}

\pgfplotstabletypeset[
    /pgfplots/colormap={whiteblue}{rgb255(0cm)=(49,130,189); rgb255(.14cm)=(255,255,255); rgb255(1cm)=(222,45,38)},
    colorCell/.style={
        color cells={min=0,max=7}
    },
    itemCell/.style={
        string type,
        column name={},
    },
    col sep=comma,
    columns={permutation,qsort,splatt,1-sadilla,2-sadilla,3-sadilla,k-sadilla,radix},
    columns/permutation/.style={itemCell},
   columns/qsort/.style={colorCell},
   columns/splatt/.style={colorCell},
   columns/1-sadilla/.style={colorCell},
   columns/2-sadilla/.style={colorCell},
   columns/3-sadilla/.style={colorCell},
   columns/k-sadilla/.style={colorCell},
   columns/radix/.style={colorCell},
]{
permutation,qsort,splatt,1-sadilla,2-sadilla,3-sadilla,k-sadilla,radix
1234,20.770,1.000,0.332,0.327,0.413,0.000,4.343
1243,22.909,1.000,0.480,0.443,1.901,1.626,4.054
1324,22.956,1.000,0.537,1.745,1.726,1.329,3.879
1342,23.227,1.000,0.533,1.708,3.015,2.461,3.717
1423,23.051,1.000,0.557,1.420,1.472,1.170,3.518
1432,26.742,1.000,0.578,1.695,2.972,2.594,3.561
2134,22.259,1.000,1.236,0.864,0.939,0.599,3.169
2143,19.480,1.000,1.219,0.774,1.720,1.432,2.388
2314,13.315,1.000,1.262,0.844,0.892,0.755,1.505
2341,13.161,1.000,1.301,0.945,1.108,0.966,1.289
2413,15.664,1.000,1.313,0.929,0.920,0.711,1.773
2431,15.227,1.000,1.361,1.031,1.167,0.994,1.320
3124,28.372,1.000,1.135,1.029,1.097,0.682,3.198
3142,23.350,1.000,1.195,0.927,2.083,1.774,2.575
3214,17.077,1.000,1.164,1.001,1.031,0.776,1.665
3241,19.917,1.000,1.387,1.430,1.474,1.293,1.714
3412,19.349,1.000,1.332,1.100,1.077,0.897,1.872
3421,19.859,1.000,1.341,1.364,1.356,1.166,1.531
4123,18.486,1.000,1.233,0.768,0.637,0.388,2.046
4132,10.972,1.000,1.310,0.568,0.779,0.635,1.121
4213,9.827,1.000,1.381,0.642,0.553,0.428,0.980
4231,8.491,1.000,1.292,0.812,0.624,0.503,0.688
4312,8.793,1.000,1.316,0.524,0.458,0.365,0.800
4321,9.135,1.000,1.374,0.806,0.622,0.521,0.687
}
\end{figure*}
\newpage

\begin{figure*}[h]
    \centering
    \caption{delicious-4d results normalized by splatt (parallel)}

\pgfplotstabletypeset[
    /pgfplots/colormap={whiteblue}{rgb255(0cm)=(49,130,189); rgb255(.14cm)=(255,255,255); rgb255(1cm)=(222,45,38)},
    colorCell/.style={
        color cells={min=0,max=7}
    },
    itemCell/.style={
        string type,
        column name={},
    },
    col sep=comma,
    columns={permutation,qsort,splatt,1-sadilla,2-sadilla,3-sadilla,k-sadilla,radix},
    columns/permutation/.style={itemCell},
   columns/qsort/.style={colorCell},
   columns/splatt/.style={colorCell},
   columns/1-sadilla/.style={colorCell},
   columns/2-sadilla/.style={colorCell},
   columns/3-sadilla/.style={colorCell},
   columns/k-sadilla/.style={colorCell},
   columns/radix/.style={colorCell},
]{
permutation,qsort,splatt,1-sadilla,2-sadilla,3-sadilla,k-sadilla,radix
1234,22.000,1.000,0.271,0.238,0.287,0.000,8.430
1243,22.767,1.000,0.285,0.244,2.224,1.978,8.094
1324,27.099,1.000,0.449,2.089,2.102,1.869,6.871
1342,28.567,1.000,0.479,2.169,3.674,3.451,6.935
1423,25.411,1.000,0.428,1.786,1.814,1.590,7.220
1432,26.426,1.000,0.461,1.880,3.770,3.517,7.038
2134,7.868,1.000,1.023,0.670,0.725,0.642,1.607
2143,7.846,1.000,1.037,0.694,1.047,1.075,1.542
2314,8.316,1.000,1.102,1.247,1.201,1.155,1.631
2341,8.172,1.000,1.009,1.205,1.331,1.286,1.506
2413,8.067,1.000,1.040,0.869,0.850,0.810,1.462
2431,7.816,1.000,1.035,0.885,1.281,1.229,1.419
3124,27.343,1.000,1.189,1.005,0.960,0.789,4.365
3142,29.794,1.000,1.218,1.247,2.603,2.408,4.503
3214,29.222,1.000,1.314,2.893,2.819,2.643,4.100
3241,29.220,1.000,1.289,3.015,3.396,3.189,3.800
3412,28.520,1.000,1.336,2.649,1.689,1.458,4.316
3421,27.021,1.000,1.240,3.697,3.239,2.983,3.670
4123,45.127,1.000,1.391,1.229,1.227,0.945,9.086
4132,42.007,1.000,1.296,1.136,3.938,3.620,7.744
4213,35.604,1.000,1.310,3.592,3.612,3.358,6.289
4231,36.139,1.000,1.420,3.672,5.605,5.266,6.032
4312,36.126,1.000,1.311,2.004,2.211,1.893,5.576
4321,38.518,1.000,1.345,2.171,4.624,4.482,5.195
}
\end{figure*}
\newpage

\begin{figure*}[h]
    \centering
    \caption{enron results normalized by splatt (parallel)}

\pgfplotstabletypeset[
    /pgfplots/colormap={whiteblue}{rgb255(0cm)=(49,130,189); rgb255(.14cm)=(255,255,255); rgb255(1cm)=(222,45,38)},
    colorCell/.style={
        color cells={min=0,max=7}
    },
    itemCell/.style={
        string type,
        column name={},
    },
    col sep=comma,
    columns={permutation,qsort,splatt,1-sadilla,2-sadilla,3-sadilla,k-sadilla,radix},
    columns/permutation/.style={itemCell},
   columns/qsort/.style={colorCell},
   columns/splatt/.style={colorCell},
   columns/1-sadilla/.style={colorCell},
   columns/2-sadilla/.style={colorCell},
   columns/3-sadilla/.style={colorCell},
   columns/k-sadilla/.style={colorCell},
   columns/radix/.style={colorCell},
]{
permutation,qsort,splatt,1-sadilla,2-sadilla,3-sadilla,k-sadilla,radix
1234,16.969,1.000,0.742,0.629,0.276,0.000,3.349
1243,11.956,1.000,1.100,0.999,0.810,0.615,1.617
1324,12.330,1.000,1.033,1.305,1.224,1.041,1.844
1342,11.909,1.000,1.062,1.415,2.020,1.923,1.638
1423,11.059,1.000,1.146,0.882,0.803,0.624,1.371
1432,10.938,1.000,1.132,1.278,1.463,1.392,1.326
2134,34.836,1.000,1.260,1.382,1.278,0.901,4.455
2143,35.464,1.000,1.397,1.638,2.394,2.079,3.504
2314,28.284,1.000,1.315,2.388,2.382,2.077,3.809
2341,29.326,1.000,1.297,2.580,3.548,3.143,4.008
2413,31.839,1.000,1.354,1.731,1.669,1.414,3.193
2431,33.919,1.000,1.347,1.968,3.217,2.911,3.681
3124,29.615,1.000,1.242,1.231,1.232,1.004,3.246
3142,27.789,1.000,1.276,1.363,2.335,2.162,2.840
3214,26.360,1.000,1.259,1.966,1.866,1.581,2.862
3241,27.015,1.000,1.284,2.129,2.570,2.320,3.005
3412,24.393,1.000,1.244,1.723,1.568,1.424,2.465
3421,26.190,1.000,1.273,1.984,2.469,2.148,2.850
4123,45.372,1.000,1.340,1.256,1.124,0.902,4.471
4132,35.291,1.000,1.386,1.236,2.476,2.251,3.449
4213,38.305,1.000,1.393,1.834,1.811,1.589,3.560
4231,38.506,1.000,1.345,1.955,3.266,3.023,3.735
4312,32.375,1.000,1.314,1.799,1.791,1.639,2.895
4321,34.322,1.000,1.307,2.009,2.865,2.568,3.202
}
\end{figure*}
\newpage

\begin{figure*}[h]
    \centering
    \caption{flickr-4d results normalized by splatt (parallel)}

\pgfplotstabletypeset[
    /pgfplots/colormap={whiteblue}{rgb255(0cm)=(49,130,189); rgb255(.14cm)=(255,255,255); rgb255(1cm)=(222,45,38)},
    colorCell/.style={
        color cells={min=0,max=7}
    },
    itemCell/.style={
        string type,
        column name={},
    },
    col sep=comma,
    columns={permutation,qsort,splatt,1-sadilla,2-sadilla,3-sadilla,k-sadilla,radix},
    columns/permutation/.style={itemCell},
   columns/qsort/.style={colorCell},
   columns/splatt/.style={colorCell},
   columns/1-sadilla/.style={colorCell},
   columns/2-sadilla/.style={colorCell},
   columns/3-sadilla/.style={colorCell},
   columns/k-sadilla/.style={colorCell},
   columns/radix/.style={colorCell},
]{
permutation,qsort,splatt,1-sadilla,2-sadilla,3-sadilla,k-sadilla,radix
1234,22.941,1.000,0.270,0.235,0.296,0.000,9.649
1243,23.131,1.000,0.283,0.248,2.089,1.817,8.613
1324,26.799,1.000,0.469,1.930,1.967,1.702,6.997
1342,27.102,1.000,0.522,1.903,3.378,3.172,7.131
1423,23.797,1.000,0.308,2.008,2.017,1.787,9.631
1432,24.546,1.000,0.472,1.871,3.568,3.379,7.582
2134,4.424,1.000,1.032,0.589,0.569,0.557,1.027
2143,4.506,1.000,1.028,0.581,0.766,0.748,1.022
2314,4.476,1.000,1.030,0.880,0.910,0.882,1.125
2341,4.402,1.000,1.024,0.909,1.044,1.016,1.141
2413,4.437,1.000,1.056,0.683,0.665,0.642,1.006
2431,4.282,1.000,1.025,0.682,0.874,0.847,0.972
3124,33.111,1.000,1.215,1.124,1.133,0.901,5.688
3142,31.626,1.000,1.171,1.087,2.496,2.285,5.693
3214,29.923,1.000,1.267,4.040,4.216,4.044,5.429
3241,29.428,1.000,1.257,4.102,4.626,4.482,5.076
3412,28.287,1.000,1.329,2.159,1.633,1.452,5.109
3421,29.004,1.000,1.396,2.645,4.596,4.398,5.258
4123,36.727,1.000,1.226,1.119,1.097,0.862,9.055
4132,37.873,1.000,1.330,1.155,3.032,2.996,7.724
4213,31.153,1.000,1.321,4.844,4.702,4.497,6.984
4231,31.862,1.000,1.316,4.814,7.083,6.821,7.577
4312,32.044,1.000,1.317,1.827,1.815,1.639,5.994
4321,31.654,1.000,1.281,1.818,5.266,5.091,5.758
}
\end{figure*}
\newpage

\begin{figure*}[h]
    \centering
    \caption{nips results normalized by splatt (parallel)}

\pgfplotstabletypeset[
    /pgfplots/colormap={whiteblue}{rgb255(0cm)=(49,130,189); rgb255(.14cm)=(255,255,255); rgb255(1cm)=(222,45,38)},
    colorCell/.style={
        color cells={min=0,max=7}
    },
    itemCell/.style={
        string type,
        column name={},
    },
    col sep=comma,
    columns={permutation,qsort,splatt,1-sadilla,2-sadilla,3-sadilla,k-sadilla,radix},
    columns/permutation/.style={itemCell},
   columns/qsort/.style={colorCell},
   columns/splatt/.style={colorCell},
   columns/1-sadilla/.style={colorCell},
   columns/2-sadilla/.style={colorCell},
   columns/3-sadilla/.style={colorCell},
   columns/k-sadilla/.style={colorCell},
   columns/radix/.style={colorCell},
]{
permutation,qsort,splatt,1-sadilla,2-sadilla,3-sadilla,k-sadilla,radix
1234,52.958,1.000,0.782,0.674,0.431,0.000,5.751
1243,50.012,1.000,0.836,0.759,3.067,2.428,5.775
1324,16.705,1.000,1.123,1.179,1.133,1.063,1.627
1342,16.793,1.000,1.125,1.174,1.930,1.906,1.714
1423,50.823,1.000,0.882,3.446,3.229,2.449,5.764
1432,12.269,1.000,1.642,1.439,1.440,1.303,1.097
2134,68.120,1.000,1.405,1.208,1.074,0.905,4.395
2143,66.320,1.000,1.594,1.388,3.340,2.671,4.624
2314,50.479,1.000,1.467,3.148,3.087,2.850,3.351
2341,50.548,1.000,1.333,3.169,3.320,3.002,3.609
2413,66.558,1.000,1.594,2.122,2.026,1.065,4.335
2431,56.642,1.000,1.589,1.902,4.109,3.222,3.806
3124,40.304,1.000,1.429,1.285,1.385,1.187,2.083
3142,39.818,1.000,1.360,1.362,2.649,2.265,2.033
3214,40.345,1.000,1.604,1.444,1.410,1.102,1.763
3241,39.291,1.000,1.387,1.494,1.589,1.423,1.710
3412,41.580,1.000,1.593,1.570,1.553,1.125,2.007
3421,38.933,1.000,1.423,1.525,1.742,1.469,1.887
4123,25.394,1.000,1.836,0.706,0.662,0.574,2.744
4132,14.841,1.000,1.919,1.438,1.119,1.044,1.388
4213,16.520,1.000,1.386,0.559,0.577,0.397,1.280
4231,17.579,1.000,1.563,0.668,1.392,1.296,1.366
4312,16.329,1.000,1.656,0.776,0.799,0.666,0.952
4321,15.332,1.000,1.349,0.775,0.835,0.661,0.773
}
\end{figure*}
\newpage

\begin{figure*}[h]
    \centering
    \caption{uber results normalized by splatt (parallel)}

\pgfplotstabletypeset[
    /pgfplots/colormap={whiteblue}{rgb255(0cm)=(49,130,189); rgb255(.14cm)=(255,255,255); rgb255(1cm)=(222,45,38)},
    colorCell/.style={
        color cells={min=0,max=7}
    },
    itemCell/.style={
        string type,
        column name={},
    },
    col sep=comma,
    columns={permutation,qsort,splatt,1-sadilla,2-sadilla,3-sadilla,k-sadilla,radix},
    columns/permutation/.style={itemCell},
   columns/qsort/.style={colorCell},
   columns/splatt/.style={colorCell},
   columns/1-sadilla/.style={colorCell},
   columns/2-sadilla/.style={colorCell},
   columns/3-sadilla/.style={colorCell},
   columns/k-sadilla/.style={colorCell},
   columns/radix/.style={colorCell},
]{
permutation,qsort,splatt,1-sadilla,2-sadilla,3-sadilla,k-sadilla,radix
1234,54.221,1.000,0.932,0.656,0.554,0.000,3.822
1243,44.925,1.000,1.098,0.787,1.768,1.613,2.056
1324,39.205,1.000,1.119,2.094,1.963,1.611,2.162
1342,41.625,1.000,1.089,2.093,3.221,3.010,1.995
1423,40.578,1.000,1.126,1.563,1.565,1.210,1.673
1432,45.749,1.000,1.134,2.085,2.774,2.630,1.726
2134,37.227,1.000,1.551,0.709,0.626,0.286,1.803
2143,30.164,1.000,2.024,0.687,1.167,1.178,1.146
2314,18.643,1.000,1.509,0.695,0.564,0.515,0.909
2341,18.121,1.000,1.511,1.044,0.812,0.692,0.717
2413,18.277,1.000,1.467,0.563,0.415,0.376,0.681
2431,19.511,1.000,1.416,0.941,0.779,0.589,0.606
3124,59.819,1.000,1.713,2.092,1.563,0.771,2.546
3142,46.188,1.000,1.681,2.201,2.041,2.067,1.742
3214,40.564,1.000,1.634,1.466,0.940,0.670,1.765
3241,38.006,1.000,1.567,2.443,1.608,1.353,1.348
3412,37.550,1.000,1.637,1.040,0.864,0.829,1.259
3421,38.199,1.000,1.660,1.365,1.262,0.972,1.170
4123,86.704,1.000,1.951,3.063,1.905,0.829,3.083
4132,63.178,1.000,1.940,4.086,2.564,2.398,1.996
4213,54.575,1.000,1.728,2.538,1.225,0.765,2.002
4231,52.889,1.000,1.818,4.867,1.983,1.516,1.583
4312,49.200,1.000,1.649,1.506,1.037,0.929,1.318
4321,50.713,1.000,1.761,2.581,1.574,1.093,1.295
}
\end{figure*}
\newpage

\begin{figure*}[h]
    \centering
    \caption{lbnl-network results normalized by splatt (parallel) (1)}

\pgfplotstabletypeset[
    /pgfplots/colormap={whiteblue}{rgb255(0cm)=(49,130,189); rgb255(.14cm)=(255,255,255); rgb255(1cm)=(222,45,38)},
    colorCell/.style={
        color cells={min=0,max=7}
    },
    itemCell/.style={
        string type,
        column name={},
    },
    col sep=comma,
    columns={permutation,qsort,splatt,1-sadilla,2-sadilla,3-sadilla,4-sadilla,k-sadilla,radix},
    columns/permutation/.style={itemCell},
   columns/qsort/.style={colorCell},
   columns/splatt/.style={colorCell},
   columns/1-sadilla/.style={colorCell},
   columns/2-sadilla/.style={colorCell},
   columns/3-sadilla/.style={colorCell},
   columns/4-sadilla/.style={colorCell},
   columns/k-sadilla/.style={colorCell},
   columns/radix/.style={colorCell},
]{
permutation,qsort,splatt,1-sadilla,2-sadilla,3-sadilla,4-sadilla,k-sadilla,radix
12345,5.597,1.000,0.897,0.708,0.611,0.374,0.000,2.571
12354,4.121,1.000,0.835,0.695,0.479,1.202,1.059,1.879
12435,3.965,1.000,0.921,0.734,0.669,0.659,0.342,1.546
12453,3.456,1.000,0.742,0.596,0.535,1.391,1.356,1.455
12534,3.662,1.000,0.745,0.509,1.068,1.130,0.977,1.501
12543,4.186,1.000,0.902,0.649,1.158,1.658,1.623,1.865
13245,4.507,1.000,0.965,0.990,0.915,0.723,0.435,1.757
13254,4.332,1.000,1.090,1.263,0.918,1.442,1.497,1.855
13425,4.558,1.000,0.961,1.206,1.243,1.093,0.803,1.925
13452,4.132,1.000,0.725,1.076,1.097,2.026,2.130,1.775
13524,3.883,1.000,0.870,0.855,1.616,1.583,1.617,1.660
13542,4.601,1.000,0.834,0.993,1.682,2.254,2.237,1.955
14235,3.972,1.000,0.783,0.714,0.671,0.640,0.346,1.578
14253,4.344,1.000,0.805,0.753,0.683,1.507,1.392,1.816
14325,4.146,1.000,0.774,0.772,0.971,1.008,0.698,1.521
14352,4.537,1.000,0.819,0.901,1.099,2.096,2.057,1.739
14523,5.552,1.000,0.882,0.887,1.976,2.019,2.155,2.182
14532,5.739,1.000,0.926,0.912,2.313,2.712,2.729,2.229
15234,5.702,1.000,0.721,1.490,1.606,1.631,1.415,2.247
15243,5.966,1.000,0.967,1.551,1.592,1.848,2.016,2.352
15324,5.615,1.000,0.793,1.642,2.162,2.120,2.047,2.302
15342,6.282,1.000,1.013,1.895,2.524,2.953,2.834,2.425
15423,6.278,1.000,0.895,1.846,2.133,2.391,2.155,2.570
15432,5.652,1.000,0.740,1.572,2.016,2.643,2.478,2.190
}
\end{figure*}
\newpage

\begin{figure*}[h]
    \centering
    \caption{lbnl-network results normalized by splatt (parallel) (2)}

\pgfplotstabletypeset[
    /pgfplots/colormap={whiteblue}{rgb255(0cm)=(49,130,189); rgb255(.14cm)=(255,255,255); rgb255(1cm)=(222,45,38)},
    colorCell/.style={
        color cells={min=0,max=7}
    },
    itemCell/.style={
        string type,
        column name={},
    },
    col sep=comma,
    columns={permutation,qsort,splatt,1-sadilla,2-sadilla,3-sadilla,4-sadilla,k-sadilla,radix},
    columns/permutation/.style={itemCell},
   columns/qsort/.style={colorCell},
   columns/splatt/.style={colorCell},
   columns/1-sadilla/.style={colorCell},
   columns/2-sadilla/.style={colorCell},
   columns/3-sadilla/.style={colorCell},
   columns/4-sadilla/.style={colorCell},
   columns/k-sadilla/.style={colorCell},
   columns/radix/.style={colorCell},
]{
permutation,qsort,splatt,1-sadilla,2-sadilla,3-sadilla,4-sadilla,k-sadilla,radix
21345,4.856,1.000,1.144,0.846,0.480,0.374,0.050,1.741
21354,3.734,1.000,1.192,0.808,0.433,0.817,0.792,1.241
21435,3.805,1.000,1.024,0.831,0.508,0.502,0.300,1.316
21453,3.396,1.000,0.885,0.747,0.448,1.012,0.999,1.022
21534,3.888,1.000,1.102,0.755,0.908,0.826,0.791,1.262
21543,3.763,1.000,1.224,0.876,0.881,1.158,1.086,1.248
23145,3.289,1.000,1.113,0.403,0.310,0.242,0.072,1.098
23154,3.107,1.000,1.210,0.581,0.349,0.709,0.644,1.106
23415,2.980,1.000,1.013,0.467,0.309,0.255,0.136,0.942
23451,3.087,1.000,1.163,0.512,0.283,0.980,0.959,0.988
23514,3.226,1.000,1.011,0.482,1.056,1.156,1.005,1.114
23541,3.146,1.000,1.033,0.472,1.120,1.019,1.020,1.058
24135,3.196,1.000,1.096,0.210,0.193,0.171,0.074,1.077
24153,3.351,1.000,1.173,0.202,0.176,0.705,0.644,1.090
24315,3.258,1.000,1.072,0.263,0.307,0.222,0.141,1.078
24351,3.298,1.000,1.124,0.237,0.239,1.047,1.089,1.052
24513,3.654,1.000,1.232,0.225,1.180,1.123,1.063,1.185
24531,3.560,1.000,1.170,0.215,1.066,1.104,1.082,1.157
25134,3.182,1.000,0.904,0.901,0.958,1.010,0.924,1.020
25143,3.674,1.000,1.151,1.031,1.102,1.170,1.200,1.164
25314,4.013,1.000,1.292,1.171,1.330,1.302,1.249,1.369
25341,3.197,1.000,0.790,0.995,1.065,1.033,0.961,1.132
25413,3.886,1.000,1.218,1.179,1.321,1.214,1.202,1.320
25431,4.021,1.000,1.087,1.090,1.237,1.153,1.268,1.302
}
\end{figure*}
\newpage

\begin{figure*}[h]
    \centering
    \caption{lbnl-network results normalized by splatt (parallel) (3)}

\pgfplotstabletypeset[
    /pgfplots/colormap={whiteblue}{rgb255(0cm)=(49,130,189); rgb255(.14cm)=(255,255,255); rgb255(1cm)=(222,45,38)},
    colorCell/.style={
        color cells={min=0,max=7}
    },
    itemCell/.style={
        string type,
        column name={},
    },
    col sep=comma,
    columns={permutation,qsort,splatt,1-sadilla,2-sadilla,3-sadilla,4-sadilla,k-sadilla,radix},
    columns/permutation/.style={itemCell},
   columns/qsort/.style={colorCell},
   columns/splatt/.style={colorCell},
   columns/1-sadilla/.style={colorCell},
   columns/2-sadilla/.style={colorCell},
   columns/3-sadilla/.style={colorCell},
   columns/4-sadilla/.style={colorCell},
   columns/k-sadilla/.style={colorCell},
   columns/radix/.style={colorCell},
]{
permutation,qsort,splatt,1-sadilla,2-sadilla,3-sadilla,4-sadilla,k-sadilla,radix
31245,9.724,1.000,0.949,0.827,0.772,0.613,0.101,3.210
31254,6.557,1.000,0.912,0.902,0.574,1.344,1.448,2.049
31425,6.255,1.000,1.028,0.838,0.819,0.679,0.442,2.064
31452,7.261,1.000,1.019,0.879,1.016,1.914,2.022,2.448
31524,6.405,1.000,0.846,0.605,1.375,1.266,1.251,2.205
31542,6.530,1.000,0.956,0.617,1.473,1.636,1.717,2.270
32145,5.888,1.000,0.860,0.576,0.498,0.427,0.133,1.849
32154,6.083,1.000,1.034,0.789,0.593,1.240,1.295,1.978
32415,5.770,1.000,1.016,0.823,0.671,0.485,0.258,1.982
32451,6.622,1.000,1.045,0.749,0.600,1.953,1.989,2.199
32514,5.971,1.000,0.840,0.640,2.156,2.054,1.798,1.933
32541,6.060,1.000,0.974,0.667,1.787,1.995,1.805,2.036
34125,6.241,1.000,1.124,0.633,0.558,0.408,0.136,1.880
34152,6.556,1.000,1.047,0.770,0.637,1.281,1.258,2.087
34215,4.760,1.000,0.815,0.782,0.481,0.413,0.225,1.585
34251,5.216,1.000,0.990,0.835,0.458,1.737,1.583,1.757
34512,6.446,1.000,1.054,0.841,2.135,2.156,1.881,2.099
34521,6.011,1.000,0.826,0.756,2.090,1.789,1.852,2.024
35124,8.430,1.000,1.053,2.386,2.600,2.355,2.604,2.680
35142,9.340,1.000,1.196,3.024,2.782,3.430,2.972,2.993
35214,8.040,1.000,1.039,2.374,2.475,2.477,2.465,2.930
35241,8.498,1.000,1.076,2.682,2.750,2.548,2.679,2.624
35412,8.657,1.000,1.000,2.443,2.766,2.753,2.674,2.857
35421,9.106,1.000,1.037,2.608,3.004,2.784,2.859,3.025
}
\end{figure*}
\newpage

\begin{figure*}[h]
    \centering
    \caption{lbnl-network results normalized by splatt (parallel) (4)}

\pgfplotstabletypeset[
    /pgfplots/colormap={whiteblue}{rgb255(0cm)=(49,130,189); rgb255(.14cm)=(255,255,255); rgb255(1cm)=(222,45,38)},
    colorCell/.style={
        color cells={min=0,max=7}
    },
    itemCell/.style={
        string type,
        column name={},
    },
    col sep=comma,
    columns={permutation,qsort,splatt,1-sadilla,2-sadilla,3-sadilla,4-sadilla,k-sadilla,radix},
    columns/permutation/.style={itemCell},
   columns/qsort/.style={colorCell},
   columns/splatt/.style={colorCell},
   columns/1-sadilla/.style={colorCell},
   columns/2-sadilla/.style={colorCell},
   columns/3-sadilla/.style={colorCell},
   columns/4-sadilla/.style={colorCell},
   columns/k-sadilla/.style={colorCell},
   columns/radix/.style={colorCell},
]{
permutation,qsort,splatt,1-sadilla,2-sadilla,3-sadilla,4-sadilla,k-sadilla,radix
41235,6.613,1.000,1.244,0.618,0.395,0.336,0.077,2.096
41253,6.199,1.000,1.042,0.532,0.325,1.249,1.235,1.975
41325,4.173,1.000,0.992,0.567,0.455,0.398,0.255,1.266
41352,4.570,1.000,0.990,0.624,0.563,1.135,1.175,1.496
41523,5.022,1.000,1.026,0.541,1.065,1.058,1.021,1.709
41532,5.124,1.000,0.996,0.554,1.090,1.237,1.277,1.669
42135,3.667,1.000,1.129,0.210,0.199,0.175,0.080,1.093
42153,3.883,1.000,0.971,0.208,0.171,0.728,0.733,1.209
42315,3.056,1.000,0.844,0.210,0.260,0.195,0.181,0.982
42351,3.985,1.000,0.995,0.269,0.267,1.231,1.213,1.284
42513,3.682,1.000,0.886,0.194,1.196,1.269,1.139,1.167
42531,3.741,1.000,0.818,0.196,1.199,1.174,1.217,1.151
43125,3.284,1.000,1.068,0.558,0.308,0.242,0.071,1.013
43152,2.874,1.000,0.855,0.495,0.300,0.608,0.552,0.917
43215,3.174,1.000,0.980,0.745,0.302,0.283,0.135,1.034
43251,3.879,1.000,1.036,0.810,0.336,1.210,1.229,1.154
43512,3.643,1.000,1.236,0.709,1.068,1.241,1.110,1.180
43521,3.633,1.000,1.240,0.633,1.170,1.072,1.043,1.175
45123,4.466,1.000,0.878,1.226,1.234,1.217,1.186,1.339
45132,4.283,1.000,1.134,1.223,1.220,1.507,1.479,1.362
45213,4.374,1.000,1.183,1.271,1.384,1.407,1.302,1.493
45231,4.724,1.000,1.096,1.300,1.370,1.530,1.391,1.570
45312,4.397,1.000,0.981,1.306,1.322,1.369,1.254,1.319
45321,4.273,1.000,0.886,1.292,1.394,1.279,1.386,1.362
}
\end{figure*}
\newpage

\begin{figure*}[h]
    \centering
    \caption{lbnl-network results normalized by splatt (parallel) (5)}

\pgfplotstabletypeset[
    /pgfplots/colormap={whiteblue}{rgb255(0cm)=(49,130,189); rgb255(.14cm)=(255,255,255); rgb255(1cm)=(222,45,38)},
    colorCell/.style={
        color cells={min=0,max=7}
    },
    itemCell/.style={
        string type,
        column name={},
    },
    col sep=comma,
    columns={permutation,qsort,splatt,1-sadilla,2-sadilla,3-sadilla,4-sadilla,k-sadilla,radix},
    columns/permutation/.style={itemCell},
   columns/qsort/.style={colorCell},
   columns/splatt/.style={colorCell},
   columns/1-sadilla/.style={colorCell},
   columns/2-sadilla/.style={colorCell},
   columns/3-sadilla/.style={colorCell},
   columns/4-sadilla/.style={colorCell},
   columns/k-sadilla/.style={colorCell},
   columns/radix/.style={colorCell},
]{
permutation,qsort,splatt,1-sadilla,2-sadilla,3-sadilla,4-sadilla,k-sadilla,radix
51234,1.879,1.000,0.957,0.531,0.495,0.496,0.572,0.611
51243,1.858,1.000,1.016,0.507,0.500,0.597,0.613,0.581
51324,1.863,1.000,1.004,0.501,0.605,0.607,0.578,0.573
51342,1.857,1.000,1.027,0.491,0.624,0.678,0.662,0.600
51423,1.889,1.000,1.008,0.497,0.586,0.595,0.578,0.588
51432,1.929,1.000,1.039,0.556,0.609,0.692,0.686,0.592
52134,1.952,1.000,1.037,0.568,0.547,0.557,0.531,0.613
52143,1.882,1.000,0.996,0.548,0.552,0.602,0.570,0.570
52314,1.911,1.000,1.051,0.538,0.574,0.570,0.559,0.635
52341,1.851,1.000,1.018,0.539,0.549,0.587,0.564,0.603
52413,1.869,1.000,0.976,0.543,0.548,0.552,0.571,0.576
52431,1.884,1.000,0.976,0.548,0.655,0.575,0.585,0.576
53124,1.915,1.000,1.060,0.554,0.540,0.529,0.528,0.615
53142,1.924,1.000,1.041,0.553,0.551,0.661,0.614,0.628
53214,1.865,1.000,0.994,0.564,0.555,0.583,0.586,0.615
53241,1.876,1.000,1.010,0.539,0.570,0.584,0.555,0.559
53412,1.926,1.000,1.035,0.549,0.558,0.563,0.603,0.610
53421,1.793,1.000,0.993,0.512,0.549,0.541,0.551,0.579
54123,1.941,1.000,1.005,0.567,0.551,0.534,0.566,0.615
54132,1.912,1.000,1.058,0.524,0.539,0.621,0.600,0.638
54213,1.874,1.000,0.983,0.532,0.630,0.561,0.546,0.579
54231,1.872,1.000,0.994,0.551,0.606,0.610,0.587,0.580
54312,1.918,1.000,1.044,0.524,0.552,0.565,0.588,0.587
54321,1.864,1.000,1.028,0.532,0.576,0.604,0.570,0.561
}
\end{figure*}
\newpage

\begin{figure*}[h]
    \centering
    \caption{vast-2015-mc1-5d results normalized by splatt (parallel) (1)}

\pgfplotstabletypeset[
    /pgfplots/colormap={whiteblue}{rgb255(0cm)=(49,130,189); rgb255(.14cm)=(255,255,255); rgb255(1cm)=(222,45,38)},
    colorCell/.style={
        color cells={min=0,max=7}
    },
    itemCell/.style={
        string type,
        column name={},
    },
    col sep=comma,
    columns={permutation,qsort,splatt,1-sadilla,2-sadilla,3-sadilla,4-sadilla,k-sadilla,radix},
    columns/permutation/.style={itemCell},
   columns/qsort/.style={colorCell},
   columns/splatt/.style={colorCell},
   columns/1-sadilla/.style={colorCell},
   columns/2-sadilla/.style={colorCell},
   columns/3-sadilla/.style={colorCell},
   columns/4-sadilla/.style={colorCell},
   columns/k-sadilla/.style={colorCell},
   columns/radix/.style={colorCell},
]{
permutation,qsort,splatt,1-sadilla,2-sadilla,3-sadilla,4-sadilla,k-sadilla,radix
12345,19.992,1.000,0.290,0.267,0.273,0.276,0.000,6.291
12354,19.863,1.000,0.283,0.279,0.272,2.170,1.896,6.238
12435,19.583,1.000,0.279,0.271,2.142,2.169,1.899,6.141
12453,19.431,1.000,0.282,0.263,2.172,4.136,3.746,6.300
12534,19.503,1.000,0.285,0.266,2.162,2.202,1.855,6.245
12543,19.488,1.000,0.282,0.263,2.142,4.041,3.762,6.067
13245,19.226,1.000,0.332,1.414,1.521,1.519,1.224,5.915
13254,19.127,1.000,0.331,1.470,1.541,3.402,3.055,5.764
13425,21.281,1.000,0.438,1.314,2.359,2.433,2.177,4.407
13452,21.230,1.000,0.453,1.273,2.324,3.612,3.344,4.289
13524,22.133,1.000,0.444,1.300,2.418,2.431,2.210,4.376
13542,21.972,1.000,0.445,1.278,2.293,3.574,3.331,4.400
14235,21.988,1.000,0.407,1.421,1.465,1.508,1.233,4.647
14253,23.114,1.000,0.410,1.463,1.483,3.061,2.801,4.731
14325,22.743,1.000,0.417,1.438,2.488,2.567,2.296,4.677
14352,22.654,1.000,0.424,1.407,2.429,3.746,3.568,4.535
14523,22.303,1.000,0.420,1.398,2.710,2.653,2.475,4.617
14532,22.422,1.000,0.421,1.415,2.717,3.688,3.583,4.599
15234,23.328,1.000,0.413,1.453,1.493,1.510,1.244,4.736
15243,23.328,1.000,0.413,1.435,1.487,3.104,2.794,4.896
15324,22.049,1.000,0.416,1.428,2.473,2.490,2.330,4.682
15342,22.869,1.000,0.434,1.478,2.505,3.829,3.551,4.605
15423,23.209,1.000,0.433,1.484,2.844,2.842,2.483,4.680
15432,22.598,1.000,0.420,1.414,2.673,3.769,3.495,4.576
}
\end{figure*}
\newpage

\begin{figure*}[h]
    \centering
    \caption{vast-2015-mc1-5d results normalized by splatt (parallel) (2)}

\pgfplotstabletypeset[
    /pgfplots/colormap={whiteblue}{rgb255(0cm)=(49,130,189); rgb255(.14cm)=(255,255,255); rgb255(1cm)=(222,45,38)},
    colorCell/.style={
        color cells={min=0,max=7}
    },
    itemCell/.style={
        string type,
        column name={},
    },
    col sep=comma,
    columns={permutation,qsort,splatt,1-sadilla,2-sadilla,3-sadilla,4-sadilla,k-sadilla,radix},
    columns/permutation/.style={itemCell},
   columns/qsort/.style={colorCell},
   columns/splatt/.style={colorCell},
   columns/1-sadilla/.style={colorCell},
   columns/2-sadilla/.style={colorCell},
   columns/3-sadilla/.style={colorCell},
   columns/4-sadilla/.style={colorCell},
   columns/k-sadilla/.style={colorCell},
   columns/radix/.style={colorCell},
]{
permutation,qsort,splatt,1-sadilla,2-sadilla,3-sadilla,4-sadilla,k-sadilla,radix
21345,27.031,1.000,1.169,1.224,1.203,1.237,1.007,4.270
21354,27.350,1.000,1.171,1.240,1.211,2.894,2.664,4.201
21435,27.900,1.000,1.202,1.226,2.893,2.944,2.851,4.167
21453,26.407,1.000,1.156,1.220,2.895,4.491,4.299,4.082
21534,26.300,1.000,1.158,1.257,2.951,2.916,2.674,4.095
21543,28.645,1.000,1.208,1.300,2.955,4.714,4.527,4.321
23145,25.021,1.000,1.163,1.637,1.716,1.666,1.605,3.698
23154,24.948,1.000,1.158,1.655,1.665,3.114,2.965,3.726
23415,23.571,1.000,1.150,1.689,1.938,1.966,1.800,3.385
23451,23.795,1.000,1.119,1.693,1.971,2.432,2.239,2.741
23514,22.625,1.000,1.115,1.659,1.919,1.936,1.775,3.215
23541,23.323,1.000,1.135,1.661,1.933,2.454,2.278,2.714
24135,23.249,1.000,1.125,1.520,1.562,1.590,1.370,3.385
24153,24.956,1.000,1.150,1.564,1.581,2.948,2.798,3.649
24315,23.719,1.000,1.152,1.570,2.051,2.067,1.901,3.489
24351,24.660,1.000,1.156,1.657,2.107,2.639,2.397,2.945
24513,23.568,1.000,1.140,1.631,2.084,2.049,1.901,2.907
24531,23.503,1.000,1.140,1.585,2.094,2.515,2.423,2.893
25134,23.424,1.000,1.144,1.554,1.584,1.564,1.370,3.502
25143,24.984,1.000,1.140,1.544,1.566,3.006,2.802,3.492
25314,22.977,1.000,1.146,1.549,2.064,2.080,1.914,3.522
25341,23.517,1.000,1.136,1.648,2.075,2.599,2.368,2.834
25413,24.154,1.000,1.181,1.652,2.151,2.145,1.963,3.024
25431,24.384,1.000,1.134,1.639,2.113,2.623,2.446,2.923
}
\end{figure*}
\newpage

\begin{figure*}[h]
    \centering
    \caption{vast-2015-mc1-5d results normalized by splatt (parallel) (3)}

\pgfplotstabletypeset[
    /pgfplots/colormap={whiteblue}{rgb255(0cm)=(49,130,189); rgb255(.14cm)=(255,255,255); rgb255(1cm)=(222,45,38)},
    colorCell/.style={
        color cells={min=0,max=7}
    },
    itemCell/.style={
        string type,
        column name={},
    },
    col sep=comma,
    columns={permutation,qsort,splatt,1-sadilla,2-sadilla,3-sadilla,4-sadilla,k-sadilla,radix},
    columns/permutation/.style={itemCell},
   columns/qsort/.style={colorCell},
   columns/splatt/.style={colorCell},
   columns/1-sadilla/.style={colorCell},
   columns/2-sadilla/.style={colorCell},
   columns/3-sadilla/.style={colorCell},
   columns/4-sadilla/.style={colorCell},
   columns/k-sadilla/.style={colorCell},
   columns/radix/.style={colorCell},
]{
permutation,qsort,splatt,1-sadilla,2-sadilla,3-sadilla,4-sadilla,k-sadilla,radix
31245,2.972,1.000,1.065,0.146,0.157,0.158,0.117,0.898
31254,3.010,1.000,1.057,0.145,0.155,0.434,0.387,0.889
31425,1.837,1.000,0.993,0.074,0.170,0.171,0.147,0.375
31452,1.770,1.000,0.988,0.072,0.161,0.258,0.245,0.340
31524,1.877,1.000,1.002,0.075,0.174,0.173,0.153,0.373
31542,1.787,1.000,1.002,0.074,0.163,0.266,0.244,0.346
32145,1.351,1.000,0.998,0.093,0.093,0.092,0.082,0.207
32154,1.412,1.000,0.996,0.094,0.093,0.169,0.155,0.205
32415,1.305,1.000,0.992,0.097,0.111,0.115,0.105,0.201
32451,1.361,1.000,0.997,0.098,0.113,0.140,0.130,0.154
32514,1.323,1.000,0.996,0.099,0.115,0.115,0.107,0.193
32541,1.315,1.000,0.996,0.095,0.112,0.135,0.125,0.157
34125,1.415,1.000,0.992,0.067,0.064,0.064,0.052,0.198
34152,1.369,1.000,0.988,0.071,0.060,0.117,0.109,0.174
34215,1.426,1.000,0.995,0.092,0.101,0.102,0.093,0.184
34251,1.426,1.000,0.994,0.089,0.100,0.121,0.114,0.138
34512,1.278,1.000,0.994,0.092,0.077,0.080,0.070,0.172
34521,1.450,1.000,0.998,0.095,0.088,0.119,0.111,0.136
35124,1.406,1.000,0.990,0.063,0.062,0.063,0.052,0.193
35142,1.332,1.000,0.995,0.066,0.062,0.114,0.109,0.175
35214,1.438,1.000,0.997,0.086,0.103,0.104,0.094,0.177
35241,1.449,1.000,0.995,0.082,0.101,0.123,0.113,0.141
35412,1.308,1.000,0.992,0.082,0.079,0.078,0.071,0.171
35421,1.413,1.000,0.997,0.083,0.089,0.119,0.111,0.132
}
\end{figure*}
\newpage

\begin{figure*}[h]
    \centering
    \caption{vast-2015-mc1-5d results normalized by splatt (parallel) (4)}

\pgfplotstabletypeset[
    /pgfplots/colormap={whiteblue}{rgb255(0cm)=(49,130,189); rgb255(.14cm)=(255,255,255); rgb255(1cm)=(222,45,38)},
    colorCell/.style={
        color cells={min=0,max=7}
    },
    itemCell/.style={
        string type,
        column name={},
    },
    col sep=comma,
    columns={permutation,qsort,splatt,1-sadilla,2-sadilla,3-sadilla,4-sadilla,k-sadilla,radix},
    columns/permutation/.style={itemCell},
   columns/qsort/.style={colorCell},
   columns/splatt/.style={colorCell},
   columns/1-sadilla/.style={colorCell},
   columns/2-sadilla/.style={colorCell},
   columns/3-sadilla/.style={colorCell},
   columns/4-sadilla/.style={colorCell},
   columns/k-sadilla/.style={colorCell},
   columns/radix/.style={colorCell},
]{
permutation,qsort,splatt,1-sadilla,2-sadilla,3-sadilla,4-sadilla,k-sadilla,radix
41235,41.458,1.000,1.116,1.055,1.057,1.058,0.789,5.827
41253,41.012,1.000,1.127,1.083,1.078,3.043,2.710,5.975
41325,40.896,1.000,1.138,1.028,2.396,2.537,2.156,5.796
41352,36.312,1.000,1.117,0.947,2.194,3.784,3.464,4.758
41523,36.887,1.000,1.105,0.941,2.585,2.540,2.327,4.851
41532,36.840,1.000,1.119,0.957,2.553,3.760,3.467,4.826
42135,25.266,1.000,1.046,1.357,1.319,1.328,1.190,3.050
42153,24.665,1.000,1.060,1.342,1.391,2.603,2.415,3.125
42315,25.016,1.000,1.077,1.391,1.793,1.785,1.650,3.159
42351,24.312,1.000,1.077,1.481,1.845,2.180,2.037,2.504
42513,25.457,1.000,1.073,1.457,1.848,1.858,1.631,2.642
42531,24.094,1.000,1.052,1.420,1.720,2.211,2.044,2.510
43125,30.434,1.000,1.063,1.502,1.340,1.348,1.147,4.392
43152,28.597,1.000,1.055,1.458,1.289,2.566,2.337,3.736
43215,22.534,1.000,1.055,1.472,1.605,1.634,1.479,2.842
43251,23.441,1.000,1.050,1.467,1.670,2.025,1.839,2.247
43512,21.040,1.000,1.039,1.490,1.260,1.263,1.093,2.782
43521,23.264,1.000,1.068,1.527,1.458,1.910,1.804,2.219
45123,22.306,1.000,1.058,0.934,0.973,0.973,0.825,3.060
45132,22.907,1.000,1.049,0.964,0.986,1.722,1.537,3.042
45213,24.261,1.000,1.072,1.120,1.631,1.631,1.536,2.353
45231,23.404,1.000,1.053,1.115,1.630,2.070,1.922,2.288
45312,21.808,1.000,1.042,0.978,1.319,1.340,1.184,2.911
45321,22.455,1.000,1.057,1.152,1.442,1.936,1.795,2.217
}
\end{figure*}
\newpage

\begin{figure*}[h]
    \centering
    \caption{vast-2015-mc1-5d results normalized by splatt (parallel) (5)}

\pgfplotstabletypeset[
    /pgfplots/colormap={whiteblue}{rgb255(0cm)=(49,130,189); rgb255(.14cm)=(255,255,255); rgb255(1cm)=(222,45,38)},
    colorCell/.style={
        color cells={min=0,max=7}
    },
    itemCell/.style={
        string type,
        column name={},
    },
    col sep=comma,
    columns={permutation,qsort,splatt,1-sadilla,2-sadilla,3-sadilla,4-sadilla,k-sadilla,radix},
    columns/permutation/.style={itemCell},
   columns/qsort/.style={colorCell},
   columns/splatt/.style={colorCell},
   columns/1-sadilla/.style={colorCell},
   columns/2-sadilla/.style={colorCell},
   columns/3-sadilla/.style={colorCell},
   columns/4-sadilla/.style={colorCell},
   columns/k-sadilla/.style={colorCell},
   columns/radix/.style={colorCell},
]{
permutation,qsort,splatt,1-sadilla,2-sadilla,3-sadilla,4-sadilla,k-sadilla,radix
51234,42.158,1.000,1.169,1.068,1.074,1.074,0.782,5.946
51243,41.325,1.000,1.146,1.077,1.072,3.004,2.649,6.000
51324,41.203,1.000,1.131,1.050,2.411,2.413,2.119,5.872
51342,37.043,1.000,1.097,0.998,2.193,3.803,3.569,4.843
51423,37.649,1.000,1.089,0.993,2.587,2.633,2.308,4.843
51432,38.030,1.000,1.139,1.014,2.650,4.004,3.590,4.968
52134,25.568,1.000,1.074,1.326,1.352,1.351,1.197,3.063
52143,24.291,1.000,1.058,1.324,1.336,2.521,2.473,3.013
52314,24.479,1.000,1.051,1.315,1.751,1.723,1.637,3.005
52341,24.320,1.000,1.056,1.430,1.803,2.204,2.085,2.495
52413,25.349,1.000,1.080,1.420,1.801,1.805,1.652,2.577
52431,24.514,1.000,1.059,1.417,1.787,2.208,2.093,2.484
53124,30.369,1.000,1.046,1.463,1.337,1.383,1.125,4.284
53142,28.808,1.000,1.058,1.463,1.308,2.606,2.356,3.735
53214,23.357,1.000,1.053,1.488,1.622,1.645,1.507,2.831
53241,23.362,1.000,1.038,1.452,1.627,1.968,1.867,2.259
53412,21.750,1.000,1.063,1.550,1.327,1.305,1.147,2.783
53421,22.620,1.000,1.050,1.464,1.376,1.837,1.699,2.103
54123,22.996,1.000,1.057,0.954,0.992,0.977,0.831,3.065
54132,22.550,1.000,1.055,0.954,0.976,1.695,1.539,2.974
54213,23.990,1.000,1.062,1.094,1.650,1.607,1.474,2.289
54231,23.141,1.000,1.063,1.129,1.655,2.018,1.859,2.296
54312,22.171,1.000,1.047,0.962,1.336,1.355,1.165,2.943
54321,23.055,1.000,1.048,1.125,1.488,1.974,1.838,2.216
}
\end{figure*}
\clearpage
\newpage

\end{document}